%% file: main.tex
\documentclass[sigconf]{aamas} 
\usepackage{balance} 

\renewcommand\footnotetextcopyrightpermission[1]{} 
\setcopyright{ifaamas}
\acmConference[AAMAS '26]{Proc.\@ of the 25th International Conference
on Autonomous Agents and Multiagent Systems (AAMAS 2026)}{May 25 -- 29, 2026}
{Paphos, Cyprus}{C.~Amato, L.~Dennis, V.~Mascardi, J.~Thangarajah (eds.)}
\copyrightyear{2026}
\acmYear{2026}
\acmDOI{}
\acmPrice{}
\acmISBN{}

\acmSubmissionID{1680}

\title{Learning Robust Markov Models for Safe Runtime Monitoring}

\author{Antonina Skurka}
\affiliation{
  \institution{Chalmers University of Technology \\and University of Gothenburg}
  \city{Gothenburg}
  \country{Sweden}}
\email{skurka@chalmers.se}

\author{Luko van der Maas}
\affiliation{\institution{Radboud University}\city{Nijmegen}\country{the Netherlands}}
\email{luko.vandermaas@ru.nl}

\author{Sebastian Junges}
\affiliation{\institution{Radboud University}\city{Nijmegen}\country{the Netherlands}}
\email{sebastian.junges@ru.nl}

\author{Hazem Torfah}
\affiliation{
  \institution{Chalmers University of Technology \\and University of Gothenburg}
  \city{Gothenburg}
  \country{Sweden}}
\email{hazemto@chalmers.se}

\input{abstract}

\usepackage{tikz}
\usetikzlibrary{intersections,positioning,shapes,positioning,calc,arrows.meta,shapes.geometric,angles,automata}
\usepackage{varwidth}
\usepackage{amsmath}
\usepackage{amsthm}

\usepackage{amssymb}
\usepackage{wasysym}
\usepackage{stackengine}
\usepackage{mdframed}
\usepackage{cleveref}
\usepackage{booktabs}
\usepackage{nicefrac}
\usepackage{xspace}
\usepackage{multirow}
\usepackage{graphicx}
\usepackage{subcaption} 
\usepackage{dblfloatfix}
\usepackage{caption}
\usepackage{subcaption}
\usepackage{wrapfig}
\usepackage{enumitem}

\DeclareMathOperator*{\argmin}{arg\,min}

\newtheorem{definition}{Definition}
\newtheorem{theorem}{Theorem}
\newtheorem{lemma}{Lemma}
\newtheorem{corollary}{Corollary}

\newcommand{\mc}{{\textsf H}}
\newcommand{\imc}{\textsf{iH}}

\newcommand{\nn}{\mathbb{N}}
\newcommand{\mon}{M}
\newcommand{\suo}{\text{SuO}\xspace}

\newcommand{\variable}[1]{{V_{\text{#1}}}}
\newcommand{\vsys}{\variable{sys}}
\newcommand{\vobs}{\variable{obs}}

\newcommand{\letters}[1][]{\Sigma_{\text{#1}}}
\newcommand{\sobs}{\letters[obs]}
\newcommand{\ssys}{\letters[sys]}

\newcommand{\tsys}[1][]{\pi#1}

\newcommand{\traces}[1]{{#1^*}}
\newcommand{\tsobs}{\traces{\sobs}}
\newcommand{\tssys}{\traces{\ssys}}

\crefname{figure}{Fig.}{Figs.}

\begin{document}


\pagestyle{fancy}
\fancyhead{}

\maketitle

\input{intro}
\section{Problem statement}
\label{sec:problem}

\input{problem}

\section{Monitors for \MakeLowercase{i}HMMs}
\label{sec:inferenceihmm}
\input{monitoring}
\section{Learning Ideal Monitor Models}
\label{sec:learning}

\input{learning}

\input{experimental_evaluation}

\input{related}

\section{Conclusion}
\input{conclusion}

\section*{Acknowledgments}
\input{acknowledgements}
\balance

\bibliographystyle{ACM-Reference-Format}
\bibliography{references}  

\pagebreak

\clearpage

\appendix


\section{Extended Discussion on Refinement Framework}
\input{appendix_refinement}


\section{Additional experimental results}

\input{appendix}

\end{document}

%% file: abstract.tex
\begin{abstract}

We present a model-based approach to learning robust runtime monitors for autonomous systems. Runtime monitors play a crucial role in raising the level of assurance by observing system behavior and predicting potential safety violations. In our approach, we propose to capture a system's (stochastic) behavior using interval Hidden Markov Models (iHMMs). The monitor then uses this learned iHMM to derive risk estimates for potential safety violations. 
The paper makes three key contributions: (1) it provides a formalization of the problem of learning robust runtime monitors, (2) introduces a novel framework that uses conformance-testing-based refinement for learning robust iHMMs with convergence guarantees, 
 and (3) presents an efficient monitoring algorithm for computing risk estimates over iHMMs. 
Our empirical results demonstrate the efficacy of monitors learned using our approach, particularly when compared to model-free monitoring approaches that rely solely on collected data without access to a system model.
\end{abstract}

%% file: intro.tex
\section{Introduction}
\label{sec:introduction}

Runtime monitoring is vital for raising the assurance in cyber-physical systems (CPS) \cite{Bartocci2018,modelplex}. Runtime monitors (RMs) continuously observe the system for possibly unsafe situations. When such situations occur, RMs raise an alarm and may consequently trigger backup fail-safe plans to keep the system safe. In modern autonomous CPS, RMs become even more essential to ensure safety, in particular when integrating machine learning (ML) components~\cite{DBLP:journals/corr/abs-2405-06624}. 
In these settings, RMs must deal with various kinds of uncertainty, both in the decisions made by ML models and in the environments in which these systems operate. Often, system dynamics are probabilistic, and the environment state cannot be fully observed.  RMs must take these aspects into account when predicting safety violations, while balancing false positives (too many alarms) and false negatives (dangerous situations go unnoticed). The type of information about the uncertainty in the system available to a monitor has a great impact on achieving this balance. 








Consider an example of an autonomous plane approaching a runway with ground traffic. The plane is equipped with camera-based perception sensors. A key safety-critical routine is to observe the behavior of ground traffic and decide that either it is safe to land or that the landing shall be aborted. A safety specification expresses a minimum distance between plane and ground traffic at all times. There will always be uncertainty about the behavior of vehicles, but also in the perceived sensor data, due to induced noise. A RM must (implicitly or explicitly) estimate the risk level: the probability that the specification will be violated during the landing, and if so, by how much. The RM must then raise an alarm if this risk level exceeds some threshold. The accuracy of the monitor will highly depend on what information one can provide about the (stochastic) behavior of the on-ground vehicles and sensors. 

In this paper, we study the problem of constructing monitors that robustly estimate the risk of violating a safety specification under uncertainty.  
We specifically introduce a novel model-based approach for learning \emph{robust} RMs. In our approach, we propose  a data-driven method for learning robust uncertainty models in the form of \emph{interval Hidden Markov Models (iHMMs)}, a variant of hidden Markov models where transition probabilities are represented as intervals rather than fixed values. On top of the learned iHMMs, we develop an efficient monitoring algorithm for cautiously (or conservatively) estimating risk levels.
Our approach begins with a refinement-based learning procedure that gradually constructs an iHMM from simulation runs collected from the system under observation. The resulting model is then integrated into a monitoring framework capable of estimating risks during system execution.  The general flow of learning iHMMs and integrating them into a monitoring framework is depicted in \Cref{fig:flow}.

Our approach fills a gap in the current state of the art in learning RMs. The state of the art can be broadly categorized into two types of approaches.
The first assumes a \emph{known} model that captures the uncertainty of the system and environment using a predefined model. Typically, the models are hidden Markov models (HMM) or  variants~\cite{DBLP:conf/cav/JungesTS20,babaee2018predictive,camilli2021}. A key challenge in this setting is acquiring sufficiently accurate models.
The second category of approaches, the \emph{model-free methods},  avoids this challenge. Such methods learn classifiers directly from data, without relying on an explicit system model~\cite{cairoli2021neural,DBLP:conf/iccps/ZhaoHFDL24,torfah-rv23}. While promising, model-free approaches typically suffer from high sample complexity and often tend to have low accuracy, as they learn a monitor purely from observable data and lack deeper structural knowledge of the system. 

In our work, we advocate for \emph{learning models} and propose a method to overcome the challenge of acquiring uncertainty models that prevents the application of model-based approaches for RMs. 
Specifically, we present a data-driven strategy to learn an iHMM that serves as a feasible abstraction of the system. We then construct a RM that reasons (logically) about this learned abstraction. While the monitor still operates on sequences of observations, its decision-making is now grounded in the learned iHMM, effectively combining data-driven learning with model-based reasoning.
Our interest in model-based learning approaches is inspired by their success in other areas. One may consider monitoring as a kind of sequential decision making under uncertainty (SDM), where the model-based approaches are dedicated (classical) planners and the model-free approaches are dedicated data-driven variations of reinforcement learning~\cite{DBLP:journals/ftml/MoerlandBPJ23}.
It has been shown that model-based approaches have great impact in increasing safety \cite{10.5555/3294771.3294858}. Furthermore,  in SDM, model-based learning approaches are known to be more sample efficient. 
Indeed, in our experiments, we show that given the same amount of data, our models can reach high safety and accuracy, unmatched by the model-free approaches. 

        
        

\begin{figure}
    \centering
    \input{flow.tex}
    \caption{Learning and integrating robust monitors in CPS}
\label{fig:flow}
\Description{}{}
\end{figure}
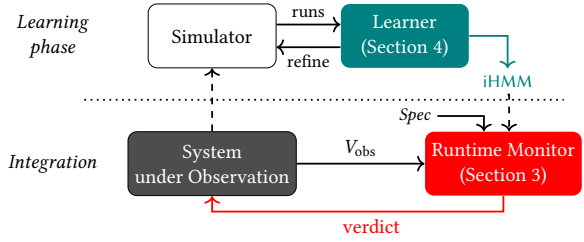

The iHMMs 
allow us to explicitly model uncertainty about the transition dynamics.
We present a refinement-based method for learning iHMMs. Here, the model is updated iteratively by linearly adapting the probability intervals~\cite{suilen2022robust}. Our refinement approach incorporates a feedback loop driven by conformance testing. This loop assesses the uncertainty in monitoring decisions over the so-far learned models and guides the learning process by refining the model in areas where it underperforms. 
We particularly prove that the refinement converges to an optimal model with increasing numbers of refinements. 
Lastly, to enable efficient monitoring, we further introduce a tractable monitoring approach over iHMMs, inspired by~\cite{DBLP:conf/cav/JungesTS20}. Our monitoring procedure runs in polynomial time with respect to the length of observation sequences. 
We evaluate our approach on a set of benchmarks from the area of autonomy. They  particularly
confirm the advantages of the model-based RM setting against model-free approaches. 

%% file: flow.tex
    \scalebox{0.85}{
\begin{tikzpicture}

    \draw[dotted,thick] (0,0) -- (7.5,0);

    \node[rotate=0, align=center](phase1) at (-0.5,1){\it Learning\\ \it phase};

    \node[rotate=0](phase1) at (-0.5,-1){\it Integration};

    \node[draw, rounded corners, minimum height=1cm, fill=black!70, text=white, align=center] (sys) at (2,-1){System \\under Observation};

     \node[ rounded corners, minimum height=1cm, minimum width = 2cm, right = 2cm of sys, align= center, fill = red, text = white] (monitor){ Runtime Monitor \\(\Cref{sec:inferenceihmm})};

     \node[above left = 0.0 and -0.2 of monitor](spec){\small \it Spec};

     \path[draw,thick,->] (spec.east) -| (monitor.121);

     \path[draw, thick, ->](sys.east)  -- node[above]{$\vobs$}(monitor.west);

     \path[draw, ->, thick, red](monitor.south) |- node[below left= 0 and 1.5cm,red]{verdict}(3,-1.75) -| (sys.south);

     \node[draw, rounded corners, minimum height=1cm, minimum width =2cm, align=center] (sim) at (2,1){Simulator};

     \path[draw, dashed, thick, ->](sys.north)  -- node[above]{}(sim.south);

     \node[rounded corners, minimum height=1cm, minimum width =2cm, align=center, fill=teal, text=white, right = 1cm of sim] (learn){Learner\\ (\Cref{sec:learning})};

     \path[draw,->,thick] (sim.10) -- node[above]{\small runs}(learn.170);

     \path[draw,->,thick] (learn.190) -- node[below]{\small refine}(sim.350);

     \node[below right= 0 and 0.1 of learn, teal](model){\small \sf iHMM};

     \path[draw,thick,->,teal] (learn.east) -| (model.north);

     \path[draw,dashed,->,thick](model.south) -| (monitor.80);
    
\end{tikzpicture}
}

%% file: problem.tex
Our goal is to construct a  monitor that raises an alarm when a safety violation is imminent. We operate in a setting where both the system dynamics and the observations made by the monitor are subject to uncertainty. 

We first give key notation and definitions.  
For a set $X$, we use $X^*$ to define the set of finite sequences over $X$, and we use $X^{\leq h}$ to denote the set of sequences up to length $h$. 
In a state-based system, 
the set of states $\ssys$ is defined in terms of valuations over a set of system variables $\vsys$.  The monitor does not necessarily observe valuations over $\vsys$, but rather over $\vobs$ - the set of observation variables (see \Cref{fig:flow} - Integration). In the remainder, we distinguish between system and observed executions, referring to sequences in $\tssys$ as \emph{paths} and sequences in $\tsobs$ as \emph{traces}. For a countable set $X$, let $\textit{Dist}(X) \subset (X \rightarrow [0,1])$ define the set of distributions over~$X$. A distribution $d\in \textit{Dist}(X)$ particularly satisfies, $\sum_{x\in X} d(x) = 1$.

\begin{definition} [System under Observation (\suo)]
\label{def:suo}
A discrete-time system is a tuple $\mathcal{S} = (\vsys,\vobs, d, \textit{obs})$, 
where $\vsys$ and $\vobs$ are sets of system and observation variables, the dynamics $d$ is a distribution over  $\tssys$, and $\mathit{obs}\colon \tssys \rightarrow \mathit{Dist}(\tsobs)$ is a mapping between paths and distributions over observation traces.
\end{definition}

\noindent A system under observation implicitly captures the interaction between a system and its environment. In the above definition, $d$ defines the probability of a path occurring, and thus captures the uncertainty in environment behavior. The observation mapping $\mathit{obs}$ captures the (noisy) behavior of the sensors, which for a given path might output different traces, depending on its stochastic behavior. 

\begin{definition}[Specification] We define safety constraints as finite linear-time properties \cite{princmc,ltlf}. A specification $\varphi$ over $\ssys$ is a set $\varphi \subseteq \tssys$. If $\pi \in \varphi$, we say $\pi$ \emph{satisfies} $\varphi$,  denoted also by $\pi \models \varphi$.
\end{definition}
\noindent Based on the latter definition of safety, we can define the probability of violating a safety specification. We call this a risk function~\cite{DBLP:conf/cav/JungesTS20}\footnote{Risk is usually defined as the probability of violation times the impact. In fact, the definition can be easily refined to also account for differences in impact.}.
\begin{definition}[Risk Function] Given $\mathcal{S} = (\vsys,\vobs, d, \mathit{obs})$ and a specification $\varphi$ over $\ssys$, a  risk function $r^h_\varphi\colon \tssys \rightarrow [0,1]$, for a horizon $h\in \mathbb N^+$, gives a probability of violating $\varphi$ within a horizon of $h$ steps for  a given initial path $\pi \in \tssys$. We set $r^h_\varphi(\tsys) = 0$ if $d(\pi) = 0$ and otherwise
\begin{align*}
r^h_\varphi(\tsys) \!= \!Pr(\{ \pi' \!\in \ssys^{\leq h} \mid \pi \!\cdot\! \pi' \!\not \models \varphi \}) \!=
\!\frac{1}{ d(\pi)} \!\!{\sum \limits_{\pi' \in \ssys^{\leq h}} \!\!\!(\pi \!\cdot \!\pi' \not \models \varphi) \!\cdot \! d(\pi \!\cdot \!\pi')}    
\end{align*}

\end{definition}
\noindent The expression on the right is derived by applying Bayes' law on the conditional probability $d(\pi' \mid \pi)$. The specification and risk function are defined on the system level and it is in general not possible to lift them to the observation level without losing preciseness. 
In terms of the landing scenario from the introduction, we can define the risk of a (near-)collision based on the trajectories of the plane and the ground vehicles, however, a monitor does not know the true locations of the ground vehicles, but only their perceived locations. That is, the challenge in monitoring now is that we cannot observe the system path directly, and thus it is impossible to directly compute the associated risk. Next, we address this challenge by relating the system state space with the observations of the monitor. 


A monitor for a \suo is simply a function $\mon\colon \tsobs \to [0,1]$, assigning to each trace in a system a value between 0 and 1, giving a risk of an observation trace. 
To relate monitors to the system-level specification, we introduce the notion of an ideal monitor.
\begin{definition}[Ideal Monitor]
\label{def:ideal-mon}
Given $\mathcal{S} = (\vsys,\vobs, d, \textit{obs})$ and a specification $\varphi$, for a horizon $h$, and a loss function $L^{\varphi, h,\mathcal{S}}$, the ideal monitor $\mon ^{*} \colon \tsobs \to [0,1]$ is one s.t. 
 \[\mon ^{*} \in \argmin_{\mon '\colon\tsobs\to [0,1]} L^{\varphi, h,\mathcal{S}}(\mon').\]
\end{definition}
\noindent The notion of the ideal monitor $\mon^{*}$, and consequently also the solution to the problem statement below, are dependent on a specific instantiation of the loss function. In this work, we consider the mean square error loss function. For some monitor $\mon \in \mathcal{M}$: 
\begin{align*}
&  L_{\sf{MSE}}^{\varphi, h,\mathcal{S} }(\mon) =  
\sum_{\tau \in \tsobs}\!\!\!Pr(\tau)\!\!\sum_{\pi \in \tssys}\!\!Pr(\pi \mid \tau) \big(\mon (\tau) - r^{h}_{\varphi}(\pi)\big)^2
\end{align*}
\noindent with
$
 Pr(\tau) \!=\! \sum_{\pi \in \tssys} \!\! d(\pi) \cdot obs(\pi)(\tau)$ and $
 Pr(\pi \!\mid \! \tau)\! = \!\frac{d(\pi) \cdot obs(\pi)(\tau)}{Pr(\tau)}
$.

Using the notion of an ideal monitor, we also define the set of \emph{cautious monitors}, i.e., monitors which overestimate the risk compared to the ideal monitor. 
\begin{definition}[Cautious Monitors]
    Given \suo $\mathcal{S}$ and ideal monitor $\mon^{*}$, the monitor $\mon$ is cautious, if $\mon^{*}(\tau) \leq \mon(\tau)$ for all $\tau\in\tsobs$. We call the set of cautious monitors $\mathcal{M}^C$.
\end{definition}

\noindent {\it Remark.} The ideal monitor may still exhibit wrong risk assessments, due to the mismatch between the system and observation world, and we only ask  a cautious monitor to have risk assessments that are at least  as conservative as the ideal monitor. 

Ideal monitors may not be concisely representable or match other constraints that are important for monitors to be deployable to a system~\cite{modd}. We assume there exists a class of admissible monitors that implement a function $\tsobs\to [0,1]$ and a distance function $\delta$ comparing two monitors.
The goal of this paper is to algorithmically find an admissible cautious monitor that is closest to an ideal monitor:


\vspace{1em}
\begin{mdframed}
\noindent \textbf{Problem Statement } Given 
an ideal monitor $\mon^{*}$, 
a set of admissible monitors $\mathcal{M}$, a set of cautious monitors $\mathcal{M}^C$, and a distance function $\delta \colon (\tsobs\to [0,1])^2  \rightarrow \mathbb{R} $,
find a monitor $\mon \in \mathcal{M} \cap \mathcal{M}^C$, s.t.: 
\[\mon \in \argmin_{\mon' \in \mathcal{M} \cap \mathcal{M}^C} \delta (\mon^{*}, \mon')\]
\end{mdframed}

The distance function $\delta$ can be instantiated in different ways, in this paper, we use the average distance over all traces.


Key to solving the above problem is the access to the ideal monitor $\mon ^{*}$. This, in turn, depends heavily on the ability to translate a system-level risk function to the observation level, i.e., having an accurate representation of the function $\mathit{obs}$ (see \Cref{def:ideal-mon}).

%% file: monitoring.tex
In this section, we study SuOs that are described by a (known) hidden Markov model (HMM). We provide the necessary background, define (ideal) monitors on HMMs and then discuss how to account for uncertainty about the transition probabilities in the HMM by defining monitors based on so-called \emph{interval} HMMs. 


\subsection{(Interval) Hidden Markov Models}
\begin{definition}[Hidden Markov Model (HMM)]\label{def:hmm} 
An HMM is a tuple $\mc = (S, \iota, P, Z, \mathcal{O})$, 
where $S$ is a finite set of states, $\iota\colon S \rightarrow [0,1]$ is an initial distribution, $P\colon S \times S \rightarrow [0,1]$ is a transition probability function, and $\mathcal{O} \colon S \rightarrow Distr(Z)$ is a observation function, where Z is a set of observations. 
The transition probability function is such that for every $s\in S$, it holds that $\sum_{s'\in S} P(s,s') = 1$.
\end{definition} 

\noindent A path in a HMM $\mc$ is a finite 
sequence $\pi_{\mc} = s_{0} s_{1}\dots s_{n}$, where for all $0 \leq i < n$, $P(s_{i},s_{i+1})>0$. Let $\Pi_{\mc}$ denote the set of paths of $\mc$. Let $\Pi_{\mc}^{n}$ denote the set of paths of $\mc$ of length $n \in \nn$. An observation trace of $\mc$ associated with some path $\pi_{\mc}$ is a sequence $\tau_{\mc} = z_{0}z_{1} \dots z_{n}$, where for all $0\leq i\leq n$, $\mathcal{O}(s_{i})(z_{i})>0$. Let $T_{\mc}$ denote the set of observation traces. Let $T_{\mc}^{n}$ denote the set of observation traces of length $n \in \nn$. Note that there exists a mapping $obs_{\mc}\colon \Pi_{\mc} \rightarrow Distr(T_{\mc})$  that assigns to each path a distribution over observation traces, induced by the labeling function $\mathcal{O}$. 
A system under observation is modeled by an HMM $\mc$, if the set of states of $\mc$ is defined over $V_{sys}$, the observation function of $\mc$ maps to the distribution over $V_{obs}$, and the mapping $obs_{\mc}$ is equivalent to $obs$ for the system. Consequently, the sets of paths and traces between the system and the model coincide. Throughout the remainder of this work, we use $\Pi$ and $T$ as the sets of paths and traces of both the system and $\mc$ that models it.

\begin{definition}[Interval Hidden Markov Model (\MakeLowercase{i}HMM)]
\label{def:ihmm} 
An Interval Hidden Markov Model is a tuple $\imc = (S, \iota_{l},\iota_{u}, P_{l},P_{u}, Z, \mathcal{O})$, where $S,Z$, and $\mathcal{O}$ are like in $\mc$, $\iota_{l},\iota_{u}:S \rightarrow [0,1]$, with $\iota_{l}(s) \leq \iota_{u}(s)$ for all $s\in S$, are lower and upper bounds on the initial distribution and $ P_{l},P_{u}\colon S \times S \rightarrow [0,1]$, with $P_{l}(s,s') \leq P_{u}(s,s')$ for all $s,s'\in S$, are lower and upper bounds on transition probabilities. We lift notions like paths and traces from HMMs to iHMMs. 
\end{definition}

An iHMM $\imc$ is \emph{refined} by an HMM $\mc$ if the have the same set of states and  set of observations and  $\iota_{l}(s) \leq \iota(s) \leq \iota_{u}(s)$ for all and $s \in S$ and $P_{l}(s, t) \leq P(s, t) \leq P_{u}(s, t)$ for all pairs of states $s, t \in S$. Let $\mathcal{H}_\imc$ be the set of all HMMs that refine a given iHMM $\imc$.

\subsection{Monitors for HMMs}
\label{sec:HMMmonitoring}
A monitor over an HMM $\mc$ is a function $\mon_{\mc}\colon T_\mc \to [0,1]$. Forward filtering~\cite{DBLP:conf/rv/StollerBSGHSZ11,DBLP:conf/rv/WilcoxW10} is an implementation of such monitor, which computes the expected value for a risk function $r^h_\varphi$, conditioned on a trace: 
$\mathrm{R}_{\mc}(\tau) = \mathbb{E}_\pi\left[r^h_\varphi(\pi)\mid\tau\right].$
If we choose our loss function to be the MSE,
then the expected value is the optimal estimator for our chosen loss function \cite{pishro2014introduction}, and  thus the optimal monitor for $\mc$. 

When $\mc$ models a \suo, $\mathrm{R}_{\mc}$ corresponds to the ideal monitor for $L_{MSE}^{\varphi,h,\mc}$. We denote this monitor by $\mon^{*}_\mc$.


\begin{lemma} \label{lem:idealmonitorhmm}
    Given an HMM $\mc$, modeling a \suo $\mathcal S$, a specification $\varphi$, and horizon $h$, the ideal monitor for $L^{\varphi,h,\mc}$ is $\mon^{*}_\mc$.
\end{lemma}


\subsection{Monitors for iHMMs}

\label{sec:iHMMmonitoring}
Monitoring iHMMs follows the same general principles as monitoring HMMs, but with an additional challenge: an iHMM $\imc$ corresponds to a set of possible HMMs, $\mathcal{H}_\imc$. When computing the risk of a trace, we must therefore decide which $\mc \in \mathcal{H}_\imc$ to use. To remain cautious, we  overestimate the risk of a trace. 

\begin{definition}[iHMM Monitor] \label{def:ihmmmonitor}
    Given an iHMM $\imc$, specification $\varphi$, horizon $h$, observation trace $\tau$, the iHMM monitor is defined as  $\mon_\imc(\tau)=\max_{\mc \in \mathcal{H}}\mathrm{R}_\mc(\tau)$.
\end{definition}

\paragraph{Relation to cautious monitors}
A safety specification $\varphi$, horizon $h$, and a \suo $\mathcal{S}$, which is modeled by a HMM $\mc$, induce a risk function $r_\varphi^h$ on states, an ideal HMM monitor $\mon^{*}_{\mc}$, and a set of cautious monitors $\mathcal{M}^C$. 
Any iHMM monitor $\mon_\imc$, where  $\imc$ is refined by our system model $\mc$, is a cautious monitor.

\begin{lemma} \label{lem:safeihmmmonitor}
    Given a SuO which is modeled by a HMM $\mc$, any iHMM $\imc$ that is refined by $\mc$ induces a cautious monitor $\mon_{\imc} \in \mathcal{M}^C$.
\end{lemma}
\begin{proof}
    We know that $\mc\in \mathcal{H}_\imc$. Thus, for every trace $\tau\in T_\mc$,
    \begin{align*}
        \mon_\imc(\tau) \!=\! \max_{\mc' \in \mathcal{H}} \!\!\sum_{\pi\in\Pi^{|\tau|}_{\mc'}}\!\!\!Pr_{\mc'}(\pi\!\mid\!\tau)\!\cdot\! r^{h}_{\varphi}(\pi) \geq \!\!\sum_{\pi\in\Pi^{|\tau|}_\mc}\!\!Pr_\mc(\pi\mid\tau)\cdot r^{h}_{\varphi}(\pi) \!=\! \mon^{*}_\mc
    \end{align*}
    satisfying the condition for $\mon_\imc\in\mathcal{M}^C$.
\end{proof}

We show now that the risk can be computed in polynomial time by model checking a maximum bounded reachability property on the iHMM. Using this, the iHMM monitoring verdicts are computable in polynomial time using the unrolling algorithm from \cite{DBLP:conf/cav/JungesTS20}.
\begin{theorem}\label{thm:ptime}
    The decision problem $\mon_\imc(\tau) > \lambda$ for a threshold $\lambda \geq 0$ is computable in polynomial time. 
\end{theorem}
The proof of \Cref{thm:ptime} combines the work from \cite{DBLP:conf/cav/JungesTS20,DBLP:conf/tacas/BaierKKM14,DBLP:journals/mor/Iyengar05,DBLP:conf/tacas/SenVA06}. We  illustrate the polynomial time decision procedure in the remainder of this section on the small iHMM, shown in \Cref{fig:ihmmex}.
The iHMM has three states: the initial state $s_0$ with exactly observation {\color{blue}blue}, and states $s_1$, $s_2$ with exactly observation {\color{orange}orange}. The probabilistic observations of the HMM from \Cref{def:hmm} can be transformed in polynomial time to exact observations as used here~\cite{DBLP:conf/aaai/ChatterjeeCGK15}. Any point intervals are written as single probabilities. 
Let the specification be $\varphi=\lnot s_2$, and the horizon $h=1$. The maximizing risk function then becomes $r=[s_0\mapsto0.3, s_1\mapsto0, s_2\mapsto1]$. Suppose we observe the trace $\tau=({\color{blue}blue})({\color{orange}orange})$. To compute $\mon_\imc(\tau)$, we perform the transformation as shown in \Cref{fig:ihmmtrans}. We unroll the model over the length of the trace, $|\tau|=2$. This gives us the states $s_i^j$ with $j$ being the unrolled depth of the state. At the horizon we add transitions towards two new states $\top$, and $\bot$, with $P(s_i^h, \top) = r(s_i)$, $P(s_i^h, \bot) = 1-r(s_i)$.
In this transformed model, monitoring a trace reduces to a conditional reachability problem,
\[\mon_\imc(({\color{blue}blue})({\color{orange}orange}))=Pr_\imc^{\max}(\mathit{Reach}(\top)\mid ({\color{blue}blue})({\color{orange}orange})).\]

Now, we apply the standard transformation from iHMMs to MDPs \cite{DBLP:journals/mor/Iyengar05}. There is a one-to-one correspondence between policies in this MDP and refined HMMs of the iHMM. \cite[Proposition 2]{DBLP:conf/tacas/SenVA06} showed that the probability measures are equal given that we allow randomized policies in the MDP. Thus, we now have an MDP $\mathsf{mdp}$ with randomized policies for which $Pr_\imc^{\max}(\mathit{Reach}(\top)\mid ({\color{blue}blue})({\color{orange}orange})) = Pr_\mathsf{mdp}^{\max}(\mathit{Reach}(\top)\mid ({\color{blue}blue})({\color{orange}orange}))$. Now \cite{DBLP:conf/tacas/BaierKKM14} has shown that the policy which induces the maximum conditional probability in an MDP is deterministic, and thus we can apply their translation to get a new MDP where $Pr_\mathsf{mdp}^{\max}(\mathit{Reach}(\top)\mid ({\color{blue}blue})({\color{orange}orange}))$ can be computed using non-conditional reachability probability model checking. Thus allowing us to compute $\mon_\imc^{\max}(({\color{blue}blue})({\color{orange}orange}))$. 

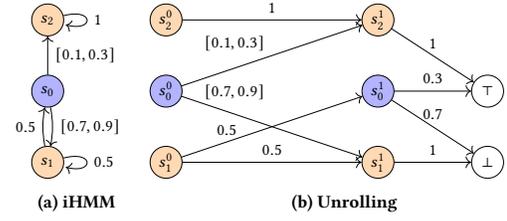
\begin{figure}
\scalebox{0.85}{
    \centering
    \begin{subfigure}{0.25\linewidth}
        \begin{tikzpicture}[font=\footnotesize]
            \tikzset{state/.append style={minimum size=5mm, inner sep=0}}

            \node[state, fill=blue!30!white]   (s0) {$s_0$};
            \node[state, fill=orange!30!white] (s1) [below=0.6 of s0] {$s_1$};
            \node[state, fill=orange!30!white] (s2) [above=0.6 of s0] {$s_2$};
            
            \path[->] (s0) edge[bend left=10] node[right] {$[0.7,0.9]$} (s1)
                      (s1) edge[bend left=10] node[left]  {$0.5$} (s0)
                      (s1) edge[loop right]   node[right] {$0.5$} (s1)
                      (s0) edge               node[right] {$[0.1,0.3]$} (s2)
                      (s2) edge[loop right]   node[right] {$1$} (s2);
        \end{tikzpicture}
        \caption{iHMM}
        \label{fig:ihmmex}
    \end{subfigure}
    \hfill
    \begin{subfigure}{0.72\linewidth}
        \begin{tikzpicture}[font=\footnotesize]
            \tikzset{state/.append style={minimum size=5mm, inner sep=0}}
            
            \node[state, fill=blue!30!white] (s00) {$s^0_0$};
            \node[state, fill=orange!30!white] (s01) [below=0.6 of s00] {$s^0_1$};
            \node[state, fill=orange!30!white] (s02) [above=0.6  of s00] {$s^0_2$};
            \node[state, fill=blue!30!white] (s10) [right=2.8 of s00] {$s^1_0$};
            \node[state, fill=orange!30!white] (s11) [below=0.6  of s10] {$s^1_1$};
            \node[state, fill=orange!30!white] (s12) [above=0.6  of s10] {$s^1_2$};
            \node[state] (bot) [right=1.2 of s10] {$\top$};
            \node[state] (top) [right=1.2 of s11] {$\bot$};
    
            \path[->] (s00) edge node [above=0.55, left] {$[0.7,0.9]$} (s11);
            \path[->] (s01) edge node [near start, above] {$0.5$} (s10);
            \path[->] (s01) edge node [above] {$0.5$} (s11);
            \path[->] (s00) edge node [above=0.2, left] {$[0.1,0.3]$} (s12);
            \path[->] (s02) edge node [above] {$1$} (s12);
            \path[->] (s10) edge node [above] {$0.3$} (bot);
            \path[->] (s10) edge node [above] {$0.7$} (top);
            \path[->] (s12) edge node [above] {$1$} (bot);
            \path[->] (s11) edge node [above] {$1$} (top);
        \end{tikzpicture}
        \caption{Unrolling}
        \label{fig:ihmmtrans}
    \end{subfigure}
    }
    \caption{Transformation for monitoring an iHMM}
    \label{fig:ihmmmonitoringex}
    \Description{}{}
\end{figure}

%% file: learning.tex
Above, we derived cautious monitors from iHMMs.
In this section, we iteratively learn and refine the iHMMs using paths sampled from the SuO. We learn iHMMs with  an adaptation Linearly Updating Intervals~\cite{suilen2022robust} and then use conformance checking to decide whether the iHMMs give sufficiently precise monitors.



\subsection{Learning interval HMMs}\label{subsec:LUI}
We describe a process of learning iHMMs from a dataset of paths sampled from the SuO. Our process  
closely follows the learning linearly updating  intervals (LUI) method, applied to learning MDPs~\cite{suilen2022robust}. 
To support learning observation distributions, we encode the uncertainty about the distributions into the transitions, motivated by the transformation~\cite{DBLP:conf/aaai/ChatterjeeCGK15}. 
As in LUI for MDPs, we assume knowledge of the state and transition space of the model. We define our state space using the relevant variables ($V_{sys}$ and further $V_{obs}$ as explained above). We use domain knowledge (e.g., Newtonian physics) to derive the existing  transitions.
In contrast to~\cite{suilen2022robust}, we don't assume a known and unique initial state.

We now outline the initialization. We initialize transition probability intervals by assigning to each transition an initial interval. For example, for a transition between states $i$ and $j$, we assign $I_{i,j} = [I_{i,j}^{\downarrow},I_{i,j}^{\uparrow}]$. We additionally introduce a strength interval $[n_{i,j}^{\downarrow}, n_{i,j}^{\uparrow}]$, which captures the minimum and maximum number of samples the current version of the probability interval is based on.

To estimate transition probabilities, we compute transition frequencies. In particular, $N_{i}$ signifies the total number of times some state $i$ appears in the dataset. $k_{i,j}$ signifies the number of times state $j$ appears after $i$ in the dataset. The empirical transition frequency of a given transition is hence $\frac{k_{i,j}}{N_{i}}$. The probability intervals are updated as follows:
\


{\small
\[{I_{i,j}^{\downarrow}} \!:= 
\begin{cases}
       \frac{{n_{i,j}^{\uparrow}} \times {I_{i,j}^{\downarrow}} + k_{i,j}}{n_{i,j}^{\uparrow} + N_{i}}  & \text{$\forall x. \frac{k_{i,x}}{N_{i}} \!\geq\! {I_{i,j}^{\downarrow}}$}\\
       \\
       \frac{{n_{i,j}^{\downarrow}} \times {I_{i,j}^{\downarrow}} + k_{i,j}}{n_{i,j}^{\downarrow} + N_{i}} & \text{$\exists x. \frac{k_{i,x}}{N_{i}} \!<\! {I_{i,j}^{\downarrow}}$}
\end{cases}
\quad
I_{i,j}^{\uparrow} \!:= 
\begin{cases}
       \frac{n_{i,j}^{\uparrow} \times I_{i,j}^{\uparrow} + k_{i,j}}{n_{i}^{\uparrow} + N_{i}}  & \text{$\forall x. \frac{k_{i,x}}{N_{i}} \! \leq \!I_{i,j}^{\uparrow}$}\\
       \\
       \frac{n_{i,j}^{\downarrow} \times I_{i,j}^{\uparrow} + k_{i,j} }{n_{i}^{\downarrow} + N} & \text{$\exists x. \frac{k_{i,x}}{N_{i}} \!>\! I_{i,j}^{\uparrow}$}
\end{cases}\]
}

We set all initial intervals to $[\epsilon, 1-\epsilon]$ in accordance to the rule $0.5 \geq \epsilon >0$, and choose strength interval values in accordance to $0 \leq n_{i,j}^{\downarrow} \leq n_{i,j}^{\uparrow}$. We use a small $\epsilon$ to make no assumption about prior knowledge and low strength values to allow for the information from the dataset to make a fast impact. 
If all the empirical frequencies from a given state fall within the prior interval, the posterior intervals are calculated with the upper bound of the strength interval. Else, the lower bound of the strength interval is chosen instead. 
Using different strength values supports a more conservative reaction when new data does not conform to probability intervals derived from previous data. After updating the probability intervals, the strength intervals are updated as follows: $[n_{i,j}^{\downarrow},n_{i,j}^{\uparrow}] := [n_{i,j}^{\downarrow} + N_{i}, n_{i,j}^{\uparrow} + N_{i}]$.

Learning the initial state probability intervals is done correspondingly, where $N_{i}$ corresponds to the number of paths in the dataset and $k_{i,j}$ expresses the number of traces starting with state ${j}$. 

The learning progresses iteratively with access to new data. The LUI method produces valid probability intervals, i.e., they always satisfy: $0 \leq I_{i,j}^{\downarrow} \leq I_{i,j}^{\uparrow} \leq 1$. Additionally, the method guarantees  convergence to the true probability. We restate the convergence result from \cite{suilen2022robust}, now in the context of our adapted method.

\begin{theorem}{}\label{thm:lui} For an iHMM with the same state space as some HMM, applying LUI to learning the transition probabilities and initial distribution of the iHMM guarantees that the intervals will converge to the exact transition probabilities when the total number of samples of HMM  paths processed tends to infinity, regardless of how many samples are processed per iteration. 
\end{theorem}



Applying the LUI method in the limit will thus result in a model that induces the ideal monitor.
\begin{corollary}
    LUI yields a sequence of iHMMs, which induces  a sequence of monitors by the pointwise application of \Cref{def:ihmmmonitor}. That sequence converges to the ideal monitor. 
\end{corollary}


\subsection{Conformance-testing-based refinement}\label{sec:refinement}

We present a  conformance-testing-based refinement framework  to learning iHMMs that builds on the LUI procedure from last section. The framework involves an iterative procedure where transition intervals are updated until the iHMM based monitor provides outcomes that meet the stopping condition that we define later.  We particularly address some shortcomings of LUI. 
The LUI method guarantees the convergence to the true HMM, but does not necessarily efficiently converge to this HMM. Empirically, the transitions from states that are frequently visited will converge to narrow intervals very fast, and the transitions from rarely visited states will stay wide for long. 
This will lead to the situation where the outcomes of the iHMM monitors will likely differ from the outcomes of the ideal monitor. To avoid the situation where we keep sampling and learning the neighborhoods where our transition probability approximations are already accurate, we introduce a method of conformance-testing-based refinement that guides the sampling and hence the learning process to underexplored neighborhoods. 
Practically, we wish to process as few samples as possible to learn an iHMM that corresponds to an iHMM monitor that is as close as possible to the ideal monitor. This is the central motivation behind the refinement framework. 

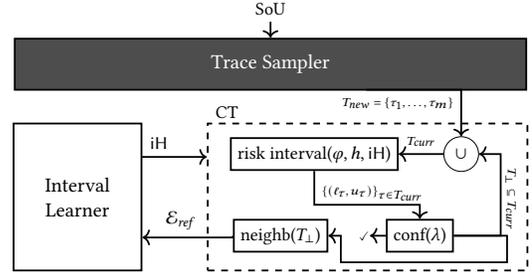
\begin{figure}[t]
\input{framework}
\caption{iHMM refinement-based learning. 
}
\vskip -0.45cm
\label{fig:framework}
\end{figure}

The framework architecture is depicted in \Cref{fig:framework}. It is composed of an \emph{Interval Learner}, which updates a current $\imc$ using LUI, a \emph{trace sampler}, which samples observation traces from SuO to be used for conformance testing, and a \emph{conformance tester} that given a monitor $\mon_\imc$  of  the current iHMM $\imc$ and a set of traces $T$, checks if a stopping condition is met. We use a  stopping condition  defined as a threshold $\theta$ on the average width of what we call the 'risk interval'. For each trace, the interval bounds are the minimum and maximum risk $(\ell_\tau, u_\tau) = (\mon^{\min}_\imc(\tau), \mon^{\max}_\imc(\tau))$, where $u_\tau$ is the outcome of the iHMM Monitor defined in \Cref{sec:iHMMmonitoring} and $\ell_\tau$ is the outcome of the corresponding risk minimizing monitor. If the average width of the risk interval for all considered traces is below the stopping condition threshold $\theta$ the iHMM learning terminates.

Following \Cref{fig:framework}, the refinement learning starts with an initial $\imc$ and a set of sampled traces $T_{new}$. The initial $\imc$ has transition and initial distribution intervals set to $[\epsilon, 1- \epsilon]$. The sampled traces are forwarded to the risk interval analysis and unless the stopping condition is met, each trace that does not meet the threshold is forwarded to the neighborhood sampler.

The neighborhood sampler samples paths from the extended neighborhood of that trace $\tau$, i.e. for a subsequence $\tau'$ of $\tau$ we sample a path $\pi$ with probability $Pr(\pi\mid\tau')$. We discuss the details of the sampling method in Appendix \ref{appendix:neighborhood}. 



All neighborhood samples $\mathcal E_\mathit{ref}$ are forwarded to the Interval Learner to perform the necessary update. The traces $T_{\bot} \subseteq T_{new}$ for which the risk interval is above the stopping condition $\theta$ are retained for the next conformance testing round along the freshly sampled new traces making up the set $T_{curr} = T_{new} \cup T_{\bot}$. The learning concludes when the stopping condition is met. 
Since the refinement does not interact with the method of learning intervals, but only impacts sampling, the following result holds. 


\begin{corollary} \label{cor:lui_with_ref}
Under the assumption that the learned iHMM has the same state as the true HMM and $\theta =0$, applying the conformance-testing based refinement procedure guarantees that the intervals will converge to the exact transition probabilities when the total number of samples processed tends to infinity, regardless of how many samples are processed per iteration.
\end{corollary}

%% file: framework.tex
\centering
\scalebox{0.85}{
\begin{tikzpicture}

\node[draw, thick, minimum width = 8cm, minimum height=0.8cm, fill=black!70] (sampler) at (0,0){\color{white}Trace Sampler};

\node[above = .25cm of sampler](system){\small SoU};
\path[draw, thick, ->](system) -- (sampler);

\coordinate (x1) at (-2.5,-0.4);
\node[draw, thick, minimum width = 2cm, minimum height = 2.3cm, below left = 0.5cm and -2.01cm of sampler, align =center](learner){Interval\\ Learner};

\node[draw, thick, dashed, minimum width = 5cm, minimum height = 2.3cm, below right = 0.5 and -5.01cm of sampler] (ct) {};
\node[above left = -0.05 and -0.6 of ct](ct-name) {\small CT};

\coordinate (x2) at (1.5,-0.4);
\node[draw, thick, minimum height = 0.5cm, below left = -2.1cm and -3.0cm of ct] (mm){\small risk interval($\varphi,h,\imc$)};
\path[draw, thick, ->](learner.32) |- node[above right ]{\small$\imc$} (ct.166);

\coordinate (x3) at (1.6,-2.2);
\node[draw, thick, minimum height = 0.5cm, below right = 0.75cm and -0.2cm of mm] (conf){\small conf($\lambda$)}; 
\path[draw, thick, ->] (mm.south) |- node[ align = center, above right = -0.05 and 0.0 of mm, text width = 1cm] {\tiny $ \{ (\ell_\tau,u_\tau)\}_{\tau \in T_\mathit{curr}}$} (x3) -| (conf.north);

\node[draw, thick, minimum height = 0.5cm, left = 0.9 cm of conf] (cex){\small neighb($T_\bot$)};
\path[draw, thick, ->] (cex.west) |- node[above left = 0 and 0.5cm]{\small $\mathcal E_{\textit{ref}}$}(learner.328);

\node[left = 0.25cm of conf, inner sep = 0](done){\tiny $\checkmark$};
\path[draw, thick, ->] (conf.west) -- (done.east);

\coordinate (x4) at (3.6, -1.6);
\coordinate (x5) at (3.7, -3.1);
\coordinate (x6) at (1.1, -3.1);
\node[draw, circle, right = 0.7cm of mm] (union){\small $\cup$};
\path[draw, thick, ->]  (x2) -|  node[below left]{\tiny $T_\mathit{new}= \{\tau_1,\dots,\tau_m\}$}(union.north);
\path[draw, thick, ->] (conf.east) -| node[above left = -0.02cm and .04cm, rotate = -90]{\tiny $T_\bot\subseteq T_\mathit{curr}$}(x4) |- (union.east);
\path[draw, thick,->](conf.east) -| (x5) -- (x6) |- (cex.east);
\path[draw, thick, ->](union.west) --node[above = 0cm,] {\tiny $T_\mathit{curr}$} (mm.east) ;

\end{tikzpicture}
}

%% file: experimental_evaluation.tex
\section{Experimental Evaluation}
\label{sec:empirical}

We address three research questions. First, does our model-based approach improve over model-free approaches? (\ref{sec:model_based_free}) Second, does using robust models enhance safety while maintaining accuracy? (\ref{sec:ihmm_vs_hmm_comp}) Third, does conformance-testing-based refinement lead to a better approximation of the Ideal Monitor? (\ref{sec:No_ref_vs_ref})

\subsection{Experimental Setup}
\label{sec:experimental_setup}

\paragraph{Terminology.} False Negative Rate (FNR), False Positive Rate (FPR), Area Under the Curve (AUC), Stopping Condition (SC).  

\paragraph{Implementation.} We provide a prototype implementation for the framework in \Cref{fig:framework} and for monitoring over iHMMs as described in \Cref{sec:iHMMmonitoring}. The implementation of the monitoring algorithm builds on Premise~\cite{DBLP:conf/cav/JungesTS20}, a model-checking-based monitoring tool, adapting it to iHMMs. 
The artifact is available in~\cite{antonina_2026_18631204}.

\paragraph{Benchmarks.} We use benchmarks adapted from \cite{DBLP:conf/cav/JungesTS20,DBLP:journals/corr/abs-2504-05963}. The benchmarks are listed in the table below. The airport benchmarks models a scenario with an autonomous airplane descending for landing with on-ground vehicles moving across the runway. 
It considers the distance of the airplane to the runway, the position of the on-ground vehicles and the observed position of these vehicles. The models airportB include a larger version of the noise model used in airportA. The evadeV benchmark models the behavior of two robots, with one robot observing the position of the other. The difference between different versions of the same benchmark (ex. airport-A-7-10-10, airport-A-7-40-20) lies in how detailed the state and observation information is. SnL is a toy benchmark modeling the snakes and ladders game and is used mainly for sanity checks. 

On each of the benchmarks, we monitor reachability properties, which in the case of airport and evade benchmarks means we are estimating the risk of a future collision and in the case of SnL-10x10 the risk of reaching the final state of the game. 

 \begin{center}
    \scalebox{0.8}{
\begin{tabular}{l l r r r r}
\toprule
& Name & States & Transitions & $|\tau|$ & $h$ \\
\midrule
& SnL-10x10$^{(*)}$ & 101 & 502 & 15 & 5 \\
& evadeV-5-3 & 1001 & 3878 & 10 & 20 \\
& evadeV-6-3$^{(*)}$ & 2 161 & 8 667 & 12 & 20 \\
& airport-A-7-10-10$^{(*)}$ & 1 170 & 5 557 & 15 & 25 \\
& airport-A-7-40-20 & 10 760 & 56 577 & 125 & 35\\
& airport-B-7-40-20$^{(*)}$ & 21 520 & 152 170 & 125 & 35 \\
\bottomrule
\end{tabular}
} 
\end{center}





The benchmarks marked with an asterisk $^{(*)}$ have been additionally tested in a 'coarse' variant, where part of the system variables are not given to the learner.  We use these variants to study the impact of a reduced state space on the accuracy of the learned monitors. Note that the convergence guarantees discussed in \Cref{thm:lui} and \Cref{cor:lui_with_ref} do not hold for the coarse models. 

\paragraph{Learning setup}
We learn iHMM using our adapted LUI method described in \Cref{subsec:LUI} with neighborhood sampling and monitor based on these models, as described in \Cref{sec:iHMMmonitoring}. 
The monitors are learned from a dataset of paths of length $|\tau| + h$. During testing we use traces of length $|\tau|$ and set the monitoring horizon to $h$. The values of  $|\tau|$ and $h$ are specified in the benchmark table. The learned monitors are tested on 500 traces. For memory reasons, we test airport-A-7-40-20 and airport-B-7-40-20 benchmarks on 200 traces. Benchmark airportB-7-40-20 did not complete within a 12 hour timeout and its results are therefore omitted. 
Each method for each benchmark is executed 10 times. Unless noted otherwise, we plot averages and (shaded) areas with ±1 standard deviation. 
All experiments were run on an Intel Xeon Gold 6338 CPU using 16 cores and an NVIDIA A100 GPU, and using a 12h timeout. Only the conformal prediction runs on the GPU, all other computations are CPU based.

\subsection{Model-based vs Model-free monitoring} \label{sec:model_based_free}

In this experiment, we compare our model-based approach to two model-free monitoring approaches. We consider two model-free approaches: (1) a Stochastic Gradient Descent (SGD) regressor learned via the sklearn library \cite{scikit-learn}, and (2) a conformal-prediction approach \cite{cairoli2021neural}, that was specifically designed for monitoring under uncertainty; this approach uses a neural approach to estimating risks and conformal prediction to validate the computed risks.

The model-free methods are trained in one shot using the highest number of samples used by the refinement learning. Our comparison is done in terms of FPR, judging conservativeness, and FNR, judging safety, across different thresholds. For all benchmarks, we choose SC on the width of risk intervals based on the size of the benchmark. Choosing a lower SC for smaller benchmarks ensures that learning requires non-trivial amounts of data.

\begin{figure*}[t]
    \centering
    \begin{subfigure}{0.33\textwidth}
        \includegraphics[width=\linewidth]{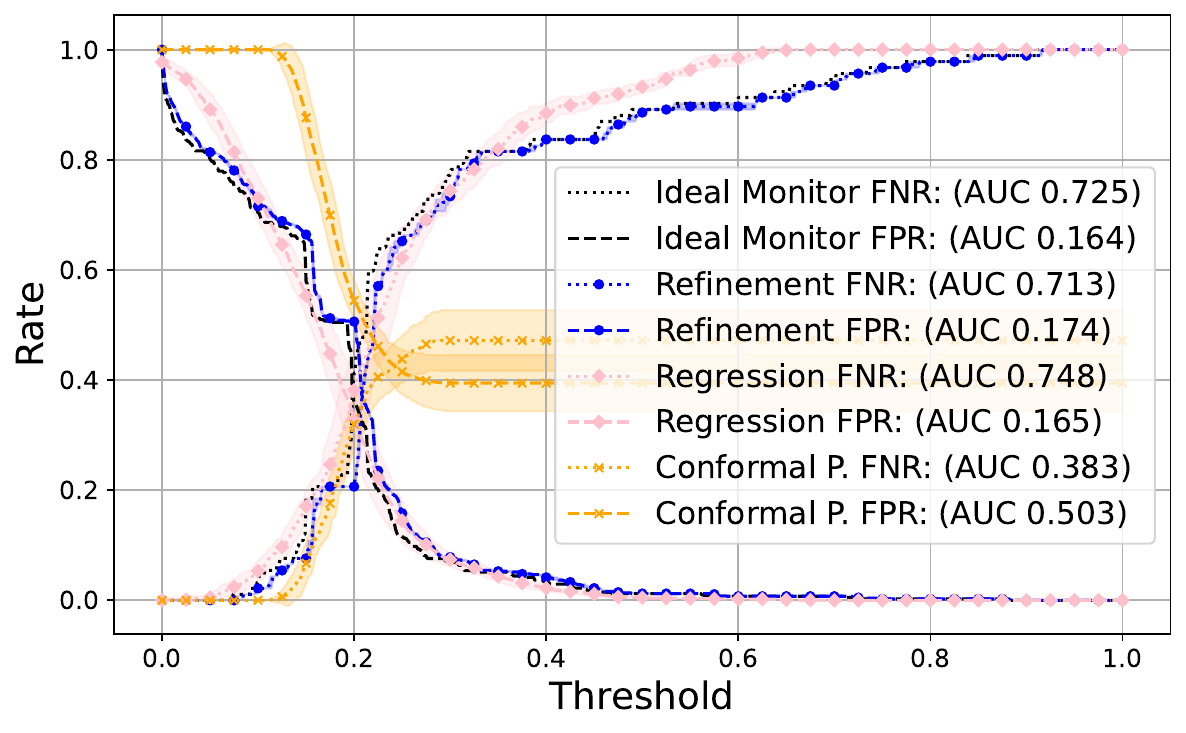}
        \caption{airportA-7-10-10, SC = 0.001}
    \label{model_free_model_based_airportA-7-10-10}
    \end{subfigure}
    \hfill
    \begin{subfigure}{0.33\textwidth}
        \includegraphics[width=\linewidth]{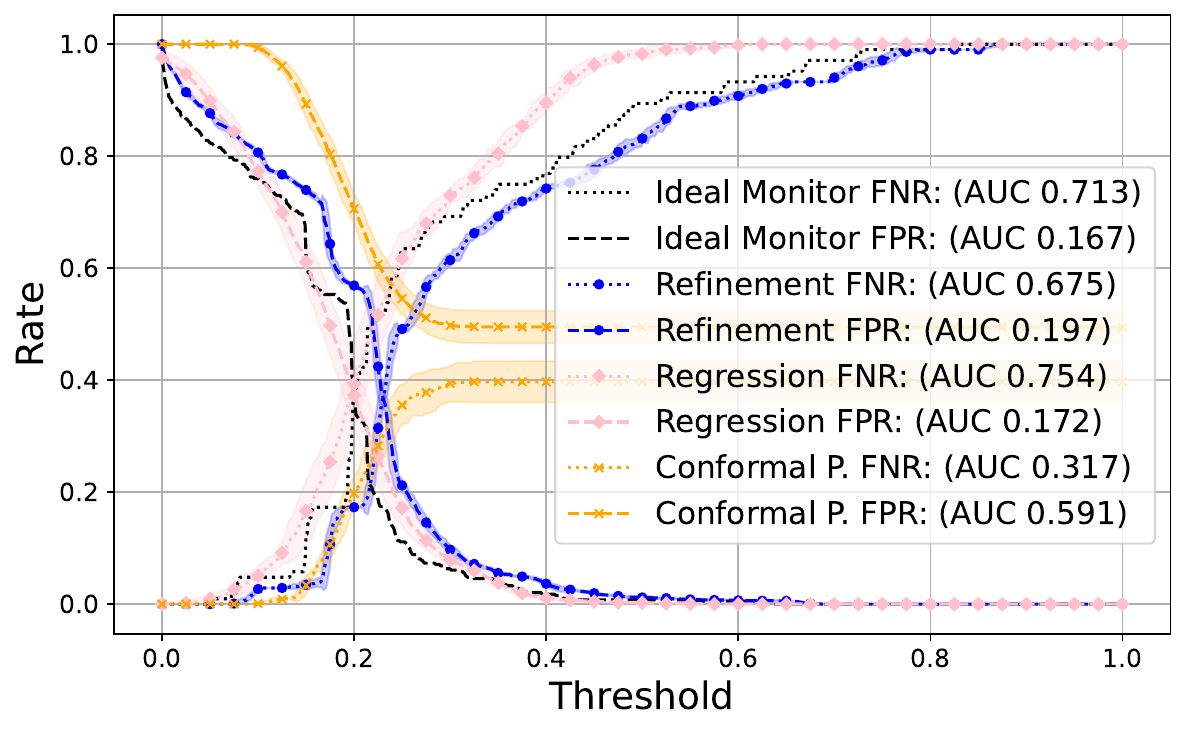}
        \caption{Coarse airportA-7-10-10, SC =  0.01}
        \label{model_free_model_based_airportA-7-10-10-coarse}
    \end{subfigure}
    \hfill
    \begin{subfigure}{0.33\textwidth}
        \includegraphics[width=\linewidth]{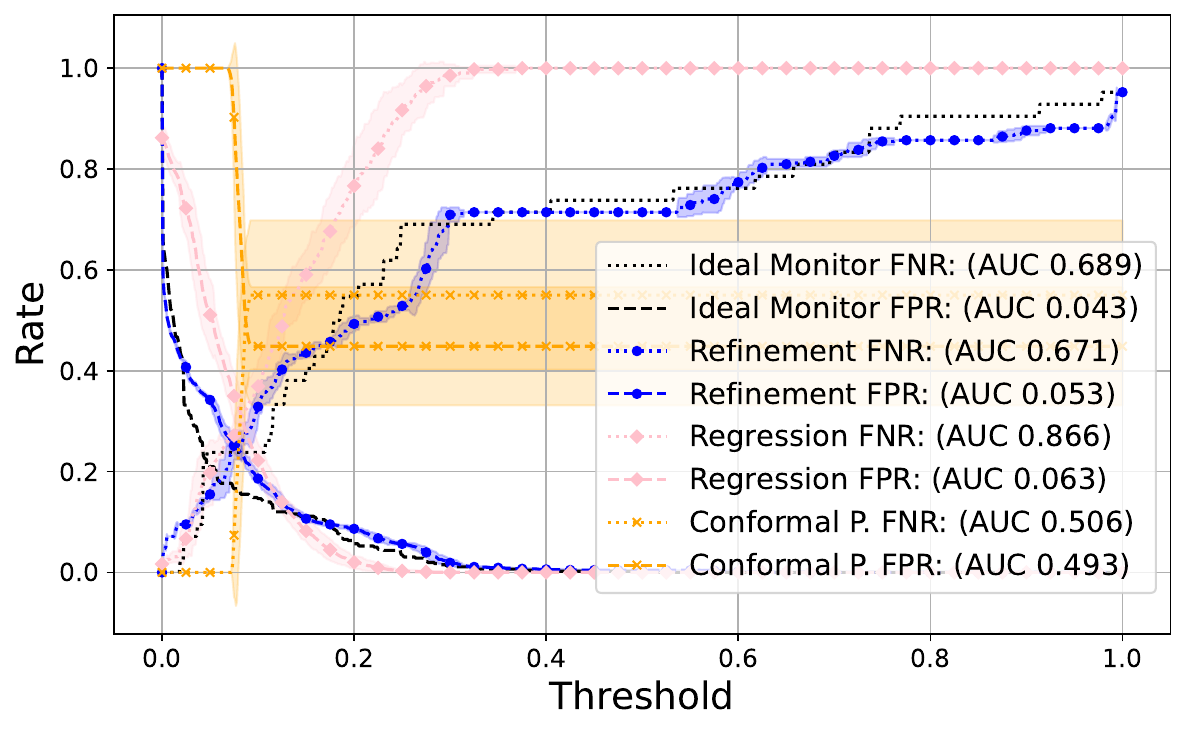}
        \caption{Coarse SnL-10x10, SC = 0.01}
        \label{main_model_free_model_based_SnL-10x10_coarse}
    \end{subfigure}
    \hfill
    \caption{Comparing FNR and FPR between model-based and model-free methods.}
    \label{fig:model_free_model_based_comp}
    \Description{}{}

\end{figure*}

We conducted this experiment over all benchmarks. In \Cref{fig:model_free_model_based_comp}, we show a representative set of our results for airportA-7-10-10 and SnL-10x10. The remaining results are found in \Cref{apndx:modelvsfree}. 
Our goal is to create a cautious monitor that is as close to the ideal monitor $\mon ^{*}$. In \Cref{fig:model_free_model_based_comp} the FNR and FPR of  $\mon ^{*}$ are depicted by the dotted and dashed black lines, respectively. We plot the rates for increasing thresholds, where a threshold of 1 means no trace is considered unsafe. We aim at monitors with a FNR curve close and consistently below $\mon ^{*}$'s FNR and hence with an FNR AUC slightly lower than that of  $\mon ^{*}$. 

\textit{We observe the advantages of iHMM-based
monitors in terms of supporting safety without being overly conservative}. The diagrams clearly show that the iHMM monitors consistently provide a lower FNR than that of $\mon ^{*}$ across all thresholds for every benchmark, aside from Coarse SnL-10x10. For some thresholds, the iHMM monitor has a higher FNR than $\mon ^{*}$ (\Cref{main_model_free_model_based_SnL-10x10_coarse}). A FNR higher than that of $\mon ^{*}$ indicates that for some traces the learned monitor underappoximated the risk leading to the trace being wrongly classified as safe. Although, the iHMM monitor having at times a higher FNR than $\mon ^{*}$ in the case of the Coarse SnL-10x10 benchmark is concerning, the alternatives perform much poorer, with regression having a far greater FNR. The FPR remains modest and closely follows the FPR of $\mon ^{*}$ (\Cref{model_free_model_based_airportA-7-10-10}). 
This is also observed for the other benchmarks (Appendix \ref{apndx:modelvsfree}). 
In general, the accuracy of the iHMM monitors is improved when considering the true state space, with \Cref{model_free_model_based_airportA-7-10-10} showing iHMM monitors that almost perfectly mimic the behavior of $\mon ^{*}$.

\textit{Model-free monitors lack accuracy and compromise safety.}
The conformal prediction monitors severely underperform in terms of accuracy. We attribute this to the fact that these monitors have been designed for systems with deterministic dynamics. Regression monitors similarly struggle with accuracy. A shortcoming of regression present in all experiments is that it fails to fit the behavior of $\mon ^{*}$. The FPR and FNR curves are overly smooth and don't fit the shape of those of $\mon ^{*}$. 


\begin{figure*}
    \centering
    \scalebox{1}[0.9]{
    \begin{subfigure}{0.23\textwidth}
        \includegraphics[width=\linewidth]{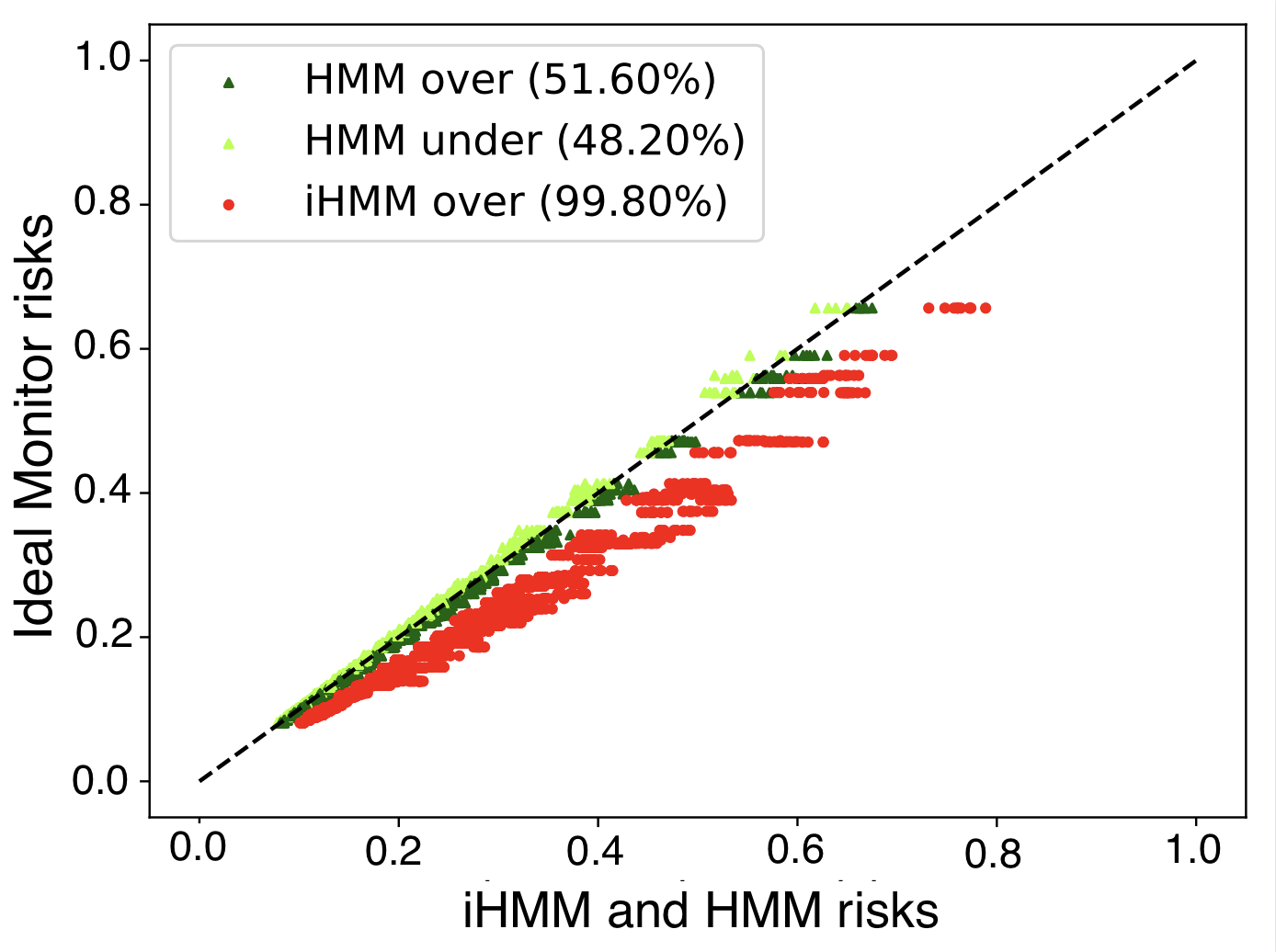}
        \caption{evadeV-5-3, SC= 0.01}
        \label{ihmm_hmm_overestimation_evadeV-5-3}
    \end{subfigure}
    }
    \hfill
    \scalebox{1}[0.9]{
    \begin{subfigure}{0.23\textwidth}
        \includegraphics[width=\linewidth]{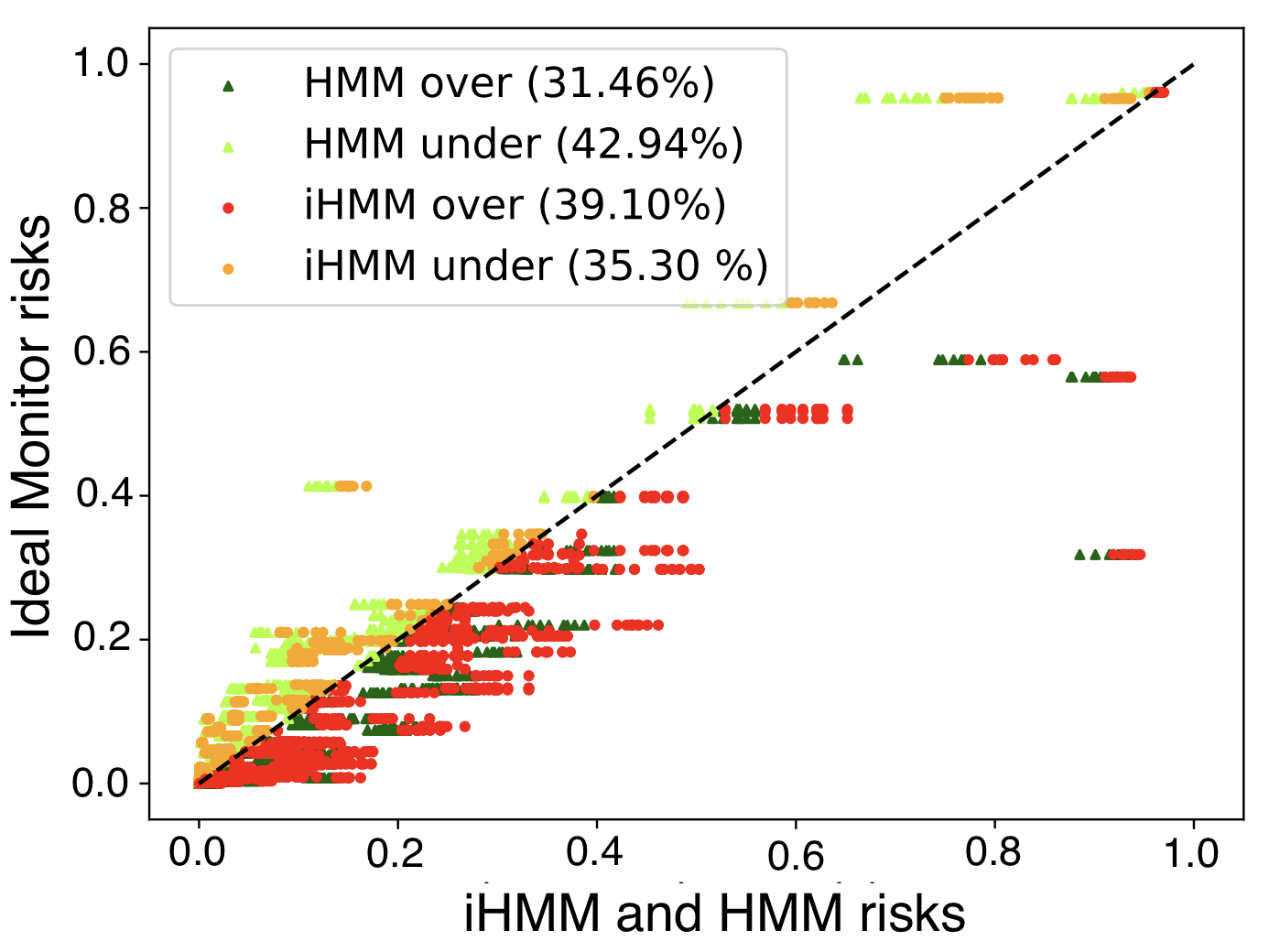}
        \caption{Coarse SnL-10x10, SC= 0.01}
        \label{ihmm_hmm_overestimation_SnL-10x10_coarse}
    \end{subfigure}
    }
    \hfill
    \scalebox{1}[0.9]{
    \begin{subfigure}{0.23\textwidth}
        \includegraphics[width=\linewidth]{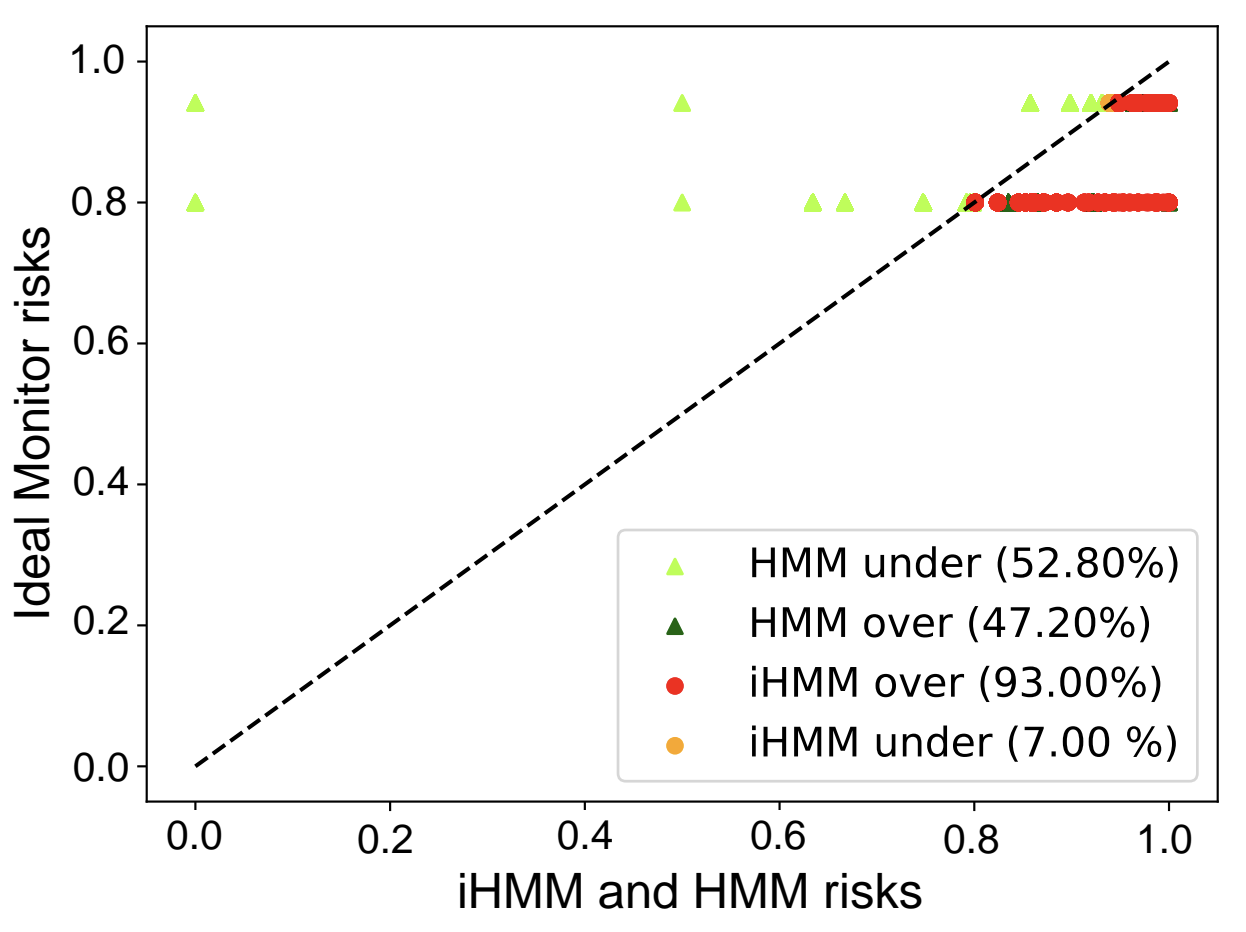}
        \caption{unlikely-15, SC= 0.01}
        \label{ihmm_hmm_overestimation_unlikely}
    \end{subfigure}
    }
    \hfill
    \scalebox{0.96}[0.9]{
    \begin{subfigure}{0.27\textwidth}
        \includegraphics[width=\linewidth]{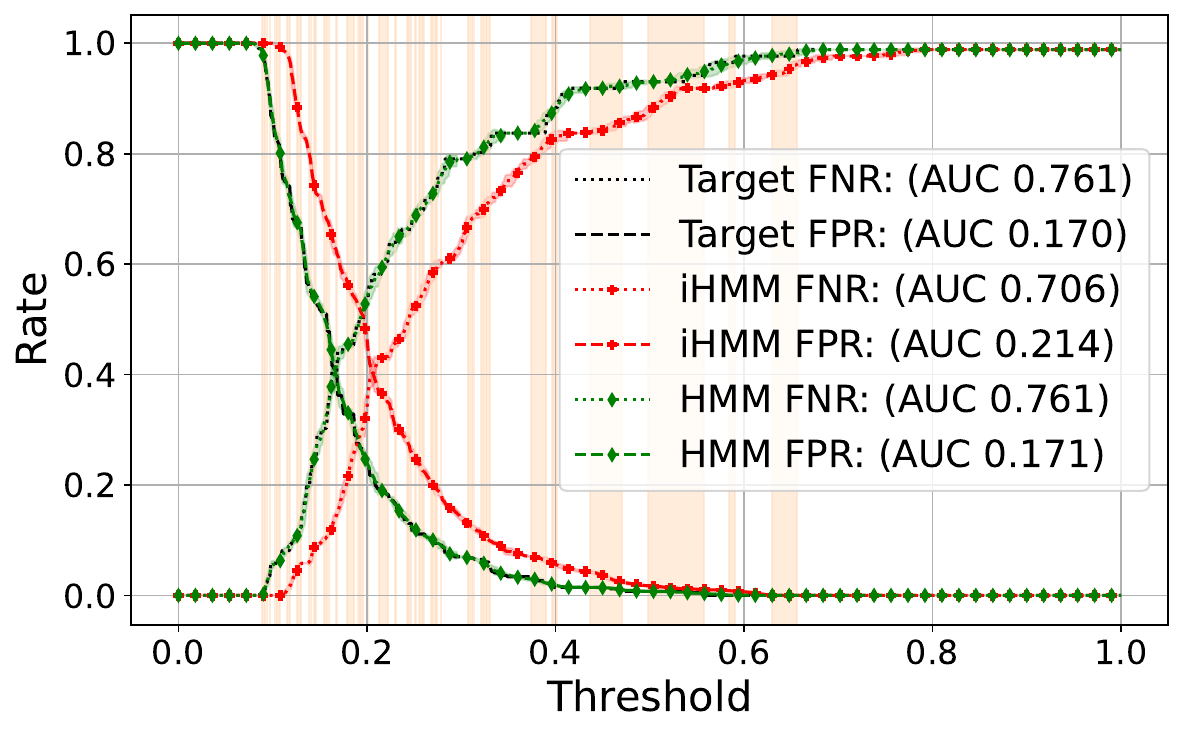}
        \caption{evadeV-5-3 SC = 0.01 }
        \label{ihmm_hmm_evadeV-5-3}
    \end{subfigure}
    }
    \caption{Comparison between iHMM and HMM-based monitors, in terms of risk estimation (a-c) and FNR and FPR (d). Scatterplots (a-c)  plot the values from all 10 runs for each method.} 
    \label{fig:ihmm_hmm_comp}
    \Description{}{}
\end{figure*}
\begin{figure*}
 \centering
    \begin{subfigure}{0.22\textwidth}
        \includegraphics[width=\linewidth]{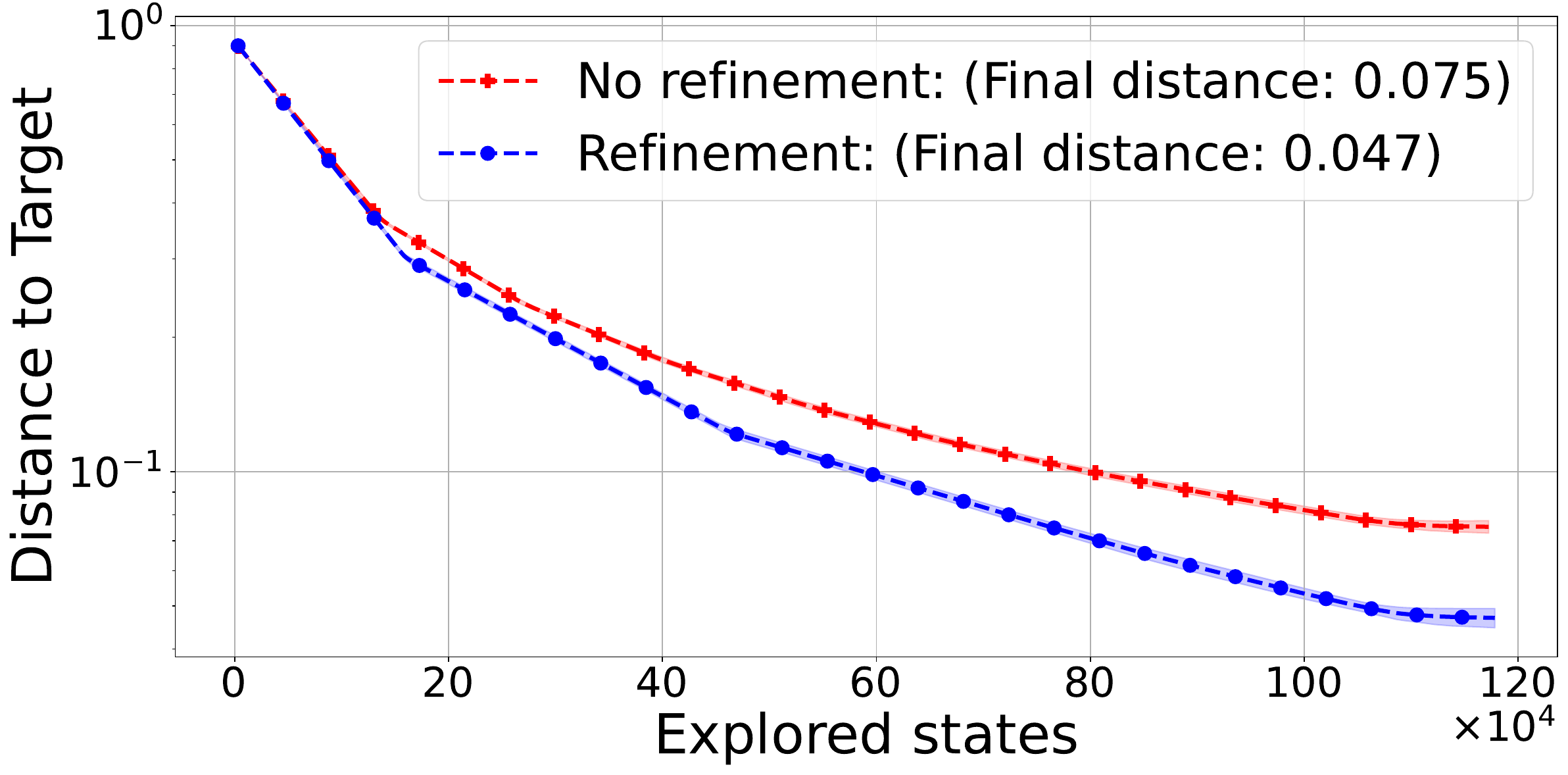}
        \caption{Coarse evadeV-6-3, SC =  0.01}
        \label{ref_no_ref_evadeV-6_coarse_distance}
    \end{subfigure}
    \hfill
    \begin{subfigure}{0.22\textwidth}
        \includegraphics[width=\linewidth]{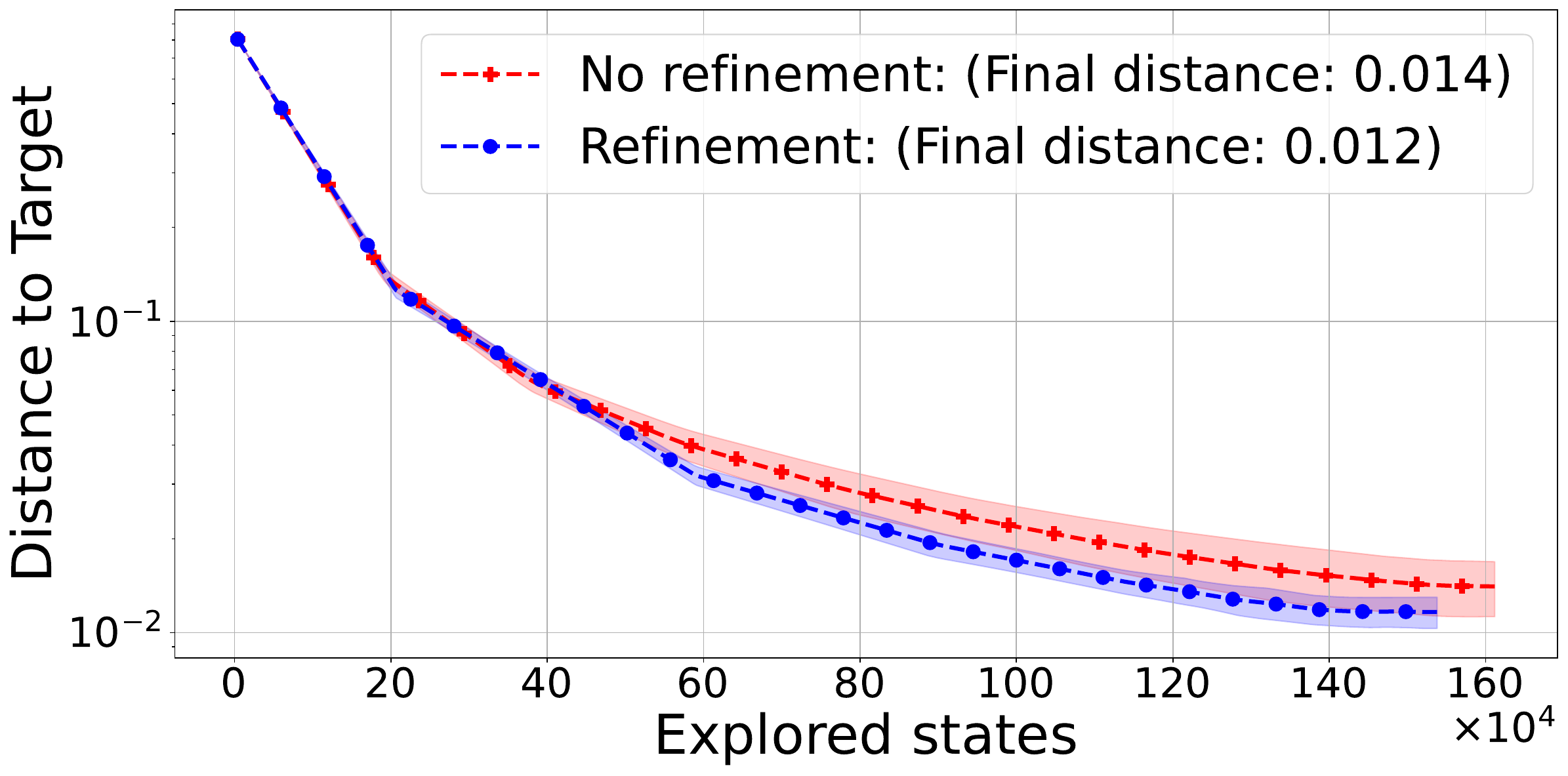}
        \caption{airportA-7-10-10, SC = 0.001}
        \label{ref_no_ref_airportA-7-10-10_distance}
    \end{subfigure}
    \hfill
    \centering
    \scalebox{1}[0.85]{
    \begin{subfigure}{0.22\textwidth}
        \includegraphics[width=\linewidth]{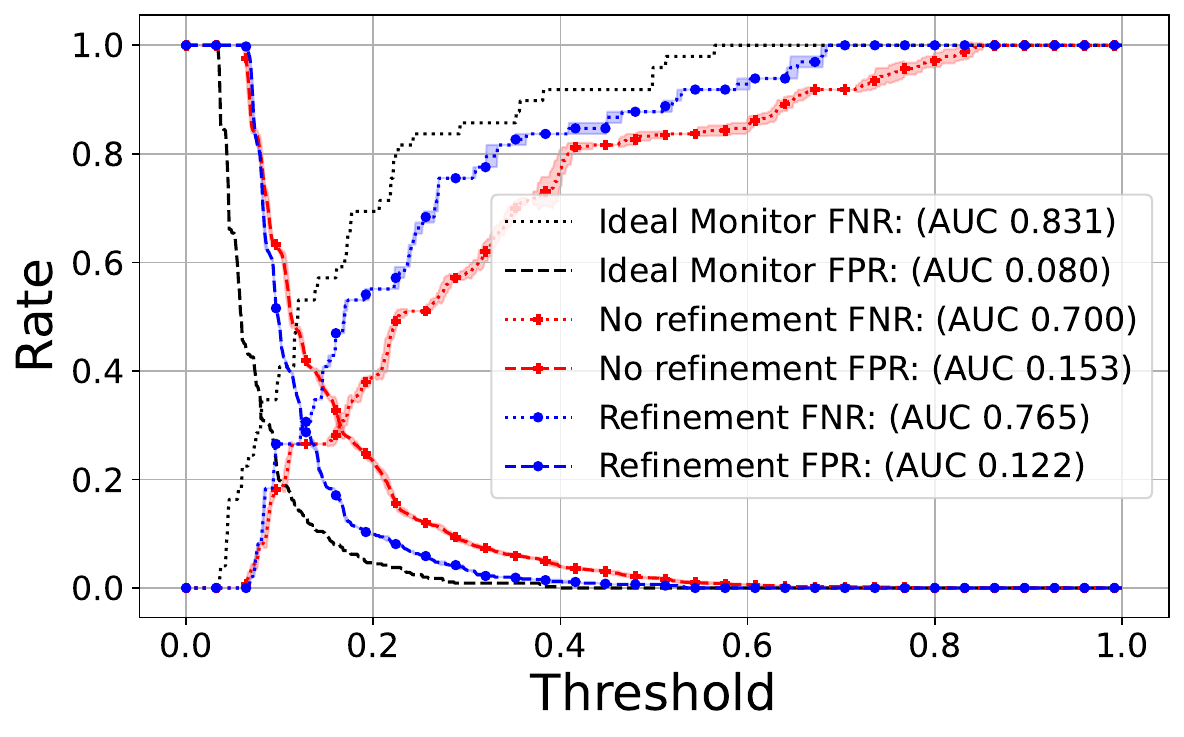}
        \caption{Coarse evadeV-6-3, SC= 0.01}
        \label{ref_no_ref_evadeV-6_fn_fp}
    \end{subfigure}
    }
    \hfill
    \scalebox{1}[0.85]{
    \begin{subfigure}{0.22\textwidth}
        \includegraphics[width=\linewidth]{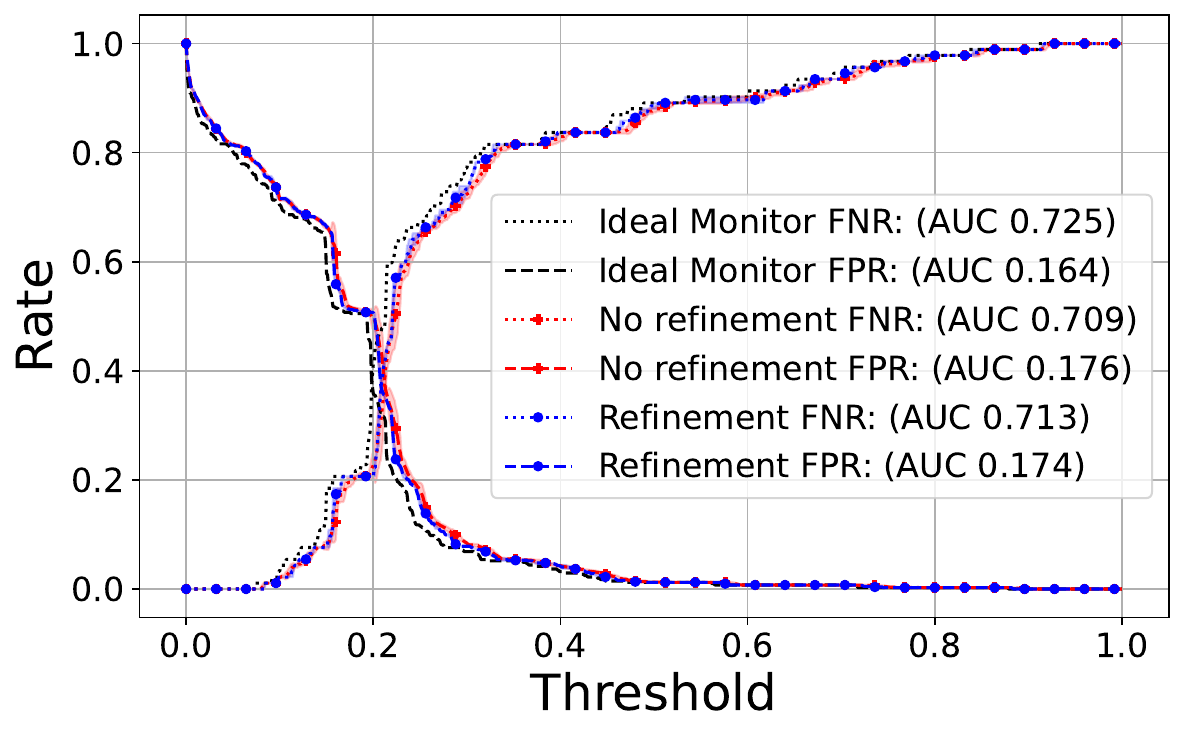}
        \caption{airportA-7-10-10, SC= 0.001}
        \label{ref_no_ref_airportA-7-10-10_fn_fp}
    \end{subfigure}
    }
     \caption{Comparing learning with and without refinement, in terms of distance to $\mon ^{*}$ (a-b), FNR and FPR (c-d)} 
     \label{fig:ref_no_ref}

\end{figure*}

\subsection{iHMM vs. HMM-based monitoring} \label{sec:ihmm_vs_hmm_comp}

In this section, we explore the impact of using robust models for safe monitoring. We compare the risk values approximated by the iHMM and HMM-based monitors to risk values from $\mon ^{*}$ and categorize them as under- and overapproximations. Furthermore, we consider the FPR and FNR to study the monitoring outcomes across different thresholds.  
We learn HMM models using the highest number of samples used by the refinement learning. We split the dataset into 10 equal batches and learn on an increasing number of batches resulting in 10 HMM models per run. The HMMs are learned using a simple frequentist approach. 

\textit{For all benchmarks (Appendix \ref{apndx:ihmm_hmm}), we observe a positive effect of iHMM-based monitors in terms of safety.} EvadeV-5-3 (\Cref{ihmm_hmm_overestimation_evadeV-5-3}) is a benchmark that shows the typical performance of iHMM and HMM monitoring. The figure plots risks approximated by HMM and iHMM monitors and compares them to Ideal Monitor risks, where a single point signifies a value of the Ideal Monitor risk and either a value of the iHMM (red) or HMM risk (green). 
We observe a significant difference between the iHMM and HMM based monitors in terms of under- and overapproximations. As seen \Cref{ihmm_hmm_overestimation_evadeV-5-3}, for almost half of the traces, the HMM-based monitor underapproximates the risk, whereas iHMM-based monitor almost exclusively overapproximates the risk, hence supporting safety. The result is existence of many thresholds for which HMM-based monitors have a higher FNR than that of $\mon ^{*}$ (orange areas in \Cref{ihmm_hmm_evadeV-5-3}). Consequently, for those thresholds, we can expect the HMM-monitor to misclassify some unsafe traces as safe. This is never the case of the iHMM-based monitor. 
The only benchmark for which we observe significant risk underapproximations produced by iHMM-based monitor is coarse SnL-10x10, \Cref{ihmm_hmm_overestimation_SnL-10x10_coarse}.  For this benchmark, however, the iHMM-based monitor underapproximates the risk for fewer traces, than the HMM-based monitor.

\textit{The HMM-based monitors can result in stark risk underapproximations.} While many traces in HMMs lead to underapproximating risk, the difference tends to be marginal. However, this is not guaranteed to always be the case. We introduce an additional benchmark, unlikely-15, which serves as an example of a HMM monitor failing to recognize high-risk scenarios. The structure of unlikely-15 is such that some of its neighborhoods are extremely unlikely to be encountered, hence the iHMM has in those neighborhoods wide probability intervals, which leads to risk overapproximations. The HMM, however, having access to very few samples, learns probabilities far from the true values, leading at times to sizable risk underapproximations. \Cref{ihmm_hmm_overestimation_unlikely} plots 100 traces with positive risk, identified through rejection sampling. It includes examples of traces, where HMM monitors misidentifies high risk traces as low risk, likely leading to violating the safety specification.  
\subsection{No refinement vs refinement learning} \label{sec:No_ref_vs_ref}

Lastly, we examine the impact of refinement-based iHMM learning on the monitoring outcomes. We consider the change in the distance to $\mon ^{*}$ and the difference between FPR and FNR. For this experiment, we further learn the iHMM models without refinement using the highest number of samples used the refinement learning. We split the dataset into 10 equal batches and learn by sequentially adding batches, resulting in 10 iHMM models per run. 

\textit{Learning with refinement consistently leads to monitors that, at the end of the learning process, are closer to $\mon ^{*}$.} 
Note that the used distance measure is $\delta = \frac{\sum_{\tau \in T_{\mc}} |\mon^{*}(\tau) - \mon_{\imc}(\tau)|}{|T_{\mc}|}$, where in the experimental evaluation we consider only a sampled subset of $T_{\mc}$. As expected, there is a relationship between the distance to $\mon ^{*}$ and the FNR and FPR values. For the benchmarks where we observe a significant impact of the refinement method in terms of distance to $\mon ^{*}$, we observe that the refinement method results in FNR and FPR curves closer to the ideal FNR and FPR curves (\Cref{ref_no_ref_evadeV-6_coarse_distance}, \ref{ref_no_ref_evadeV-6_fn_fp}). The benefit of the refinement procedure is improved sample efficiency. For example, in the case of \Cref{ref_no_ref_evadeV-6_coarse_distance}, we are able to learn as good a monitor as no refinement learning method while exploring only about two thirds of the states required by the non-refinement approach.  

\textit{The refinement procedure has a varied impact across different benchmarks.} In the case of the coarse evadeV-6-3 benchmark, the monitors corresponding to iHMMs learned with refinement get significantly closer to  $\mon ^{*}$ than with no refinement (\Cref{ref_no_ref_evadeV-6_coarse_distance}). The effect is that the refinement monitors approximate better the ideal FNR and FPR curves (\Cref{ref_no_ref_evadeV-6_fn_fp}). On the other hand, in the case of the airportA-7-10-10 benchmark, there is little difference between refinement and no refinement as  in \Cref{ref_no_ref_airportA-7-10-10_distance} and \ref{ref_no_ref_airportA-7-10-10_fn_fp}.

We additionally considered an alternative sampling method, that showed little difference during experimental evaluation. We explain the method in Appendix \ref{appendix:neighborhood} and experimentally compare the two sampling methods in Appendix \ref{sec:ihmm_vs_hmm_comp}.

\subsection{Further Remarks} 

\textit{Results on largest benchmarks.} Learning monitors for the two biggest benchmarks, airportA-7-40-20 and airport-B-7-40-20, posed challenges. For this reason, we discuss learning monitors for these benchmarks separately. Some of the previously made conclusions hold, namely the positive effect of refinement on distance to $\mon ^{*}$. We also observe that the iHMM, unlike the HMM-based monitors, makes safe risk approximations. Although promising, we refrain from drawing strong conclusions from these experiments, as all monitors from all learning methods produce outputs far removed from $\mon ^{*}$. We theorize that the reasons are too high SC, and that the violating paths are significantly rarer in these benchmarks. A detailed discussion  is given in Appendix \ref{apx:big_benchmarks}.

\textit{Parameter Sensitivity.} We note the impact of coarse state space and low stopping condition thresholds. The experimental results highlight the importance of identifying the relevant state variables and choosing a low stopping thresholds for learning monitors that closely approximate $\mon ^{*}$. 

\textit{Computational Cost.} All experiments were run with a 12h timeout, but most benchmarks required significantly less time. Taking as an example evadeV-6-3 with SC=0.1, learning an iHMM with refinement took around 110s, learning an iHMM without refinement took around 220s, learning a HMM took around 9s, training a regression model took around 10s and training a conformal prediction took around 345s, tough importantly conformal prediction is the only method trained on a GPU. The evaluated monitors differ as well in terms of execution speed. Taking again as an example evadeV-6-3 tested on traces of length 12 with a horizon set to 20, the HMM, iHMM, regression and conformal prediction monitors produce risk estimations in correspondingly: 1.8, 20, 0.003 and 1.6 milliseconds. It's worth noting that optimizing the speed of learning and executing the iHMM monitors was not the focus of this work. 


%% file: related.tex
\section{Related Work}
Monitoring systems under uncertainty have been studied across many settings, from model-based approaches \cite{babaee2018predictive, DBLP:conf/ifm/JungesST24, babaee-rv19, DBLP:conf/cav/JungesTS20,camilli2021,DBLP:conf/rv/YoonC0F019} to model-free approaches \cite{cairoli2021neural,DBLP:conf/iccps/ZhaoHFDL24,torfah-rv23}. 
Approaches based on learning Markov models include \cite{babaee2018predictive}, which introduces a monitoring algorithm for the learned HMMs that relies on selecting a single most likely hidden state during monitoring, which compared to considering a distribution over possible states, lacks the precision for accurate monitoring under uncertainty.
This is expanded on in \cite{babaee-rv19} to include importance sampling (IS) to learn rare events more efficiently, similarly to our refinement procedure. However, their IS is based on using a modified simulation in which the rare events occur with a known higher likelihood and is thus not applicable in our setting.
An example of model-free approaches can be taken from \cite{cairoli2021neural}, which uses a neural approach for state estimation and conformal prediction on top for risk estimation. A bottleneck of this approach  is that it is designed for monitoring systems with deterministic dynamics, which stands in contrast to ours. 
The literature also includes several works addressing the problem of learning Markov models from data. These include active learning method for learning Markov Decision Processes \cite{activelearning} or an approach for learning Hidden Markov Models \cite{babaee2018predictive}. Both of these are based on Baum-Welch algorithm and hence highly sensitive to the training priors used. 
Related is also the idea of shielding in reinforcement learning, which has been applied to partially observable systems \cite{shielding_nam,shielding_chatterjee}. In contrast, however,  shielding requires to (pre-)compute partial-information schedulers, and foremost the availability of models. 
Several works have been also dedicated to runtime assurance (enforcement), where the question on how to use the monitors verdict to further change the behavior of a system \cite{DBLP:journals/tcad/YalcinkayaTDS23,8809550,simplex,8104025}. Our focus was on monitoring, and our approach can be integrated  into these approaches for a complete assurance pipeline.

%% file: conclusion.tex
We presented a framework for learning robust Markov model-based monitors.  Our methods come with guarantees on convergence to an Ideal Monitor. As shown experimentally, the learned monitors outperform the alternative methods in terms of accurate and cautious monitoring, especially compared to model-free approaches.

%% file: acknowledgements.tex
This work was partly supported by the Wallenberg AI, Autonomous Systems and Software Program (WASP), funded by the Knut and Alice Wallenberg Foundation, and by the NWO grant FuRoRe (OCENW.M.22.282).
The computations were enabled by resources provided by the National Academic Infrastructure for Supercomputing in Sweden (NAISS), partially funded by the Swedish Research Council through grant agreement no. 2022-06725.

%% file: appendix_refinement.tex
\label{appendix:neighborhood}

In the following we give more details on neighborhood path sampling methods introduced in \ref{sec:refinement}. The neighborhood sampling works in two steps, the first step finds a bad neighborhood, the second step samples from that bad neighborhood. 

The first step starts with identifying  bad prefixes, $\tau \in \hat{T}_\bot$, of the bad traces, $T_\bot$. We use two different methods of identifying bad prefixes, which distinguish the two refinement methods ('refinement with no splitting' and 'refinement with splitting'). 

In the first method we chose $N_{neigh} = \{0,\ldots,10\}$ linearly increasing in size prefixes $\hat{T}_\bot = \bigcup_{\tau\in T_\bot}\{\tau^{\leq \frac{n}{N_{neigh}}|\tau|} \mid n \in N_{neigh}\}$. 

In the second method we consider the model as containing two parts which might overlap. The first part, we call the initial section, models the probability of a trace reaching a state, the second part, we call the horizon section, models the risk of a state that is reached. For each bad trace, we identify if the initial section or the horizon section is worse. We do this by comparing the minimum and maximum risk of a randomly sampled state reachable with the bad trace. If the width of this risk interval is greater than the stopping condition threshold, we add the full trace to the bad prefixes, otherwise we add the empty trace. We hance will correspondingly sample from the neighborhood that starts where the trace ends, or the neighborhood that the trace spans. 

After the bad prefixes have been identified using either of the methods, we randomly sample from the states conditioned on observing the prefix $\tau'$, i.e., we sample a state $s$ with probability $\sum_{\pi\in\Pi^{|\tau'|-1}_\imc} Pr(\pi s\mid \tau')$, we call the set of sampled states $S_\bot$. This concludes the first step. 

In the second step, we sample from the $\suo$ starting at states in $S_\bot$, resulting in the paths $\mathcal E_\mathit{ref}$. 

The goal of the first approach is to uniformly sample from the whole neighborhood for the trace. The second approach aims to identify if the root of the wide risk interval stems from inadequate risk estimation or poor estimation of distribution of paths that induce the trace, and focuses on the root issue. 

\subsection{Experimental Comparison}

Experimentally refinement with splitting performs on average marginally better than refinement with no splitting, as seen in Figures \ref{ref_no_ref_evadeV-6_coarse_distance}, \ref{ref_no_ref_airportA-7-10-10_distance}, \ref{ref_no_ref_evadeV-6_fn_fp} and \ref{ref_no_ref_airportA-7-10-10_fn_fp}. The study of the conditions under which the refinement methods are beneficial, and the development of improved refinement methods, is left for further study. The remaining results of the experimental comparison between the two refinement methods are available in Appendix \ref{apx:ref_no_ref}. The experimental evaluation in \ref{apndx:modelvsfree} included the results of both refinement methods.

\begin{figure}[H]
        \includegraphics[width=0.8\linewidth]{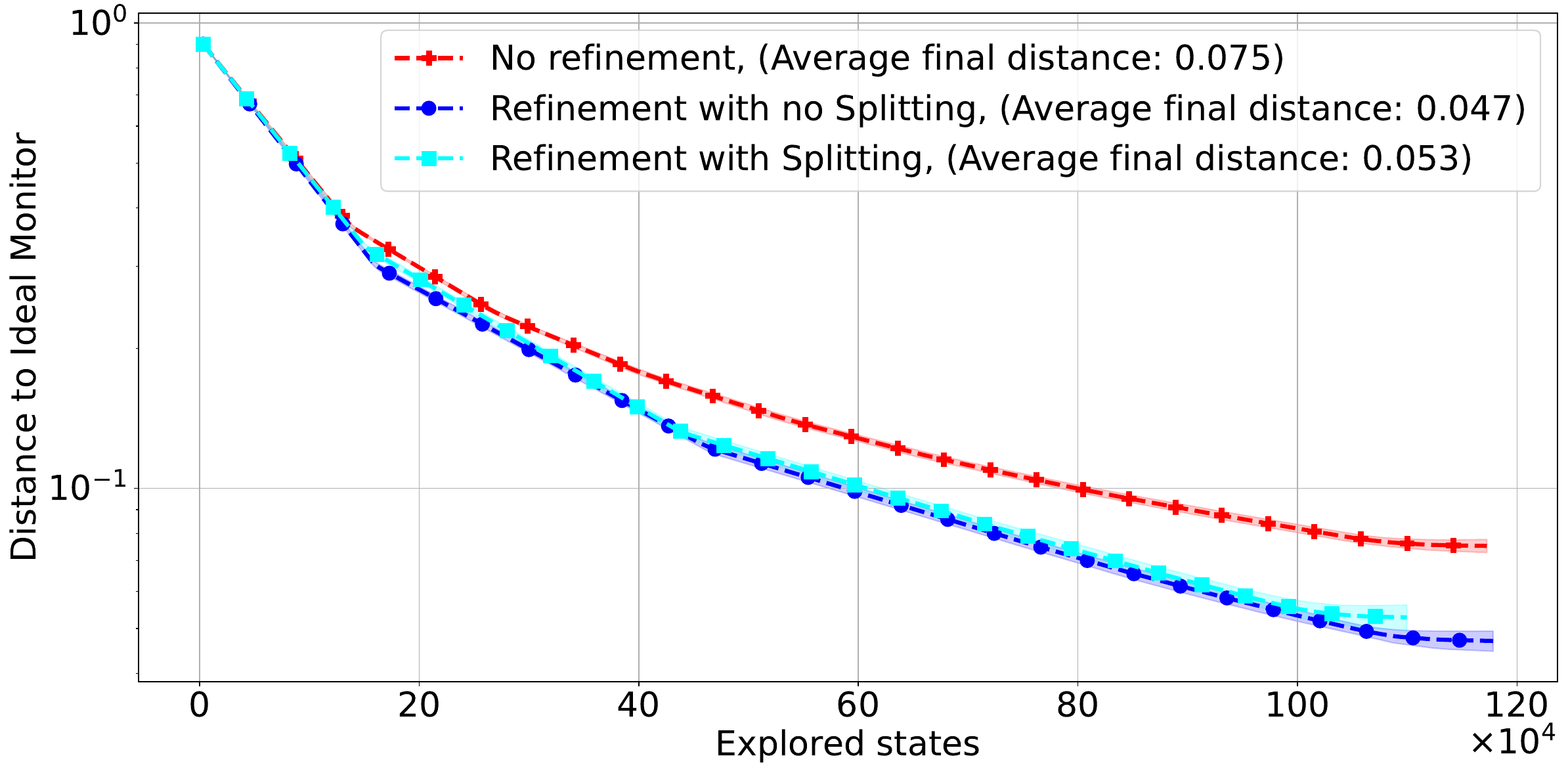}
        \caption{Comparing learning with and without refinement, in terms of distance to IM, coarse evadeV-6-3, SC =  0.01}
        \label{ref_no_ref_evadeV-6_coarse_distance}        
\end{figure}

\begin{figure}[H]
        \includegraphics[width=0.8\linewidth]{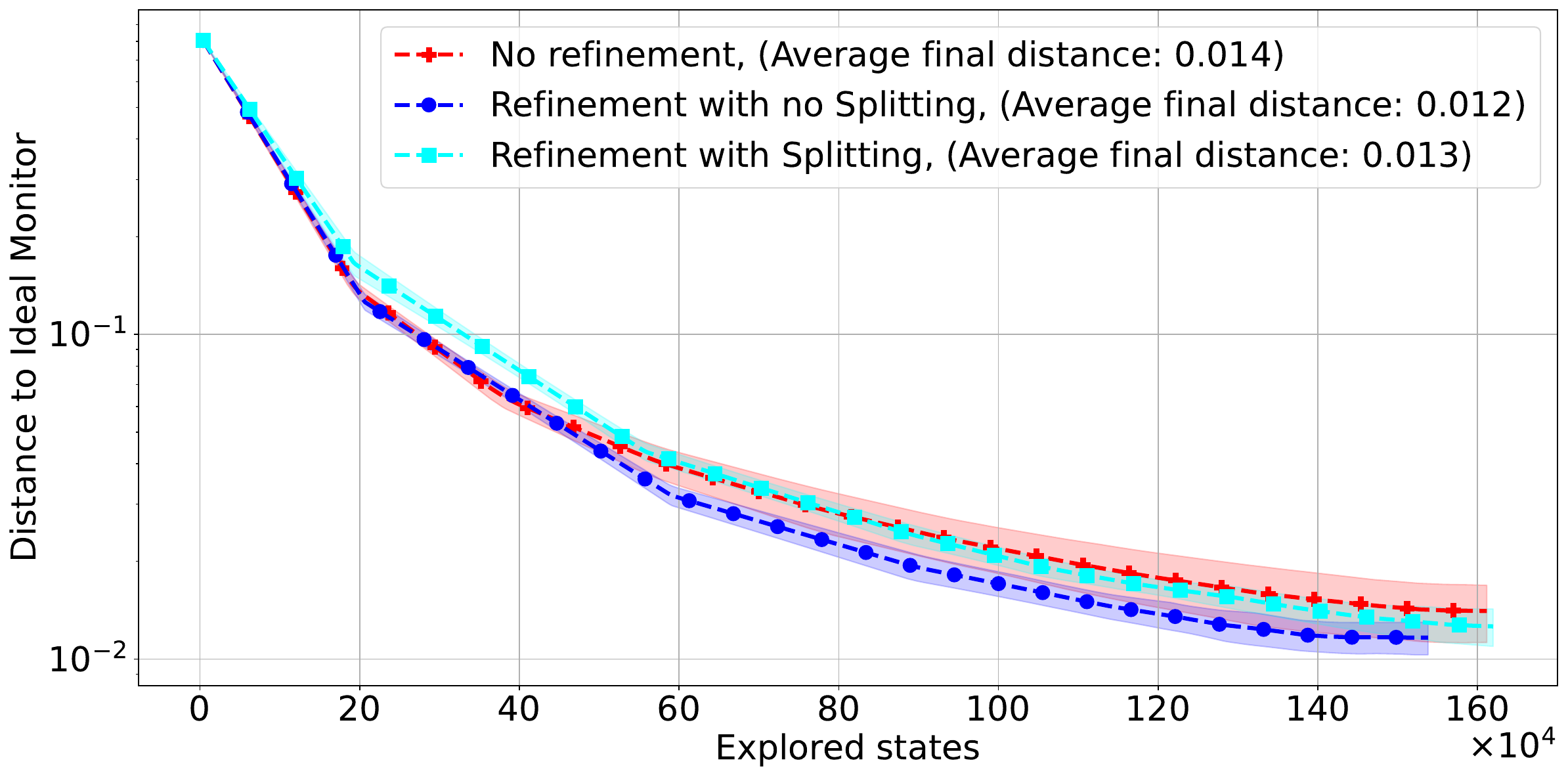}
        \caption{Comparing learning with and without refinement, in terms of distance to IM, airportA-7-10-10, SC = 0.001}
        \label{ref_no_ref_airportA-7-10-10_distance}
\end{figure}

\begin{figure}[H]
        \includegraphics[width=0.8\linewidth]{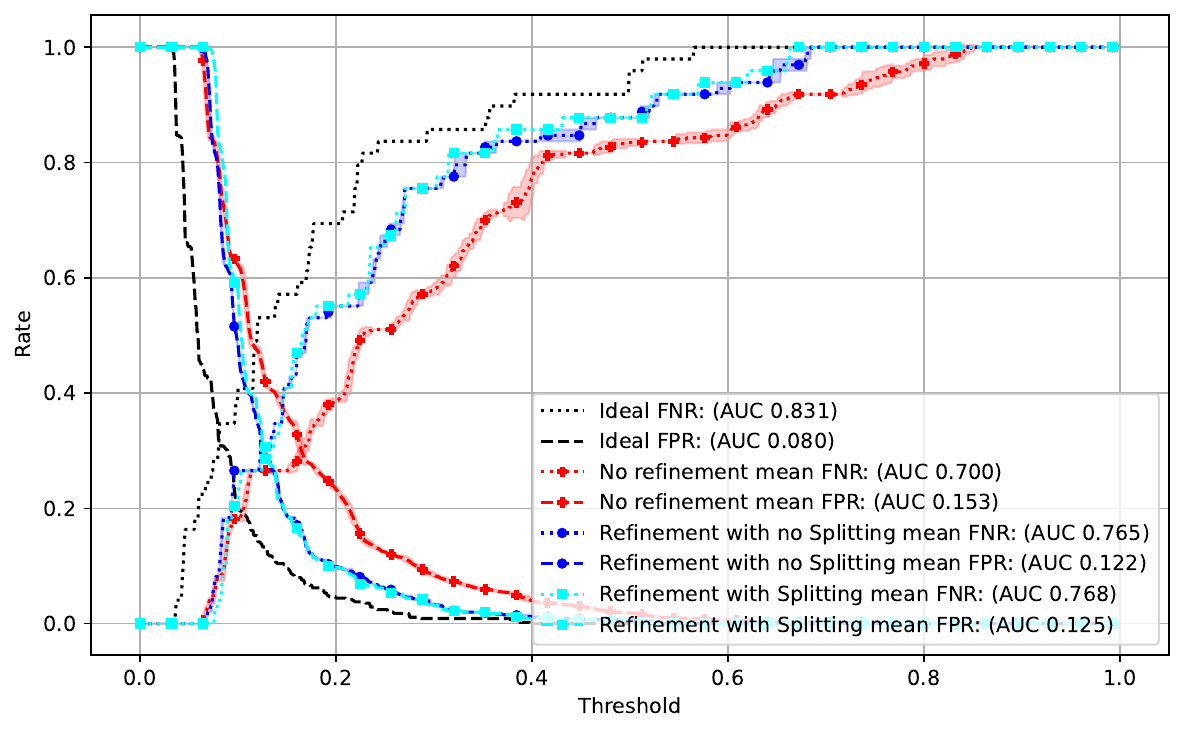}
        \caption{Comparing learning with and without refinement, in terms of FNR and FPR, coarse evadeV-6-3, SC= 0.01}
        \label{ref_no_ref_evadeV-6_fn_fp}
\end{figure}

\begin{figure}[H]
        \includegraphics[width=0.8\linewidth]{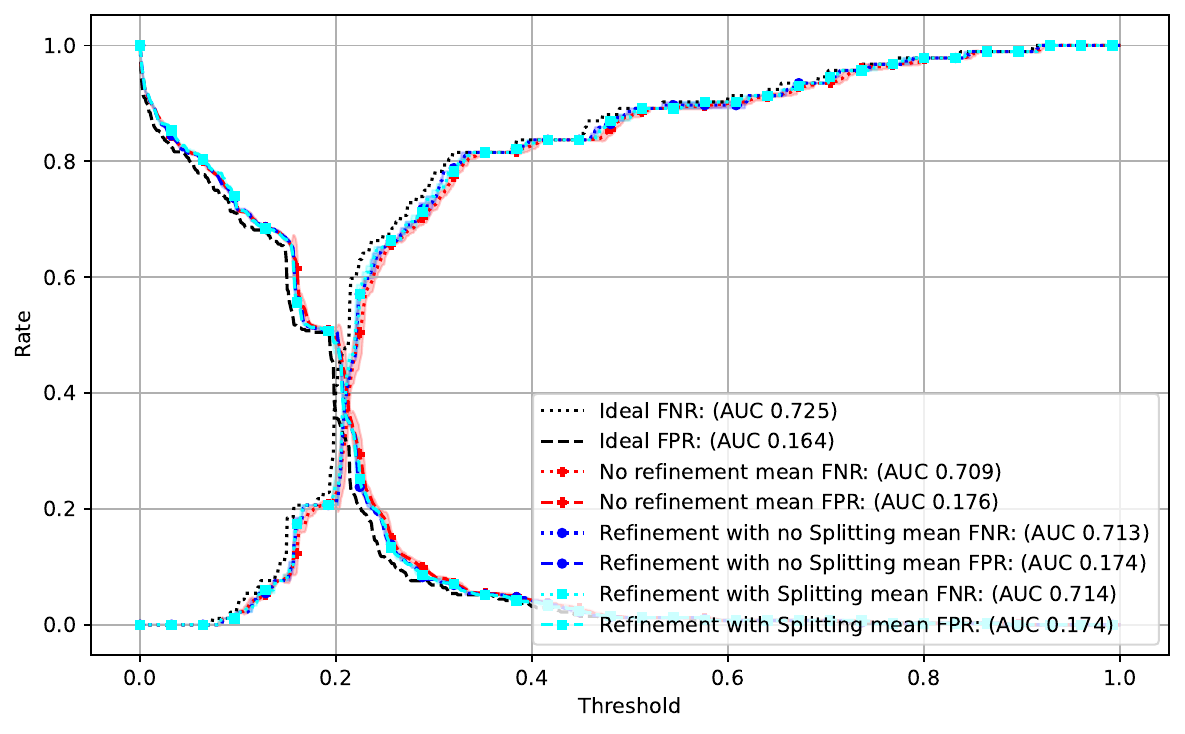}
        \caption{Comparing learning with and without refinement, in terms of FNR and FPR, airportA-7-10-10, SC= 0.001}
        \label{ref_no_ref_airportA-7-10-10_fn_fp}
\end{figure}

%% file: appendix.tex
\subsection{Model-based vs Model-free monitoring - Remaining experimental results}
\label{apndx:modelvsfree}

For each of the benchmarks not included in \Cref{sec:model_based_free}, we add corresponding analysis in terms of comparison of FNR and FPR across different thresholds between model-based and model-free monitors. The experimental evaluation of monitors for benchmarks airportA-7-40-20 and coarse airportB-7-40-20 is included in \Cref{apx:big_benchmarks}.

\begin{figure}[H]
    \centering
    \includegraphics[width=0.65\linewidth]{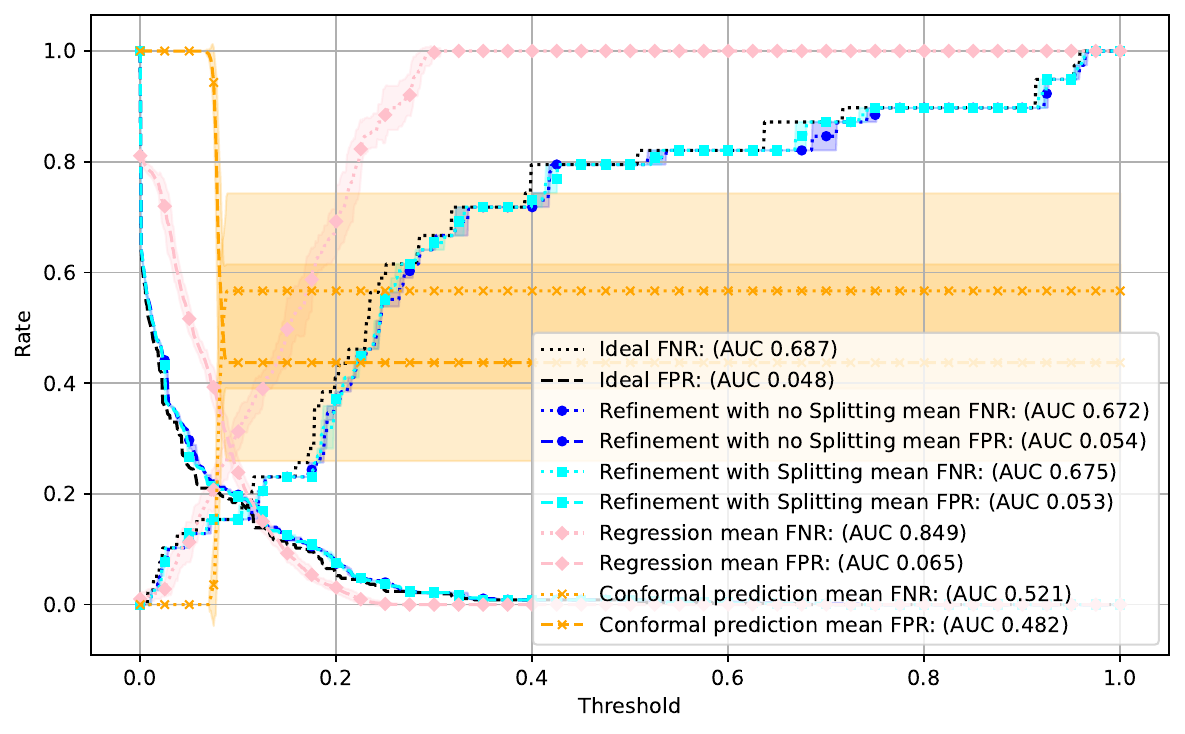}
    \caption{SnL-10x10, SC = 0.001 - FNR and FPR comparison between model-based and model-free methods}
    \label{model_free_model_based_SnL-10x10}
\end{figure}

\begin{figure}[H]
    \centering
    \includegraphics[width=0.65\linewidth]{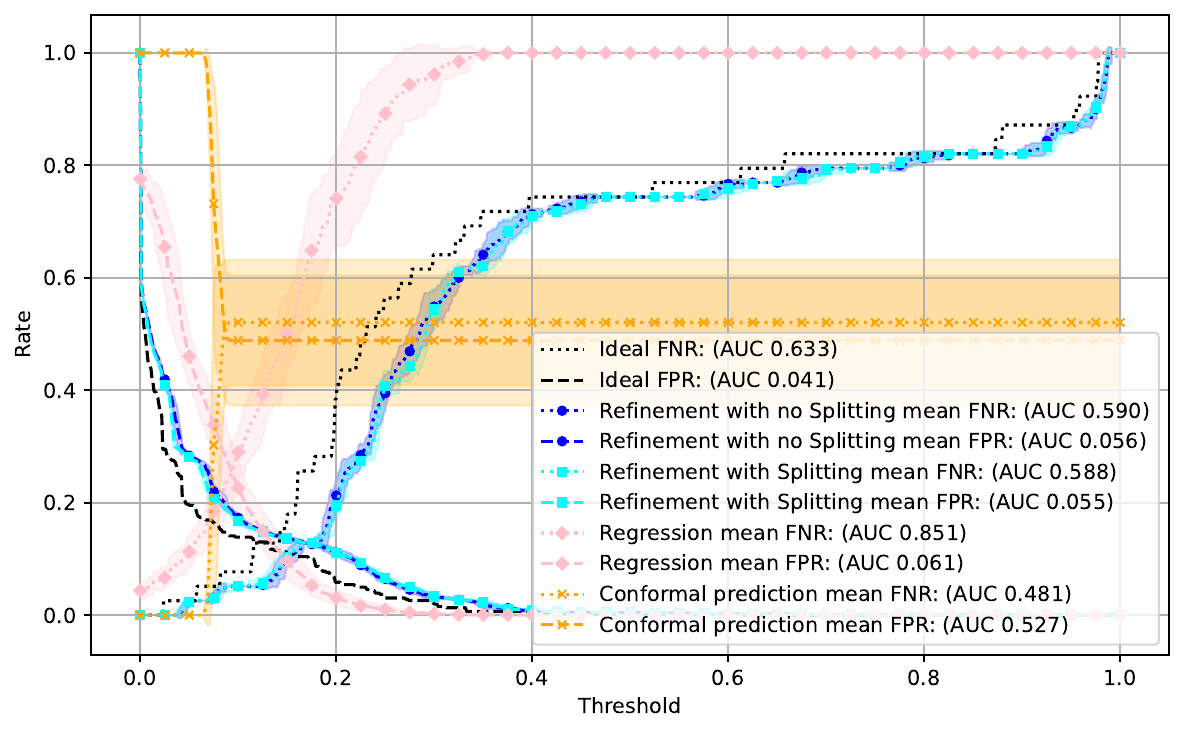}
    \caption{SnL-10x10, SC = 0.01 - FNR and FPR comparison between model-based and model-free methods}
    \label{model_free_model_based_SnL-10x10_high_st}
\end{figure}

\begin{figure}[H]
    \centering
    \includegraphics[width=0.65\linewidth]{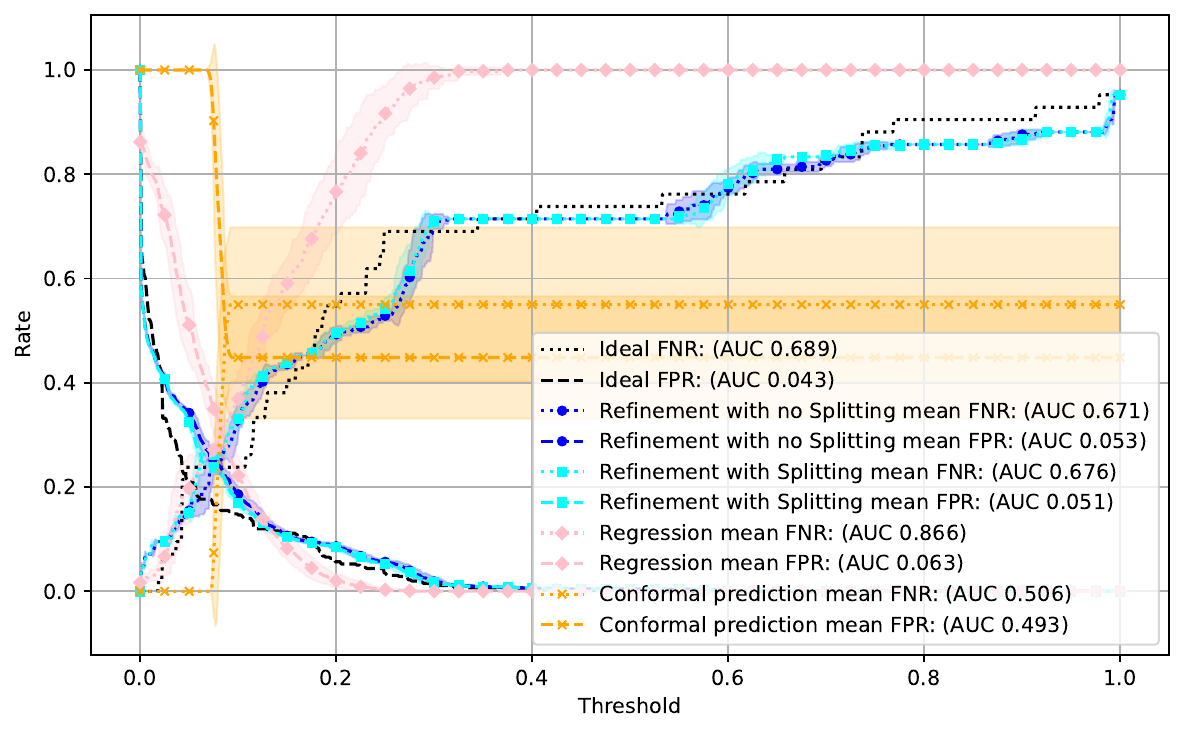}
    \caption{Coarse SnL-10x10, SC = 0.01 - FNR and FPR comparison between model-based and model-free methods}
    \label{model_free_model_based_SnL-10x10_coarse}
\end{figure}

\begin{figure}[H]
    \centering
    \includegraphics[width=0.65\linewidth]{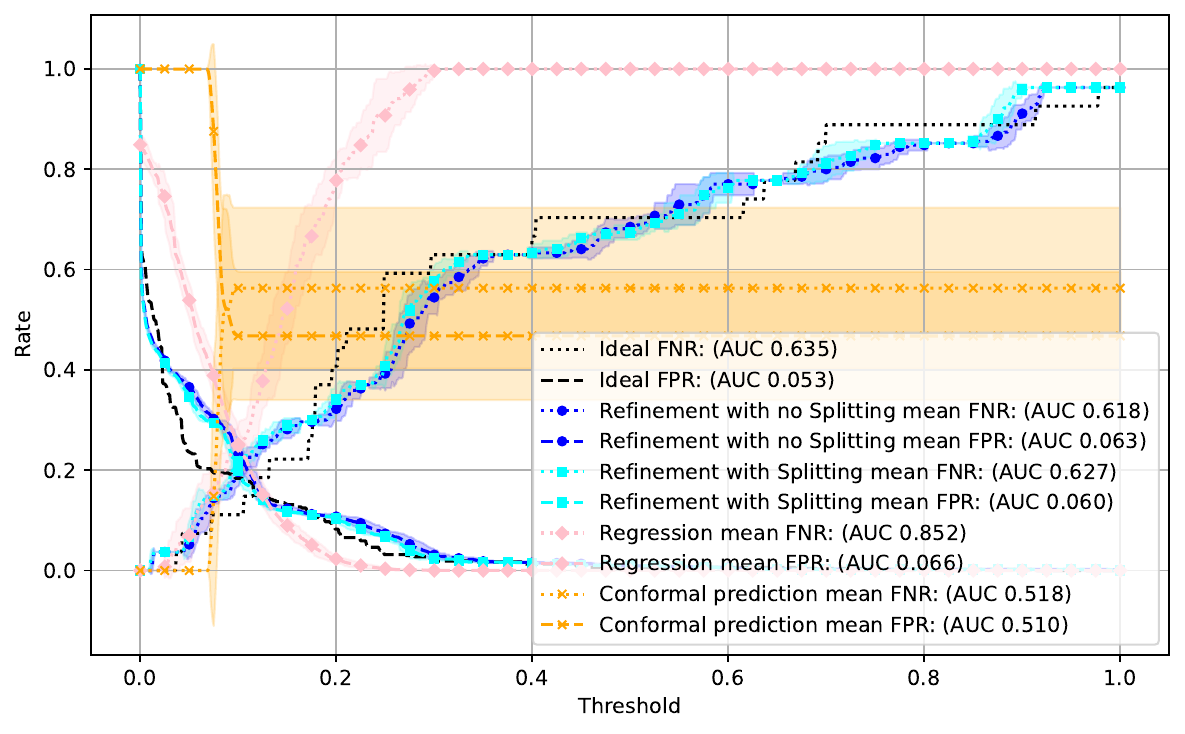}
    \caption{Coarse SnL-10x10, SC = 0.1 - FNR and FPR comparison between model-based and model-free methods}
    \label{model_free_model_based_SnL-10x10_coarse_high_st}
\end{figure}

\begin{figure}[H]
    \centering
    \includegraphics[width=0.65\linewidth]{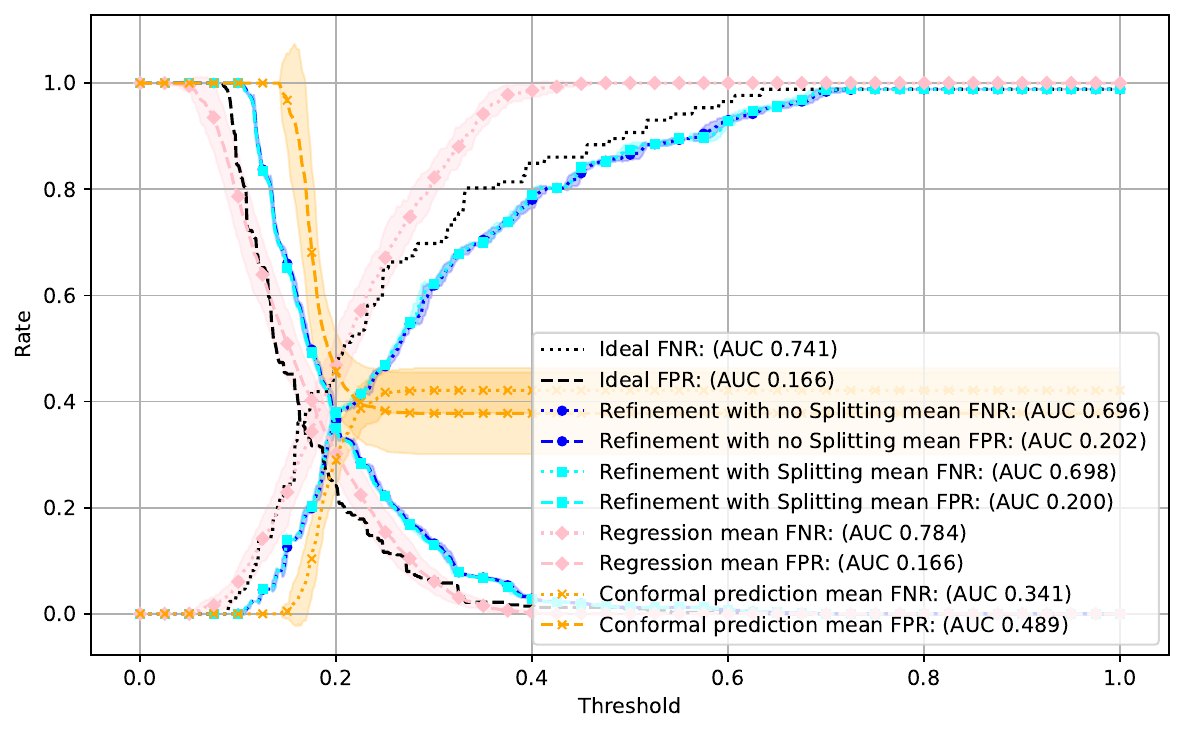}
    \caption{evadeV-5-3, SC = 0.01 - FNR and FPR comparison between model-based and model-free methods}
    \label{model_free_model_based_evadeV-5-3}
\end{figure}

\begin{figure}[H]
    \centering
    \includegraphics[width=0.65\linewidth]{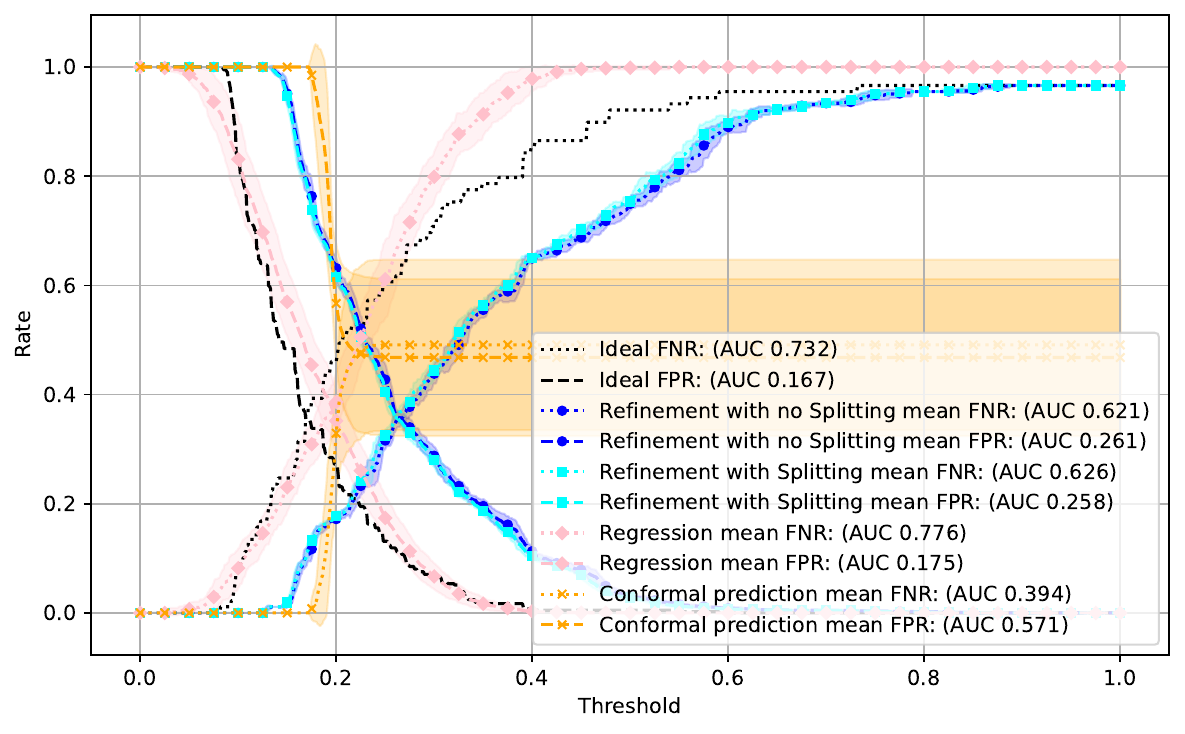}
    \caption{evadeV-5-3, SC = 0.1 - FNR and FPR comparison between model-based and model-free methods}
    \label{model_free_model_based_evadeV-5-3_high_st}
\end{figure}

\begin{figure}[H]
    \centering
    \includegraphics[width=0.65\linewidth]{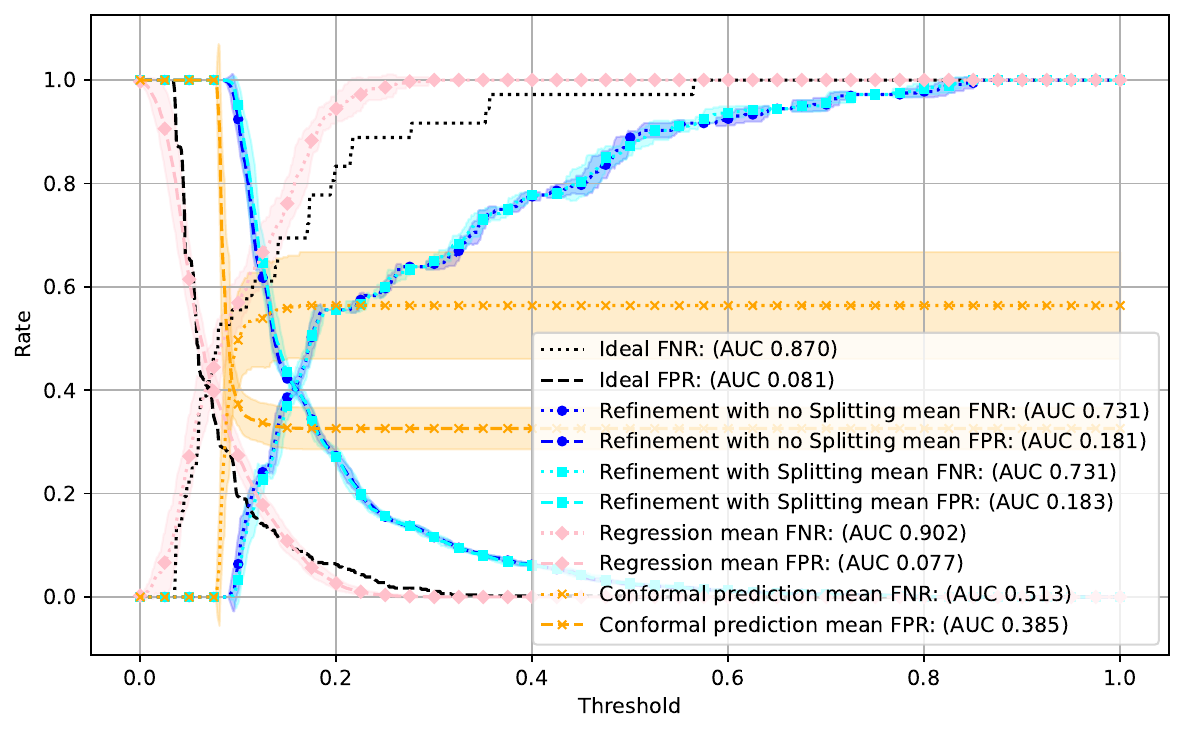}
    \caption{evadeV-6-3, SC = 0.01 - FNR and FPR comparison between model-based and model-free methods}
    \label{model_free_model_based_evadeV-6-3}
\end{figure}

\begin{figure}[H]
    \centering
    \includegraphics[width=0.65\linewidth]{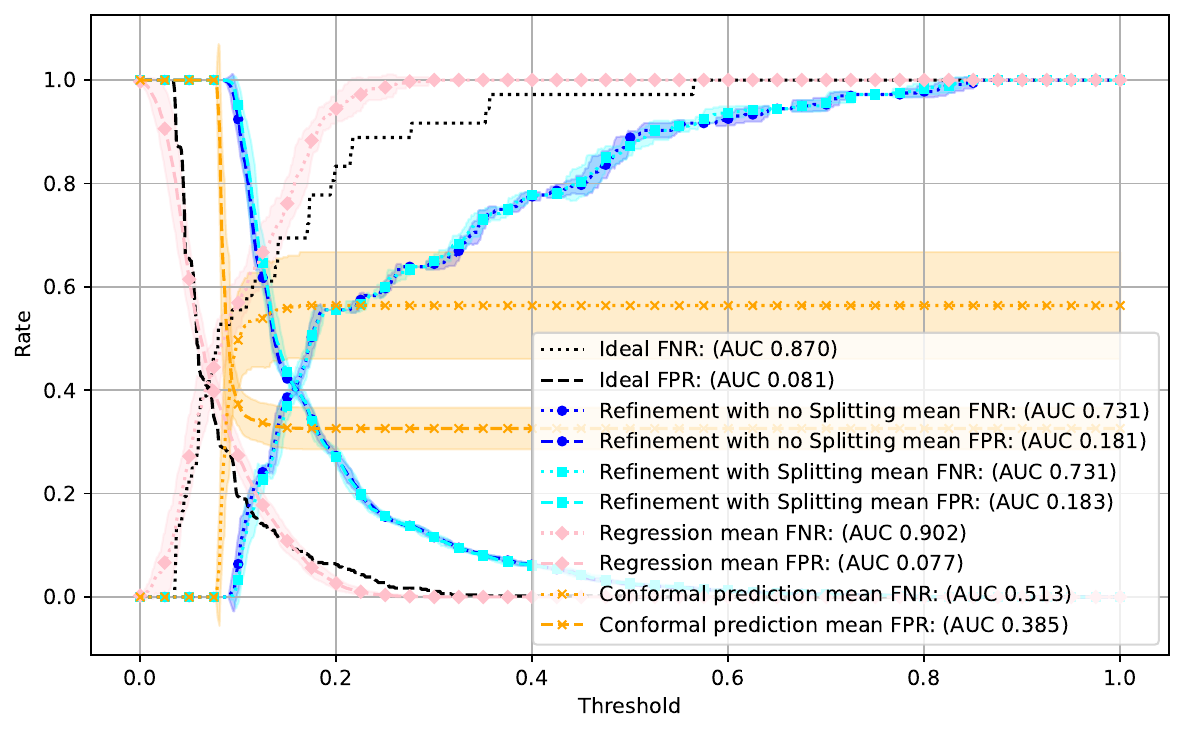}
    \caption{evadeV-6-3, SC =  0.1 - FNR and FPR comparison between model-based and model-free methods}
    \label{model_free_model_based_evadeV-6-3_high_st}
\end{figure}

\begin{figure}[H]
    \centering
    \includegraphics[width=0.65\linewidth]{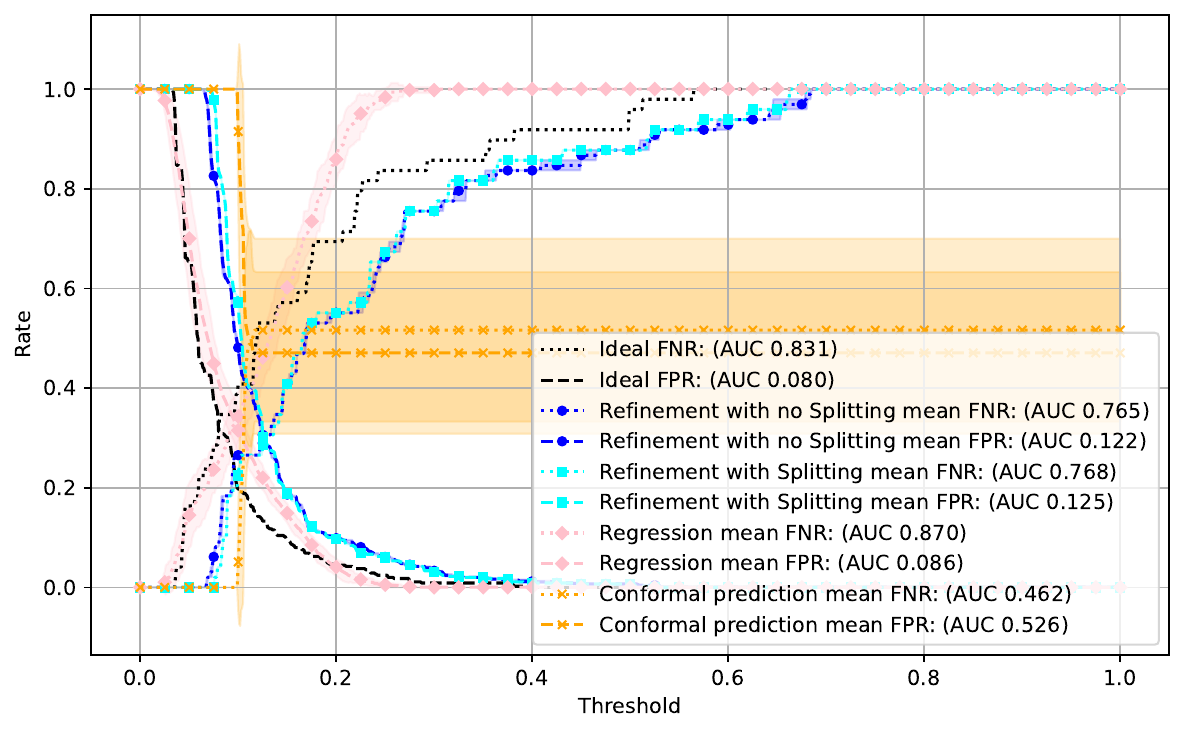}
    \caption{Coarse evadeV-6-3, SC = 0.01 - FNR and FPR comparison between model-based and model-free methods}
    \label{model_free_model_based_evadeV-6-3_coarse}
\end{figure}

\begin{figure}[H]
    \centering
    \includegraphics[width=0.65\linewidth]{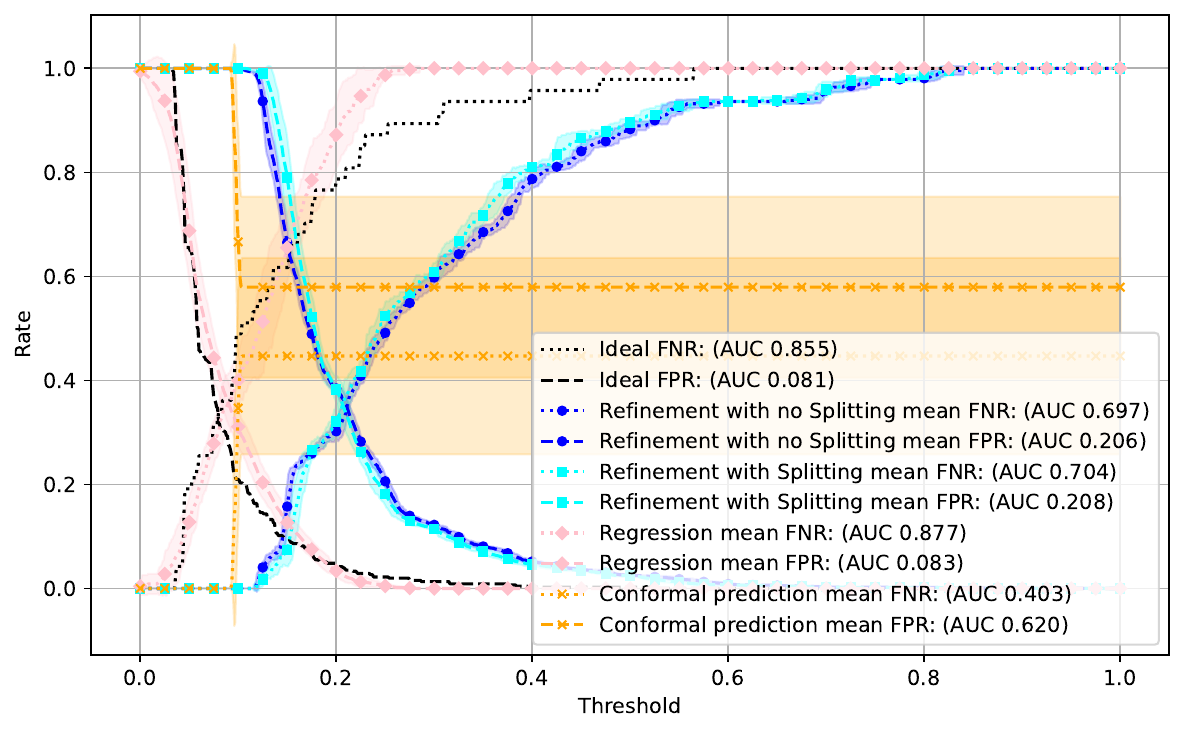}
    \caption{Coarse evadeV-6-3, SC = 0.1 - FNR and FPR comparison between model-based and model-free methods}
    \label{model_free_model_based_evadeV-6-3_coarse_high_st}
\end{figure}

\begin{figure}[H]
    \centering
    \includegraphics[width=0.65\linewidth]{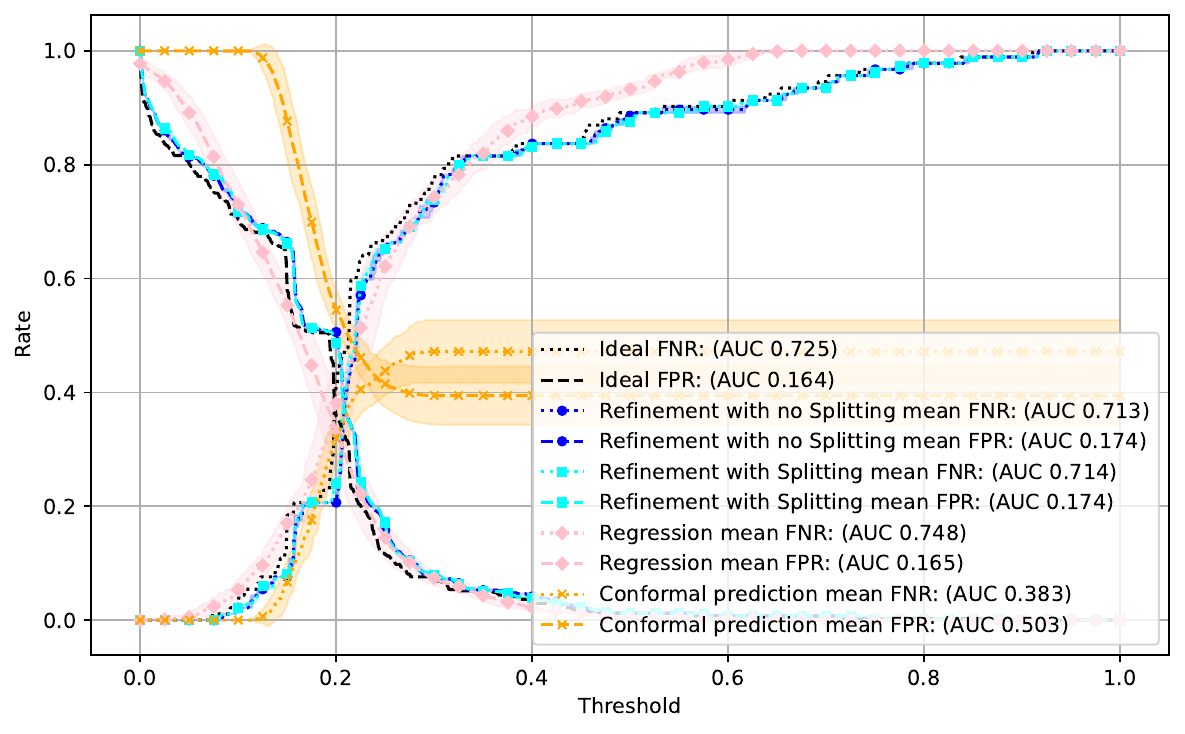}
    \caption{airportA-7-10-10, SC = 0.001 - FNR and FPR comparison between model-based and model-free methods}
\end{figure}

\begin{figure}[H]
    \centering
    \includegraphics[width=0.65\linewidth]{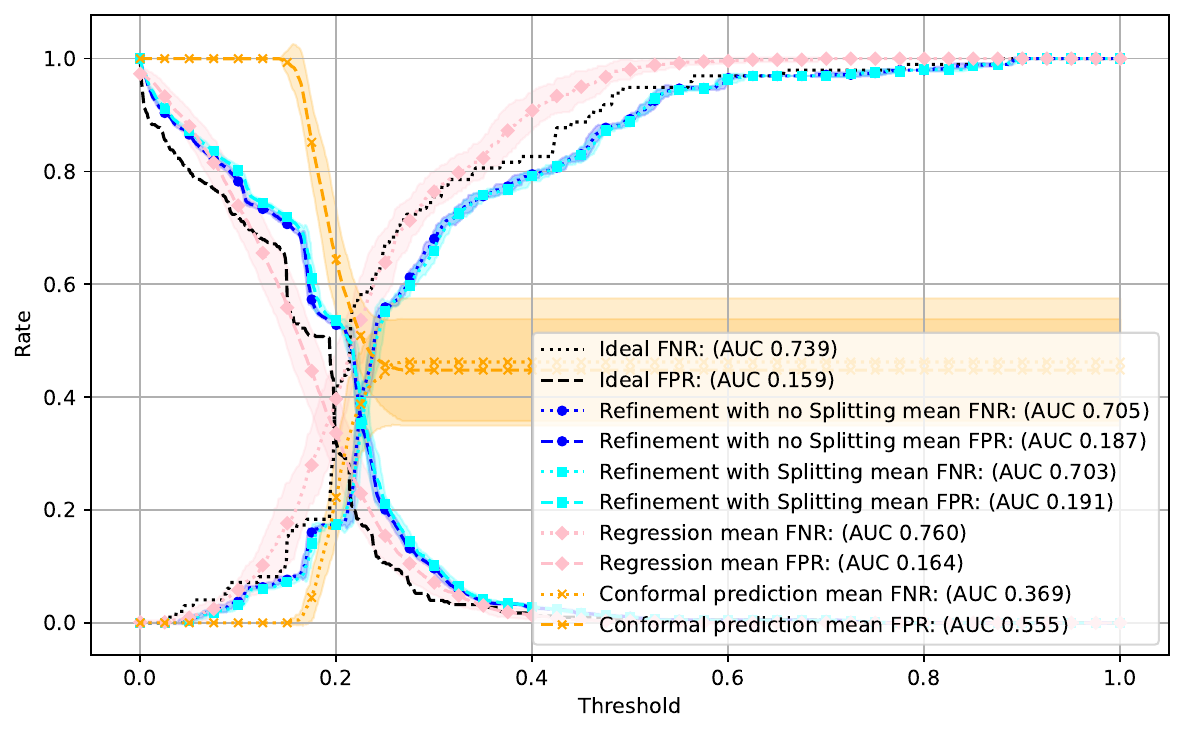}
    \caption{airportA-7-10-10, SC = 0.01 - FNR and FPR comparison between model-based and model-free methods}
\end{figure}

\begin{figure}[H]
    \centering
    \includegraphics[width=0.65\linewidth]{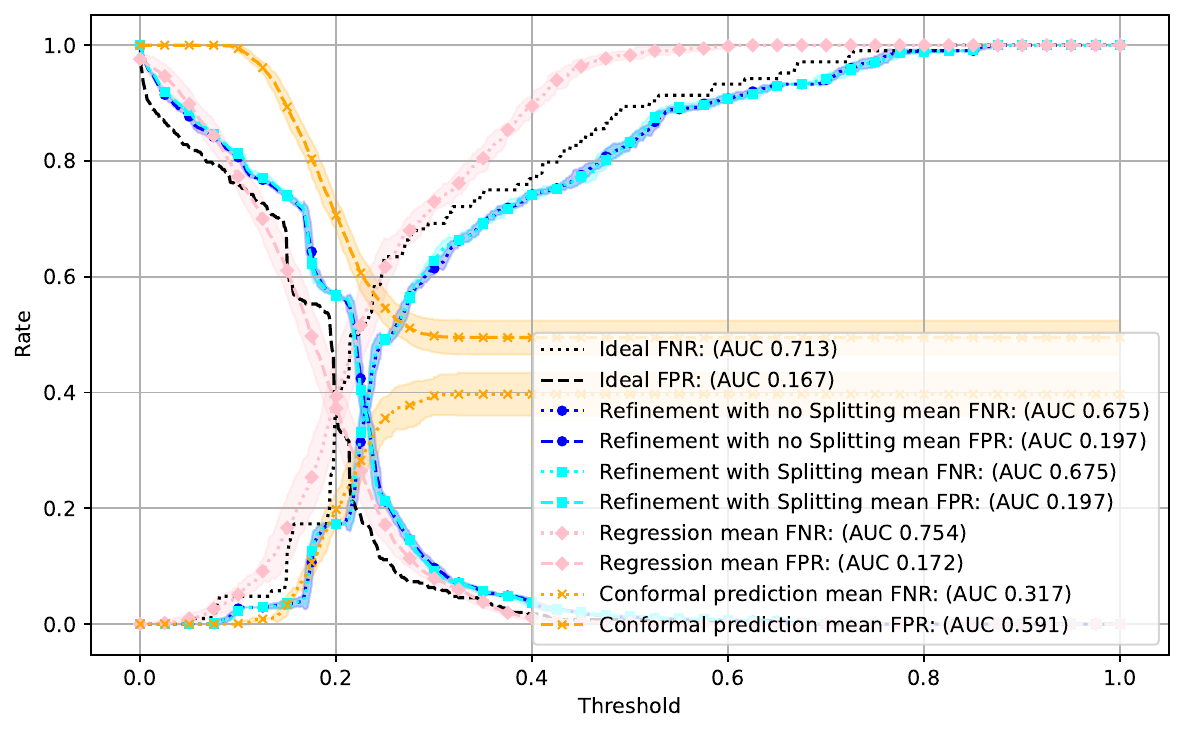}
    \caption{Coarse airportA-7-10-10, SC = 0.01 - FNR and FPR comparison between model-based and model-free methods}
\end{figure}

\begin{figure}[H]
    \centering
    \includegraphics[width=0.65\linewidth]{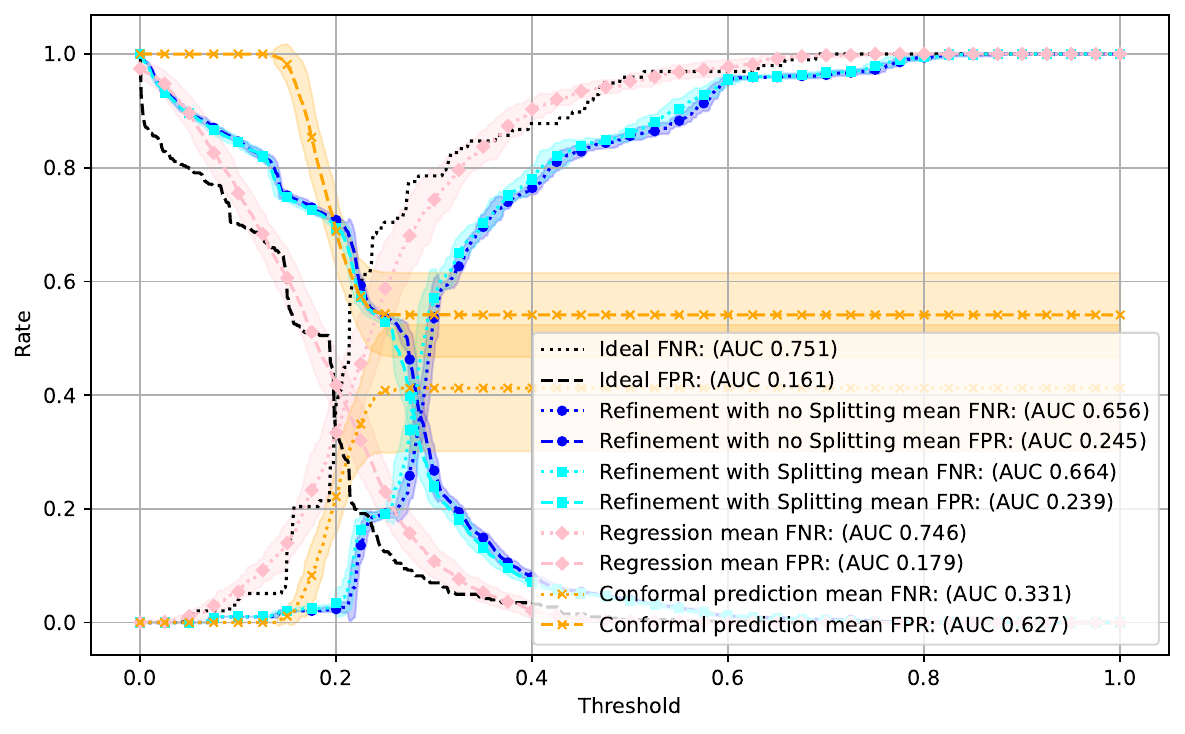}
    \caption{Coarse airportA-7-10-10, SC = 0.1 - FNR and FPR comparison between model-based and model-free methods}
\end{figure}

\subsection{iHMM vs HMM based monitoring}
\label{apndx:ihmm_hmm}

For each of the benchmarks not included in \Cref{sec:ihmm_vs_hmm_comp}, we add corresponding analysis in terms of comparison of FNR and FPR across different thresholds between iHMM and HMM monitors and classification of risks as under and over approximating. The experimental evaluation of monitors for benchmarks airportA-7-40-20 and coarse airportB-7-40-20 is included in \Cref{apx:big_benchmarks}.

\begin{figure}[H]
    \centering
    \includegraphics[width=0.65\linewidth]{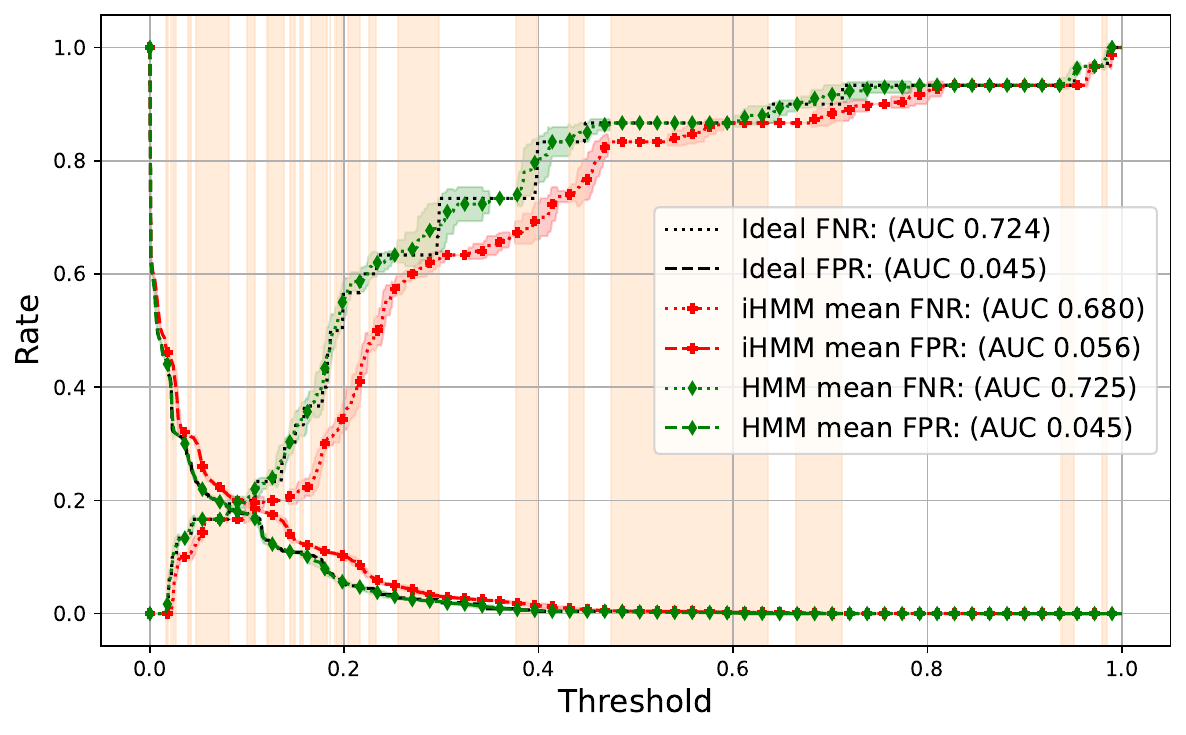}
    \caption{SnL-10x10, SC = 0.001 - FNR and FPR comparison between iHMM and HMM}
    \label{ihmm_hmm_SnL-10x10 }
\end{figure}

\begin{figure}[H]
    \centering
    \includegraphics[width=0.65\linewidth]{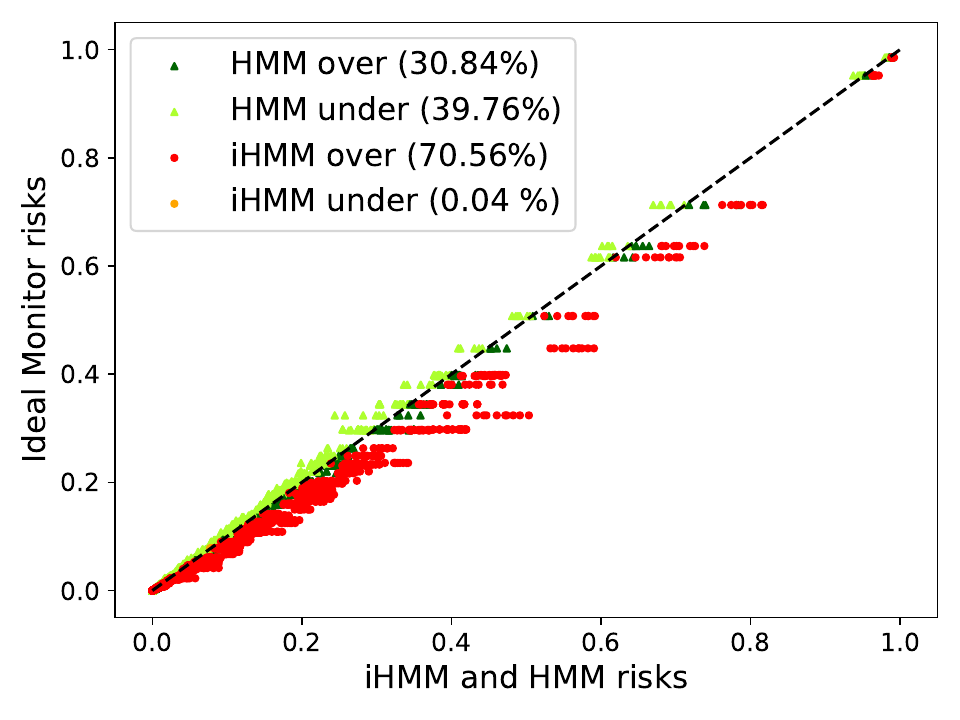}
    \caption{SnL-10x10, SC = 0.001 - risk estimation subject to target, comparison between iHMM and HMM}
    \label{ihmm_hmm_overestimation_SnL-10x10 }
\end{figure}

\begin{figure}[H]
    \centering
    \includegraphics[width=0.65\linewidth]{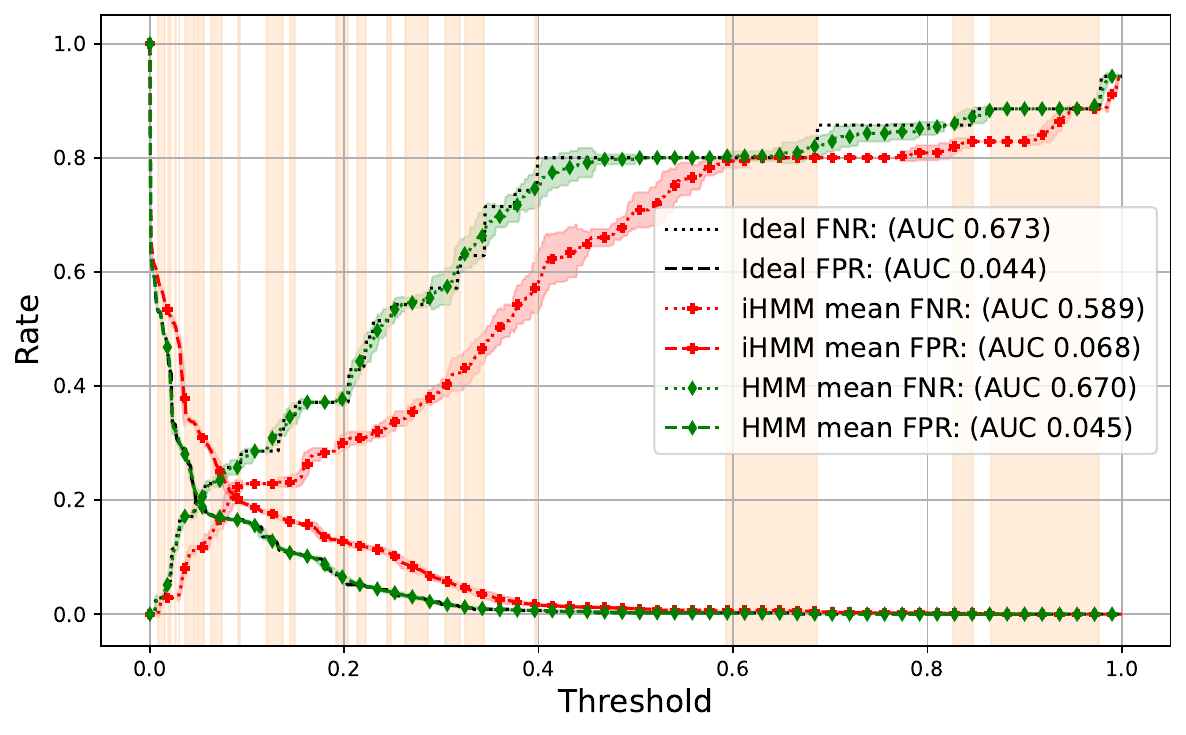}
    \caption{SnL-10x10, SC = 0.01 - FNR and FPR comparison between iHMM and HMM}
    \label{ihmm_hmm_SnL-10x10 }
\end{figure}

\begin{figure}[H]
    \centering
    \includegraphics[width=0.65\linewidth]{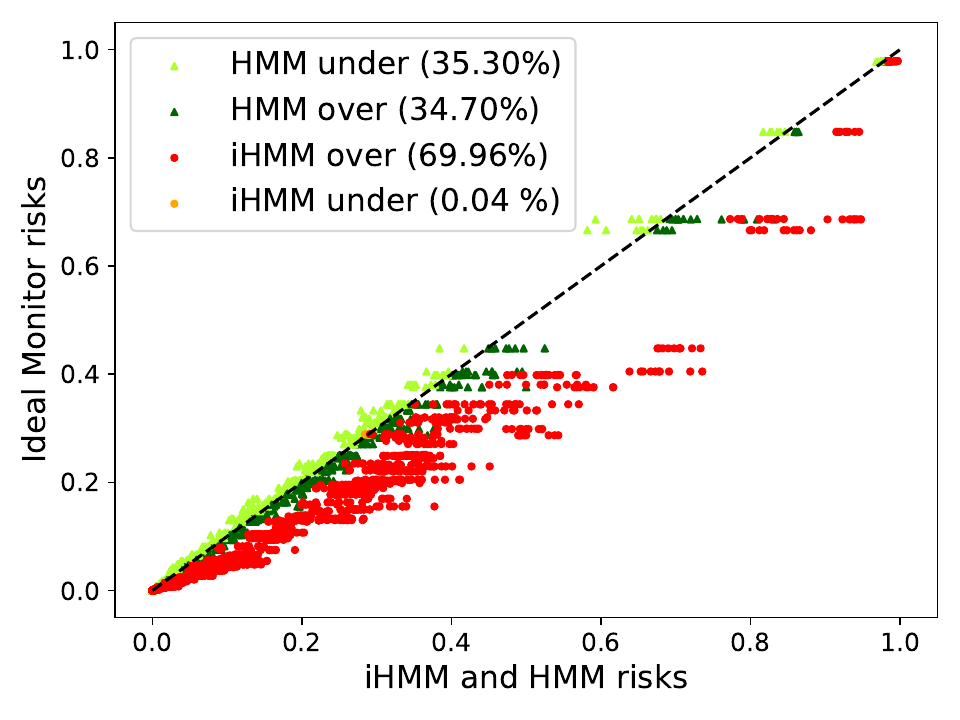}
    \caption{SnL-10x10, SC = 0.01 - risk estimation subject to target, comparison between iHMM and HMM}
    \label{ihmm_hmm_overestimation_SnL-10x10 }
\end{figure}

\begin{figure}[H]
    \centering
    \includegraphics[width=0.65\linewidth]{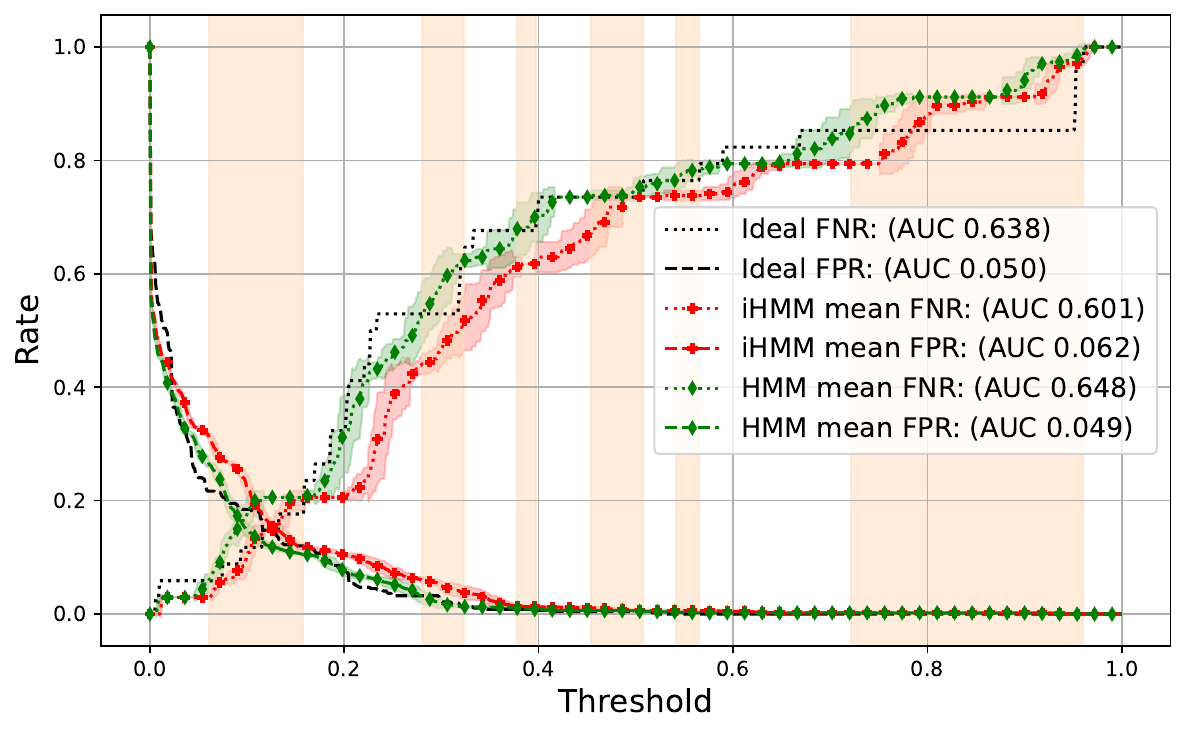}
    \caption{Coarse SnL-10x10, SC = 0.01 - FNR and FPR comparison between iHMM and HMM}
    \label{ihmm_hmm_SnL-10x10 }
\end{figure}

\begin{figure}[H]
    \centering
    \includegraphics[width=0.65\linewidth]{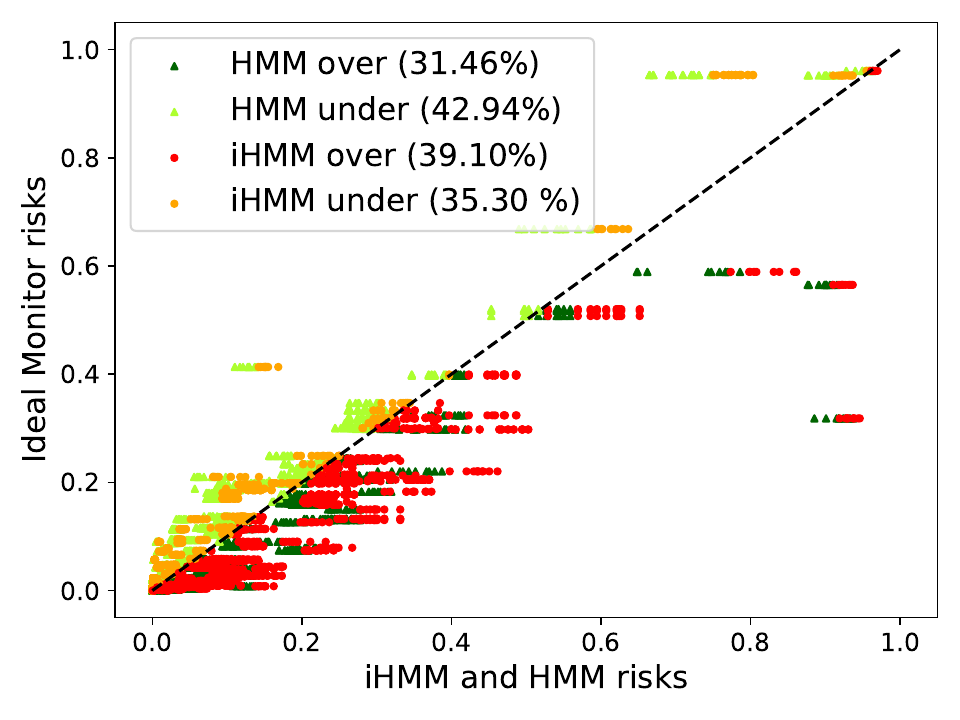}
    \caption{Coarse SnL-10x10, SC = 0.01 - risk estimation subject to target, comparison between iHMM and HMM}
    \label{ihmm_hmm_overestimation_SnL-10x10 }
\end{figure}

\begin{figure}[H]
    \centering
    \includegraphics[width=0.65\linewidth]{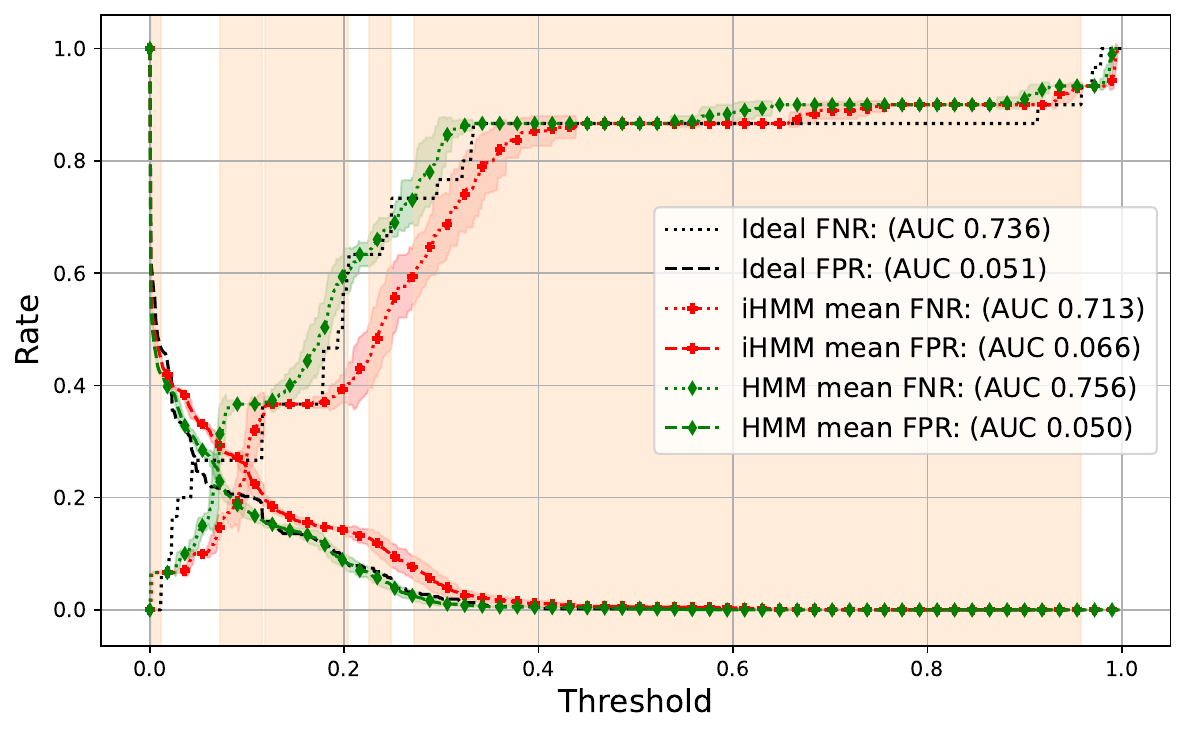}
    \caption{Coarse SnL-10x10, SC = 0.1 - FNR and FPR comparison between iHMM and HMM}
    \label{ihmm_hmm_SnL-10x10 }
\end{figure}

\begin{figure}[H]
    \centering
    \includegraphics[width=0.65\linewidth]{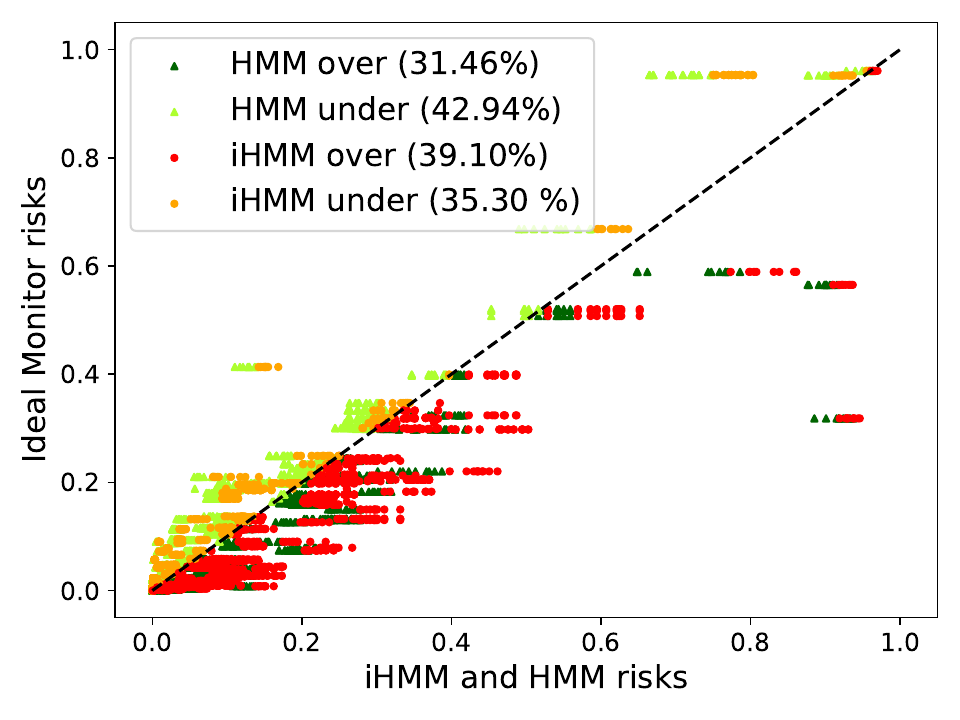}
    \caption{Coarse SnL-10x10, SC = 0.1 - risk estimation subject to target, comparison between iHMM and HMM}
    \label{ihmm_hmm_SnL-10x10_corase}
\end{figure}

\begin{figure}[H]
    \centering
    \includegraphics[width=0.65\linewidth]{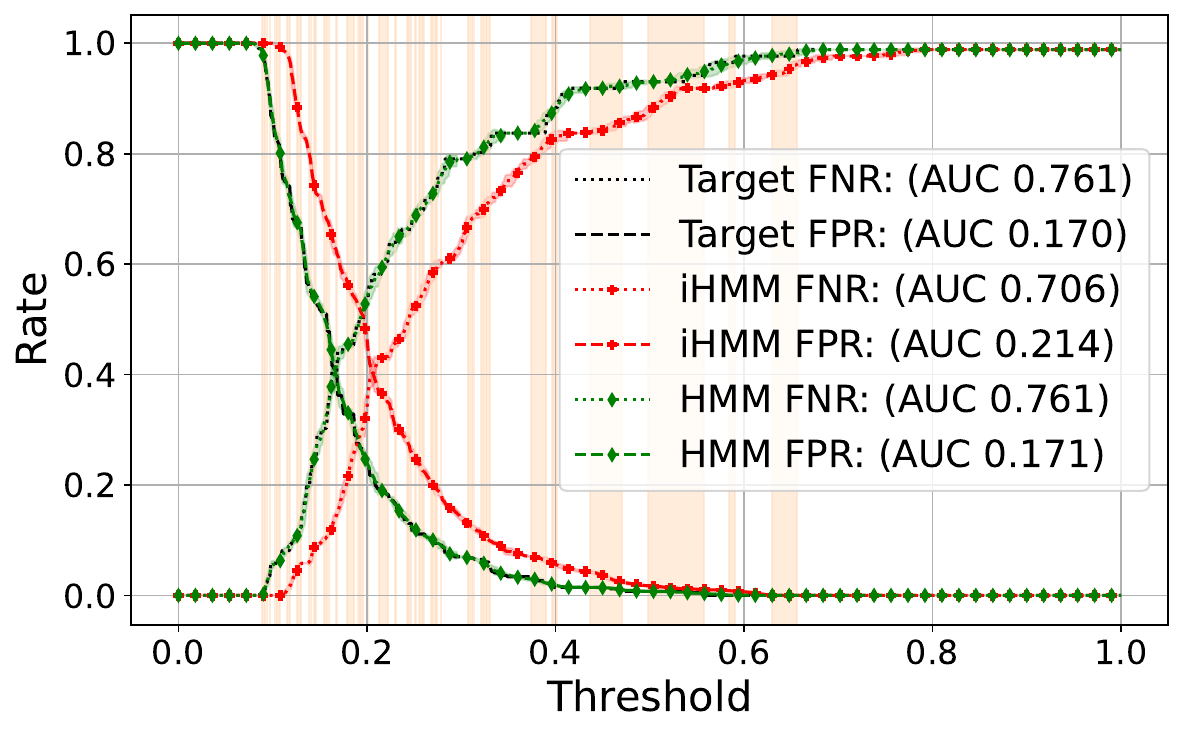}
    \caption{evadeV-5-3, SC = 0.01 - FNR and FPR comparison between iHMM and HMM}
    \label{ihmm_hmm_evadeV-5-3_high_st}
\end{figure}

\begin{figure}[H]
    \centering
    \includegraphics[width=0.64\linewidth]{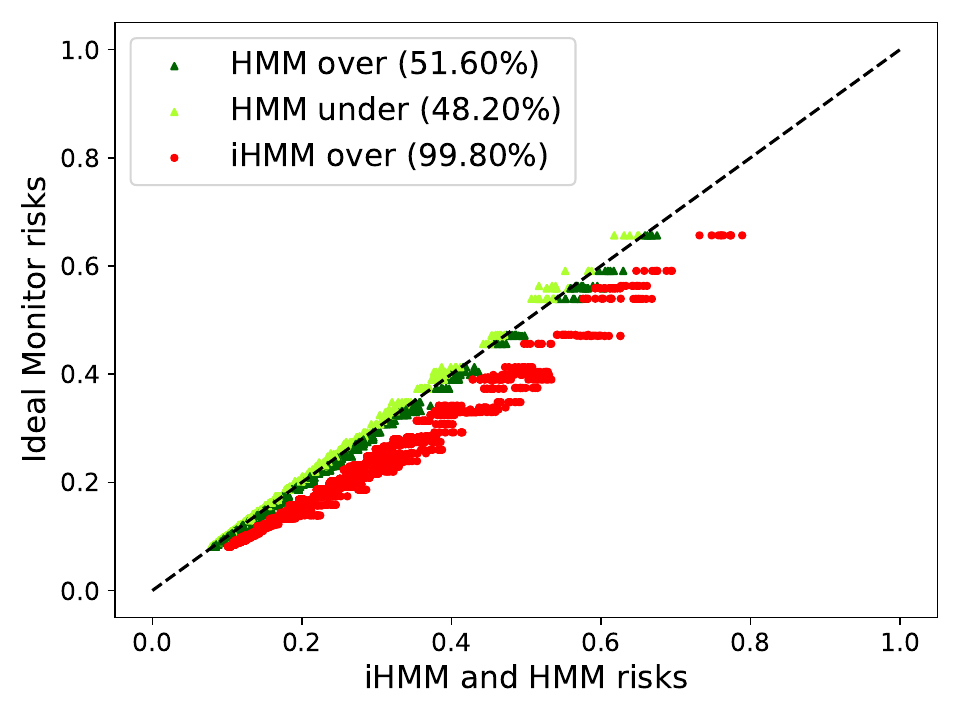}
    \caption{evadeV-5-3, SC = 0.01 - risk estimation subject to target, comparison between iHMM and HMM}
    \label{ihmm_hmm_overestimation_evadeV-5-3_high_st}
\end{figure}

\begin{figure}[H]
    \centering
    \includegraphics[width=0.65\linewidth]{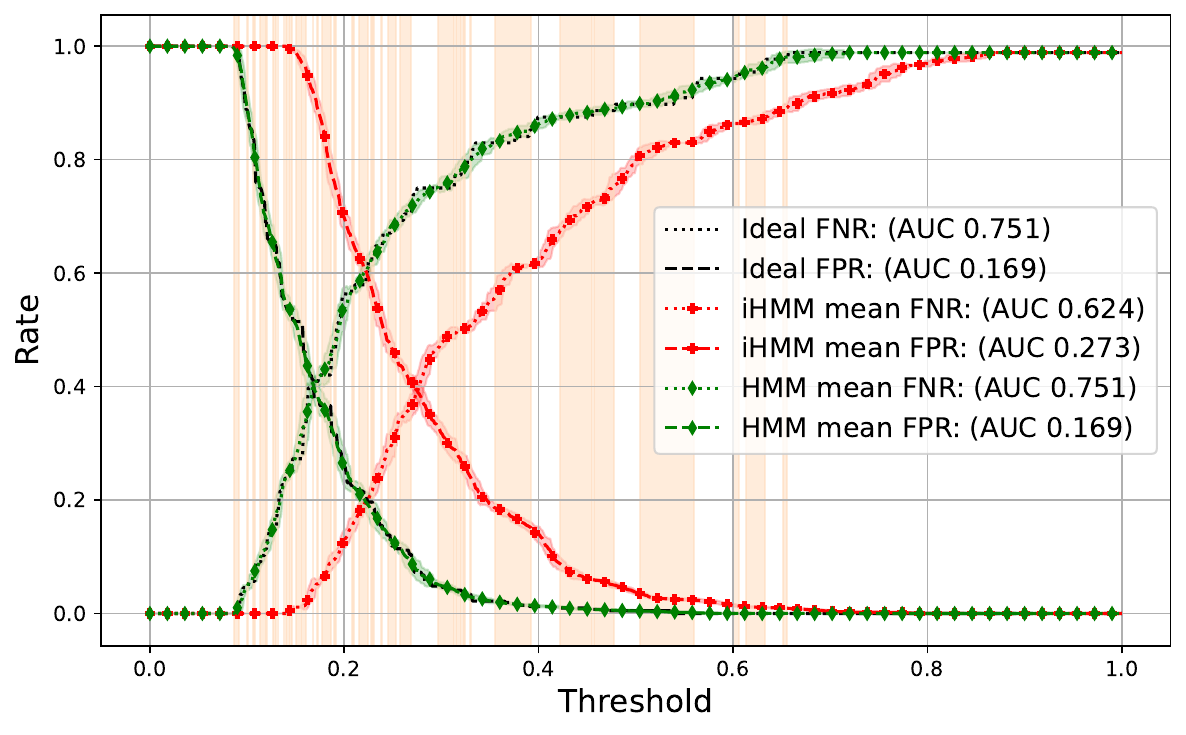}
    \caption{evadeV-5-3, SC = 0.1 - FNR and FPR comparison between iHMM and HMM}
    \label{ihmm_hmm_evadeV-5-3_high_st}
\end{figure}

\begin{figure}[H]
    \centering
    \includegraphics[width=0.64\linewidth]{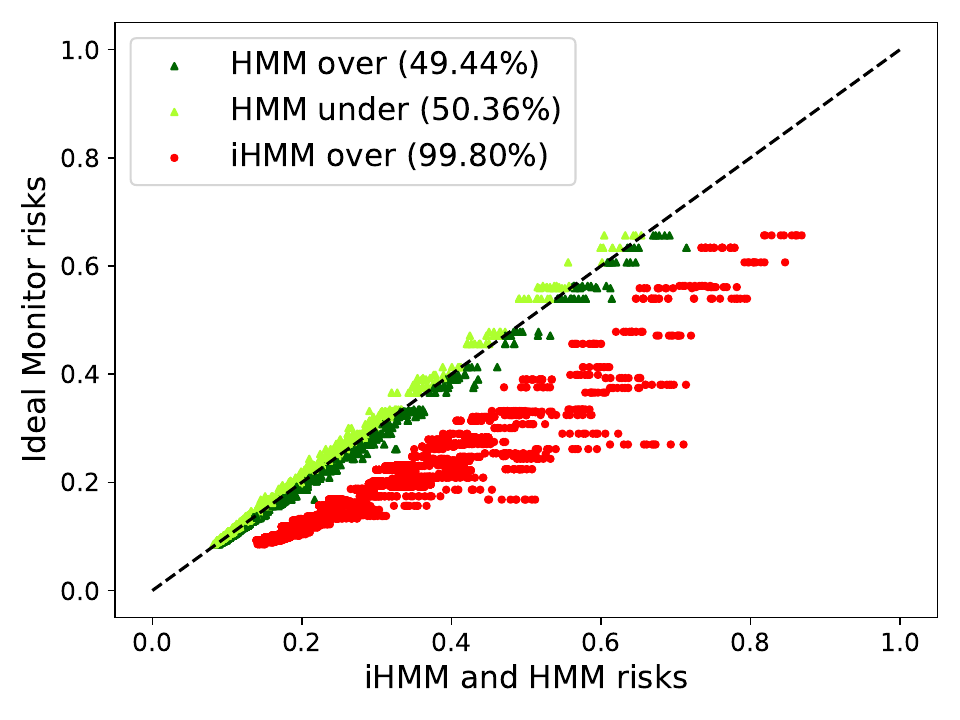}
    \caption{evadeV-5-3, SC = 0.1 - risk estimation subject to target, comparison between iHMM and HMM}
    \label{ihmm_hmm_overestimation_evadeV-5-3_high_st}
\end{figure}

\begin{figure}[H]
    \centering
    \includegraphics[width=0.64\linewidth]{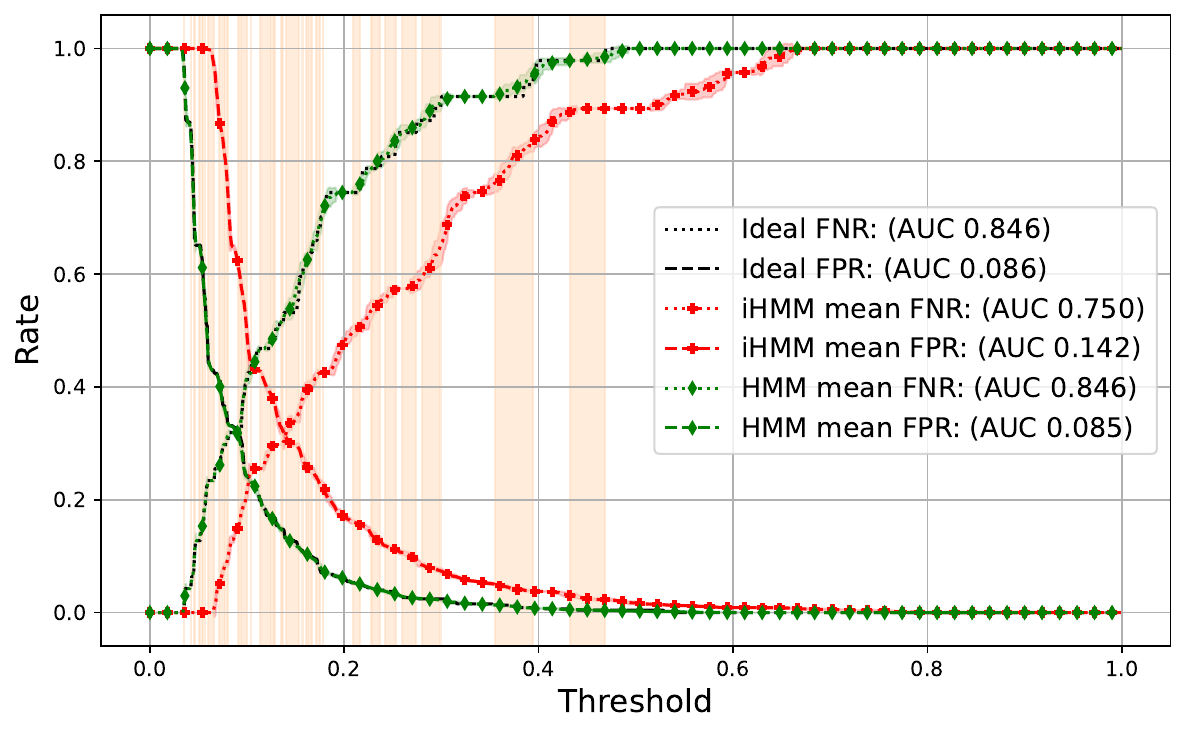}
    \caption{evadeV-6-3 , SC = 0.01 - FNR and FPR comparison between iHMM and HMM}
    \label{ihmm_hmm_evadeV-6-3}
\end{figure}

\begin{figure}[H]
    \centering
    \includegraphics[width=0.64\linewidth]{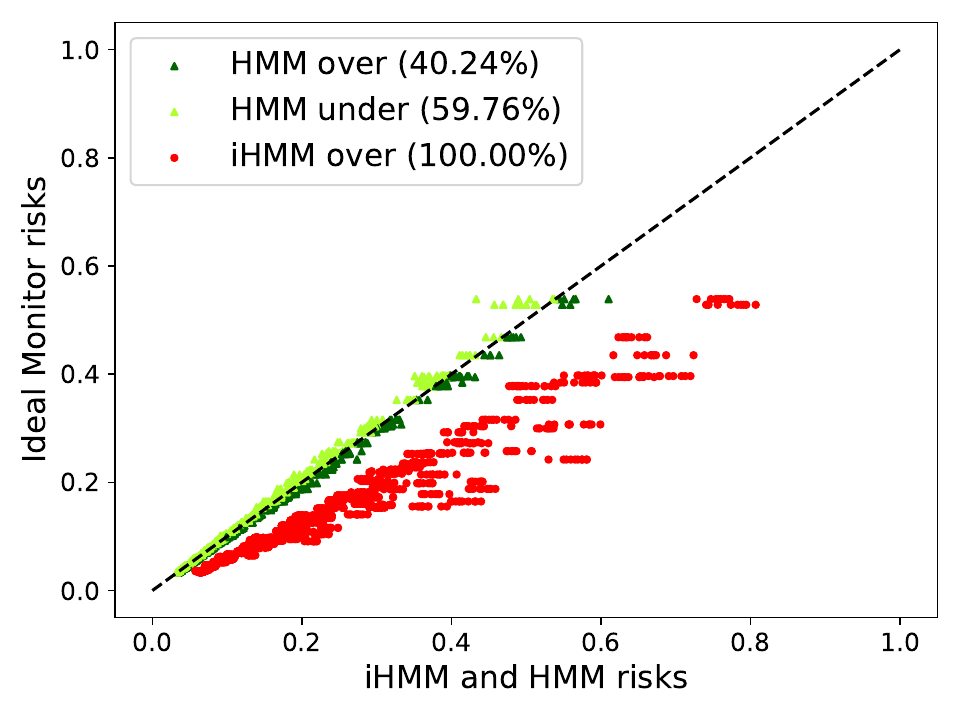}
    \caption{evadeV-6-3, SC =  0.01 - risk estimation subject to target, comparison between iHMM and HMM}
    \label{ihmm_hmm_overestimation_evadeV-6-3}
\end{figure}

\begin{figure}[H]
    \centering
    \includegraphics[width=0.64\linewidth]{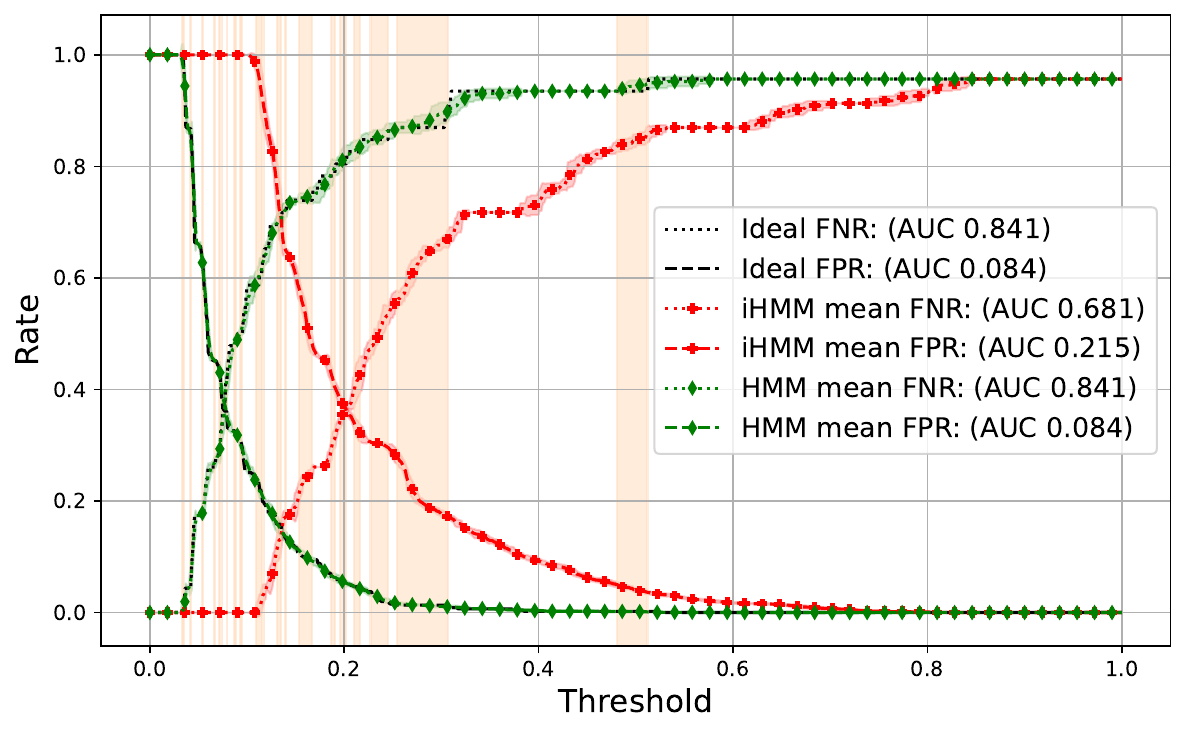}
    \caption{evadeV-6-3, SC = 0.1 - FNR and FPR comparison between iHMM and HMM}
    \label{ihmm_hmm_evadeV-6-3_high_st}
\end{figure}

\begin{figure}[H]
    \centering
    \includegraphics[width=0.64\linewidth]{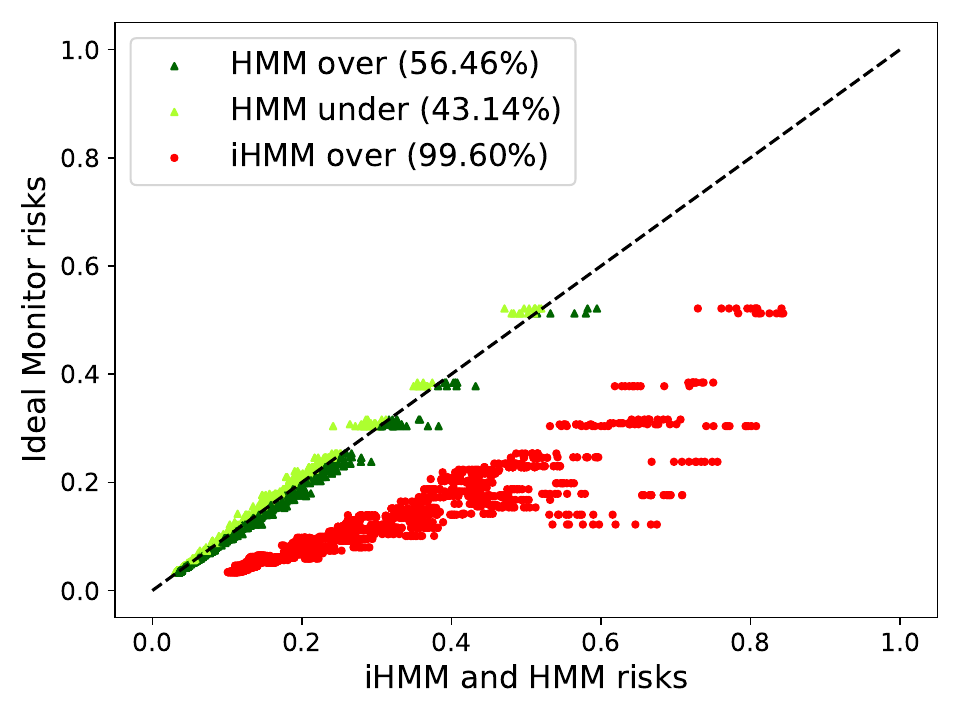}
    \caption{evadeV-6-3, SC = 0.1 - risk estimation subject to target, comparison between iHMM and HMM}
    \label{ihmm_hmm_overestimation_evadeV-6-3_high_st}
\end{figure}

\begin{figure}[H]
    \centering
    \includegraphics[width=0.64\linewidth]{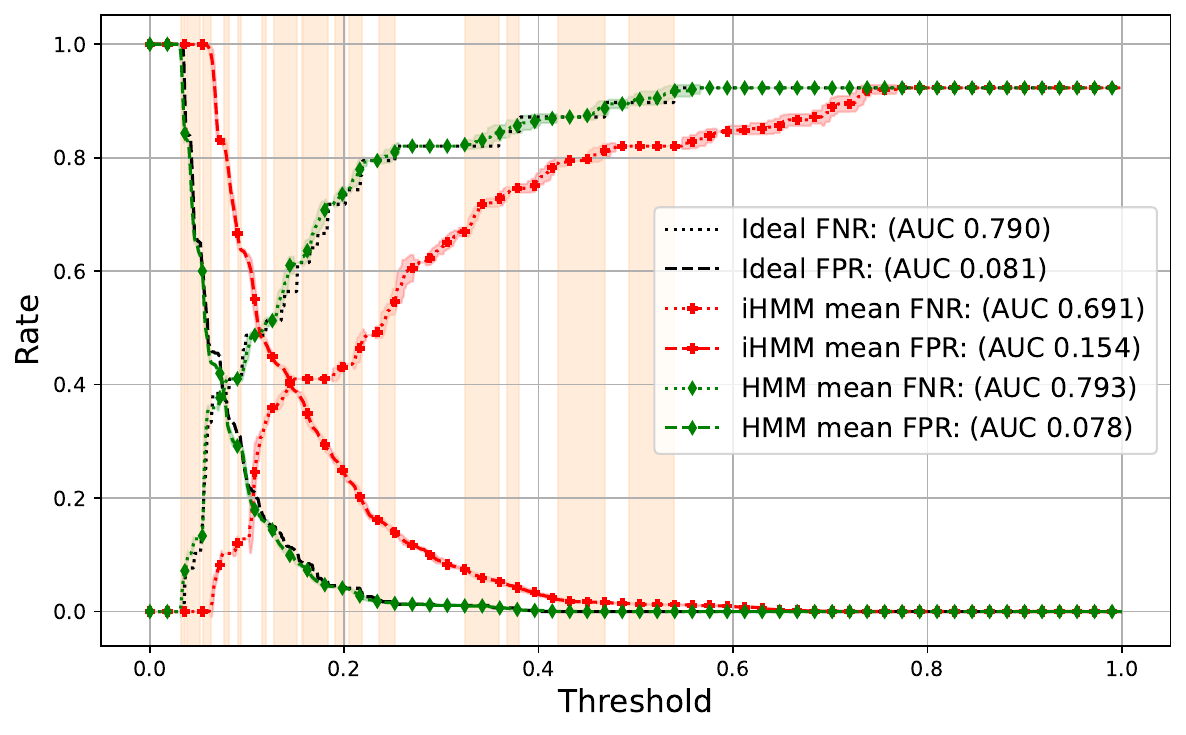}
    \caption{Coarse evadeV-6-3, SC = 0.01 - FNR and FPR comparison between iHMM and HMM}
    \label{ihmm_hmm_evadeV-6-3_coarse}
\end{figure}

\begin{figure}[H]
    \centering
    \includegraphics[width=0.64\linewidth]{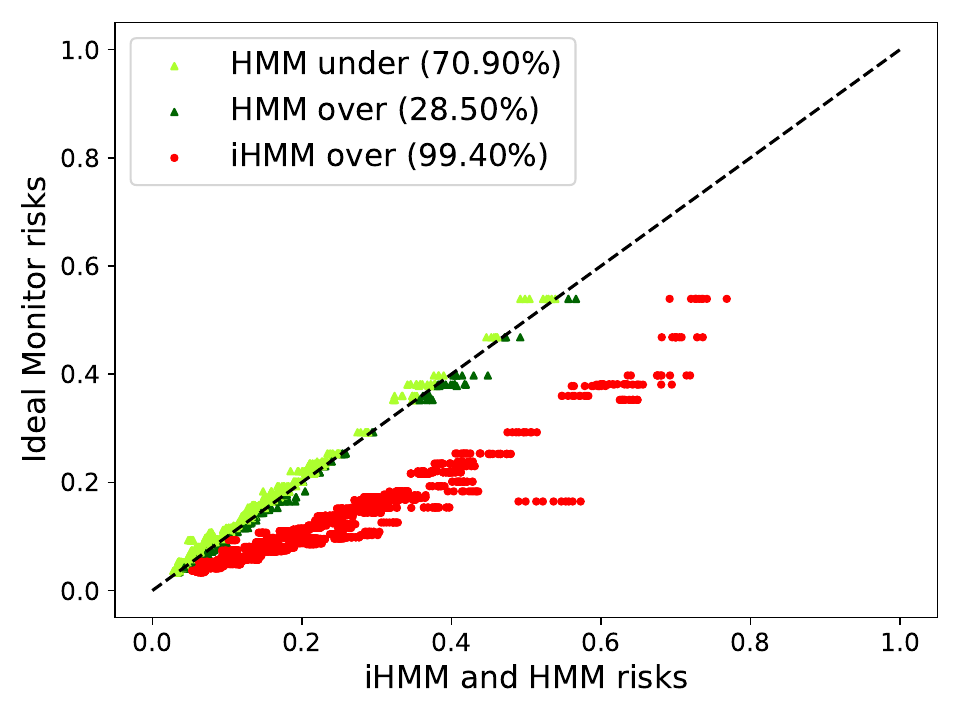}
    \caption{Coarse evadeV-6-3, SC = 0.01 - risk estimation subject to target, comparison between iHMM and HMM}
    \label{ihmm_hmm_overestimation_evadeV-6-3_coarse}
\end{figure}

\begin{figure}[H]
    \centering
    \includegraphics[width=0.64\linewidth]{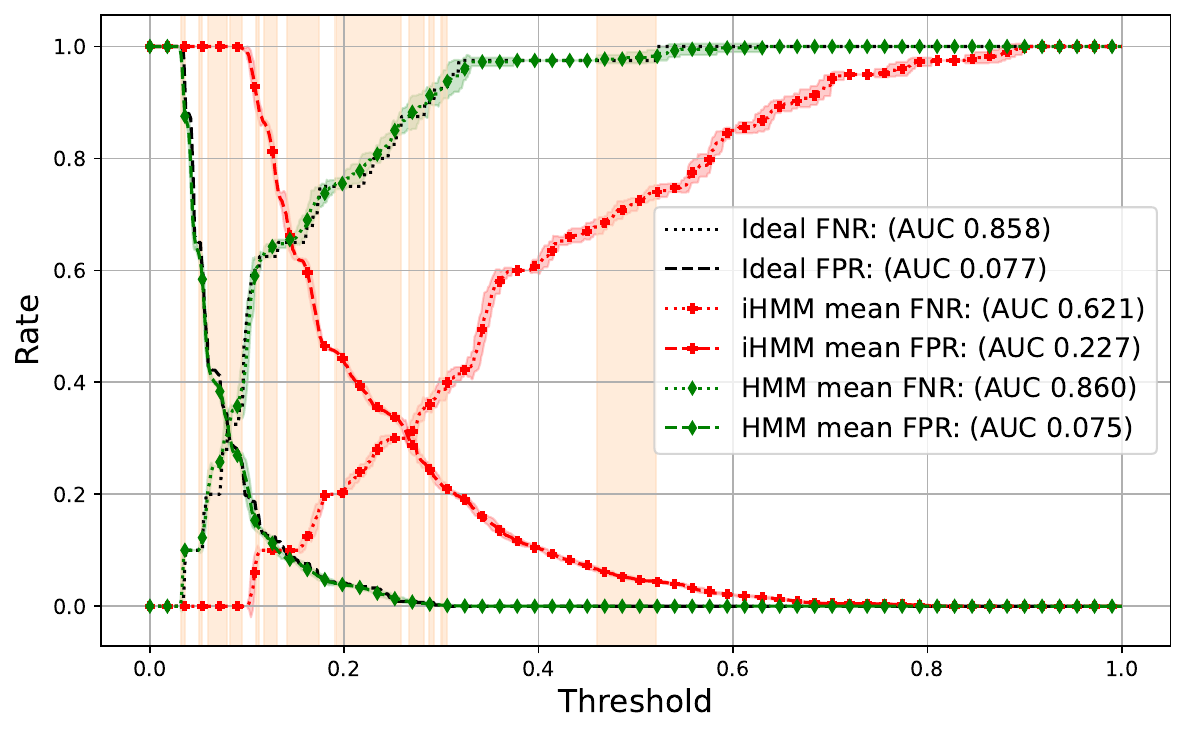}
    \caption{Coarse evadeV-6-3, SC = 0.1 - FNR and FPR comparison between iHMM and HMM}
    \label{ihmm_hmm_evadeV-6-3_high_st_coarse}
\end{figure}

\begin{figure}[H]
    \centering
    \includegraphics[width=0.64\linewidth]{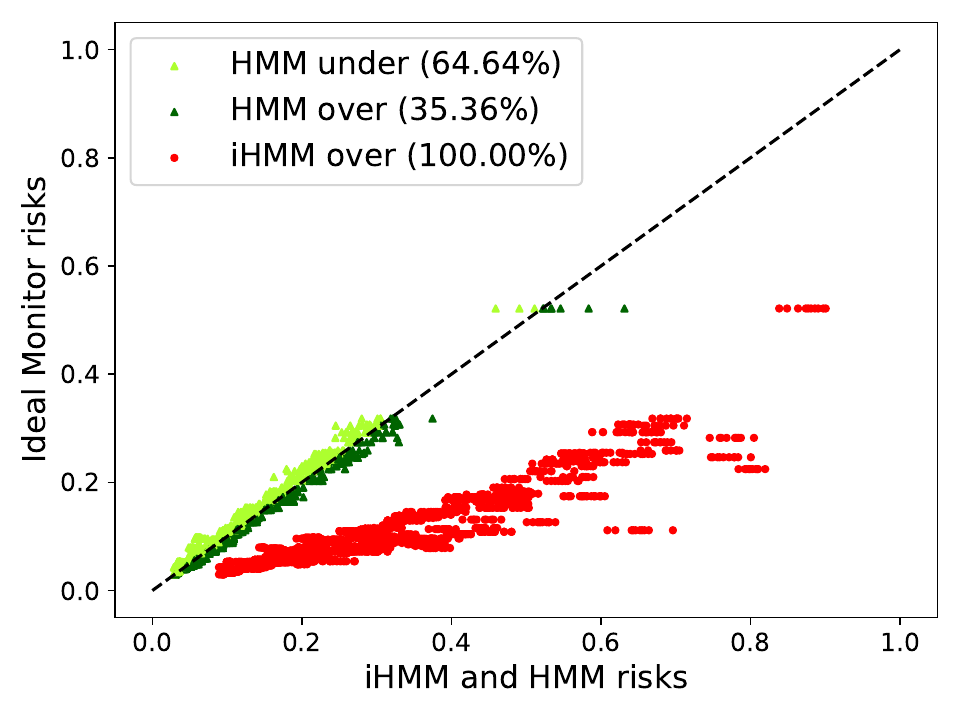}
    \caption{Coarse evadeV-6-3, SC = 0.1 - risk estimation subject to target, comparison between iHMM and HMM}
    \label{ihmm_hmm_overestimation_evadeV-6-3_high_st_coarse}
\end{figure}

\begin{figure}[H]
    \centering
    \includegraphics[width=0.64\linewidth]{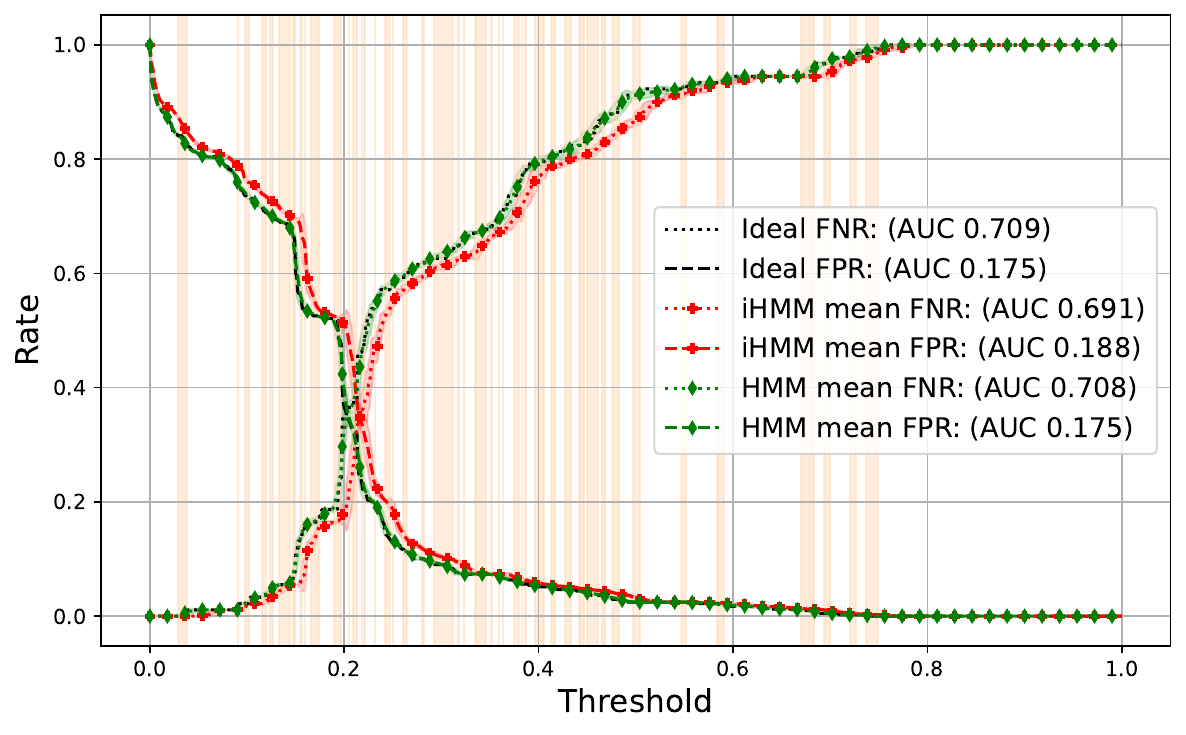}
    \caption{airportA-7-10-10, SC = 0.001  - FNR and FPR comparison between iHMM and HMM}
    \label{ihmm_hmm_airportA-7-10-10}
\end{figure}

\begin{figure}[H]
    \centering
    \includegraphics[width=0.64\linewidth]{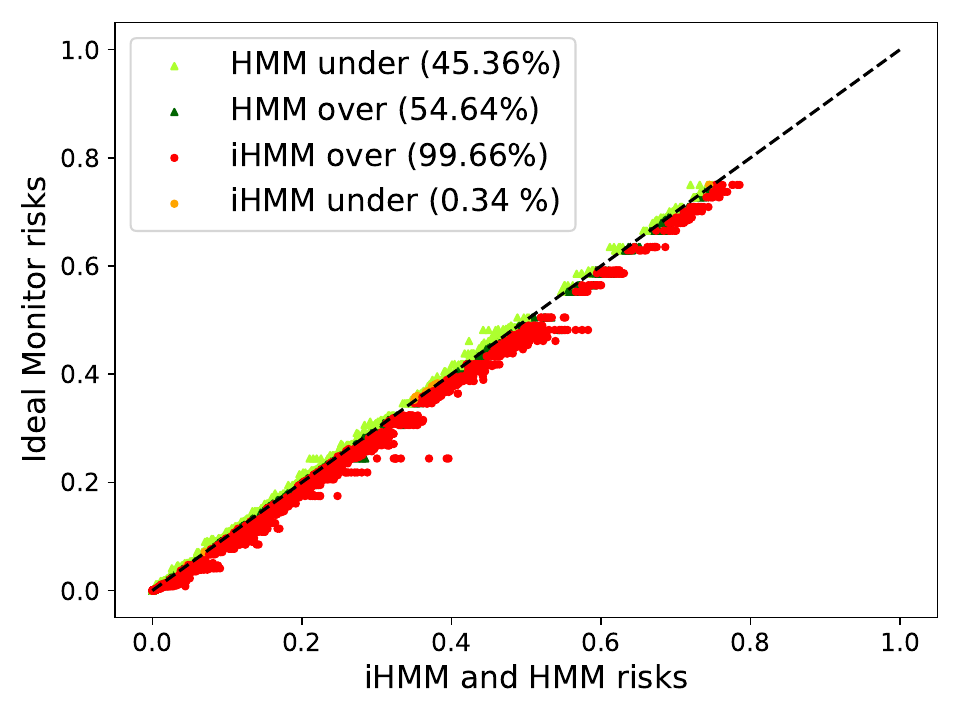}
    \caption{airportA-7-10-10, SC = 0.001 - risk estimation subject to target, comparison between iHMM and HMM}
    \label{ihmm_hmm_overestimation_airportA-7-10-10}
\end{figure}

\begin{figure}[H]
    \centering
    \includegraphics[width=0.64\linewidth]{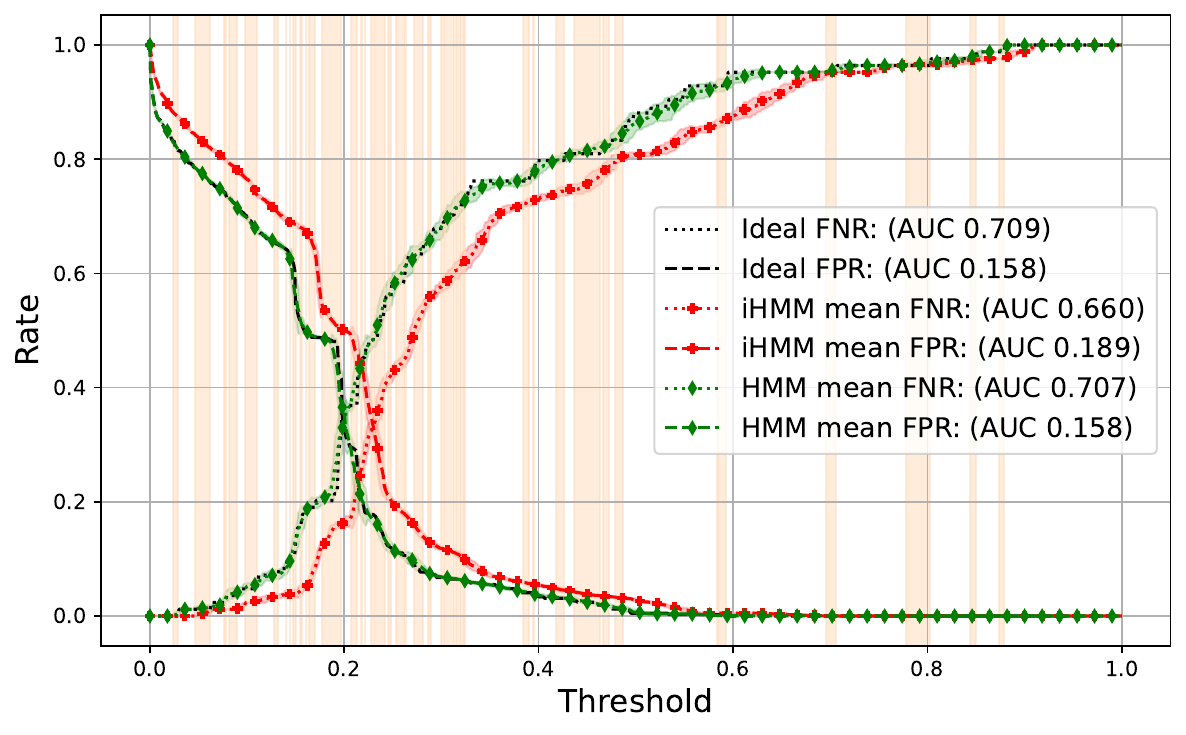}
    \caption{airportA-7-10-10, SC = 0.01 - FNR and FPR comparison between iHMM and HMM}
    \label{ihmm_hmm_airportA-7-10-10_high_st}
\end{figure}

\begin{figure}[H]
    \centering
    \includegraphics[width=0.64\linewidth]{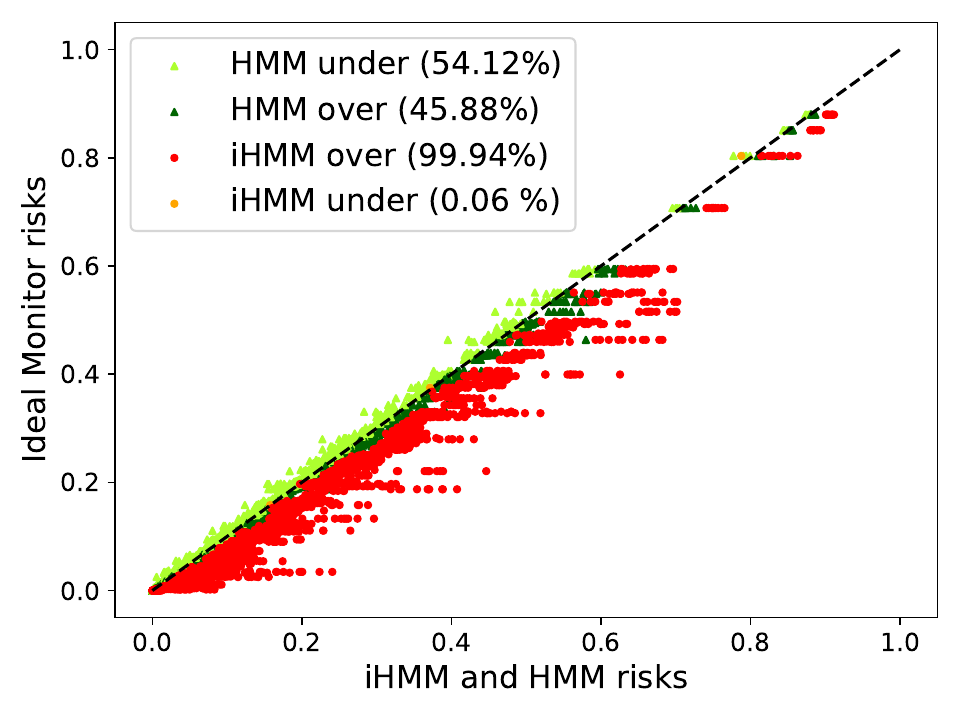}
    \caption{airportA-7-10-10, SC = 0.01 - risk estimation subject to target, comparison between iHMM and HMM}
    \label{ihmm_hmm_overestimation_airportA-7-10-10_high_st}
\end{figure}

\begin{figure}[H]
    \centering
    \includegraphics[width=0.64\linewidth]{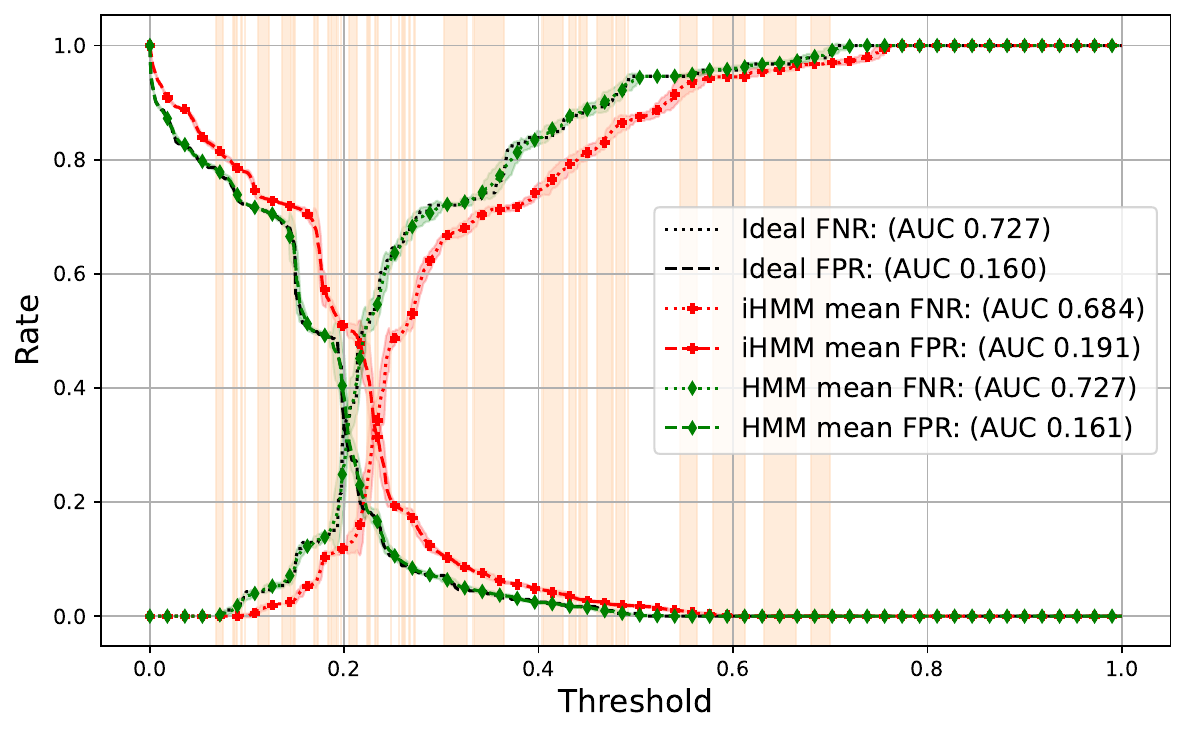}
    \caption{Coarse airportA-7-10-10, SC = 0.01 - FNR and FPR comparison between iHMM and HMM}
    \label{ihmm_hmm_airportA-7-10-10_coarse}
\end{figure}

\begin{figure}[H]
    \centering
    \includegraphics[width=0.64\linewidth]{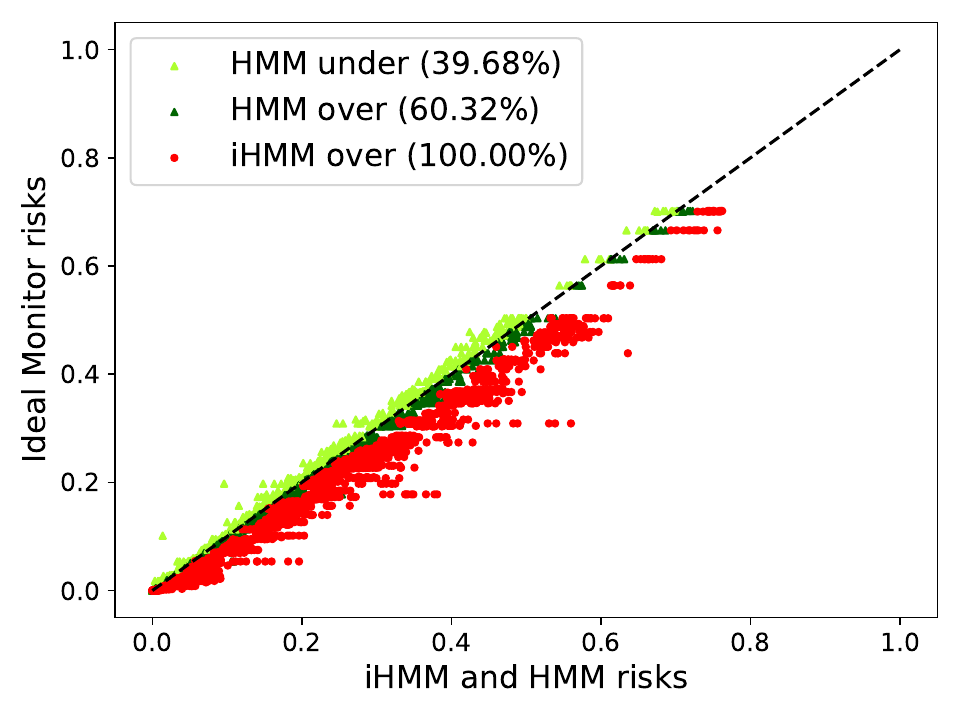}
    \caption{Coarse airportA-7-10-10, SC =  0.01 - risk estimation subject to target, comparison between iHMM and HMM}
    \label{ihmm_hmm_overestimation_airportA-7-10-10_coarse}
\end{figure}

\begin{figure}[H]
    \centering
    \includegraphics[width=0.64\linewidth]{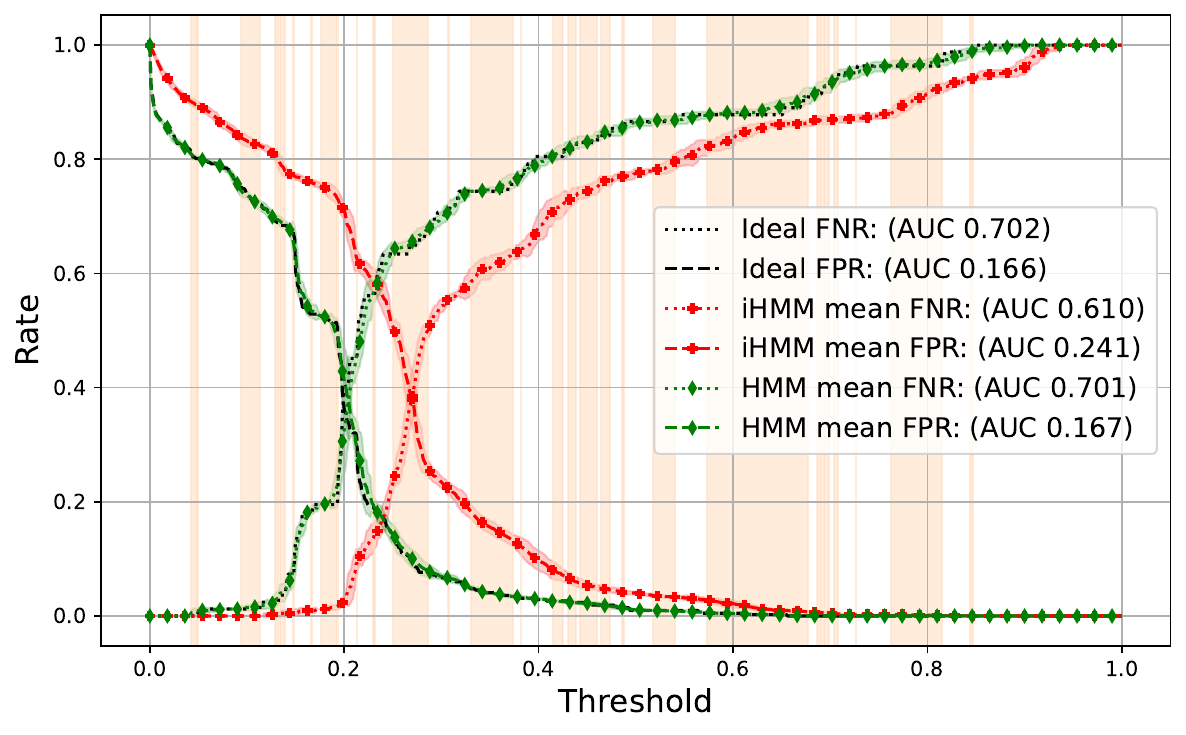}
    \caption{Coarse airportA-7-10-10, SC = 0.1 - FNR and FPR comparison between iHMM and HMM}
    \label{ihmm_hmm_airportA-7-10-10_coarse_high_st}
\end{figure}

\begin{figure}[H]
    \centering
    \includegraphics[width=0.64\linewidth]{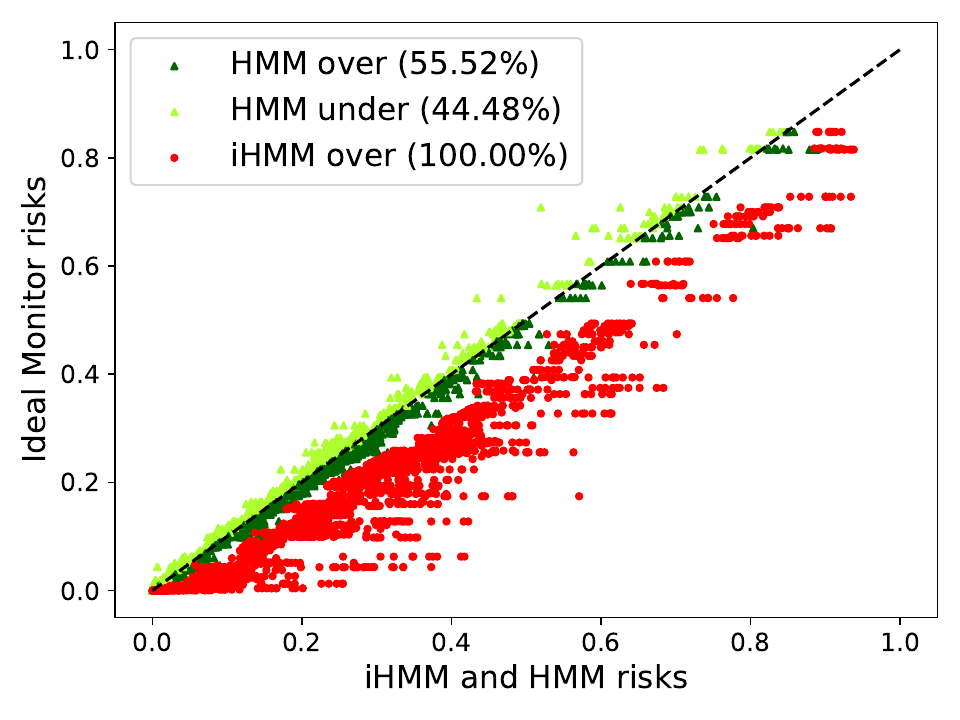}
    \caption{Coarse airportA-7-10-10, SC =  0.1 - risk estimation subject to target, comparison between iHMM and HMM}
    \label{ihmm_hmm_overestimation_airportA-7-10-10_coarse_high_st}
\end{figure}

\subsection{Unlikely-15 Benchmark}

We expand on the structure of the unlikely-15 benchmark introduced in \Cref{sec:ihmm_vs_hmm_comp}. 

Unlikely-15 contains increasingly harder to sample sections, where each section requires multiple samples to learn its transitions such that we can accurately predict the risk of the section. Each section has a unique trace leading into it, thus giving many traces whose risk is hard to predict. 

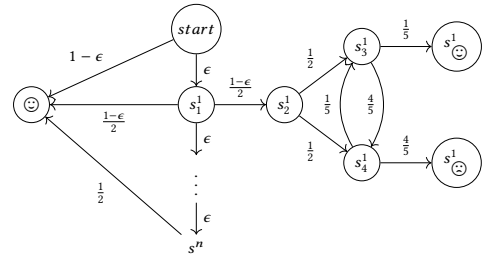
\begin{figure}[H]
    \centering
    \scalebox{0.85}{
        \begin{tikzpicture}[node distance=5mm and 8mm,font=\footnotesize]
            \tikzset{state/.append style={minimum size=5mm, inner sep=2pt}}
            
            \node[state] (start) {$start$};
            \node[state] (s11) [below = of start] {$s^1_1$};
            \node[state] (s12) [right = of s11] {$s^1_2$};
            \node[state] (s13) [above right = of s12] {$s^1_3$};
            \node[state] (s14) [below right = of s12] {$s^1_4$};
            \node[state] (s15) [right = of s13] {$s^1_{\large\smiley}$};
            \node[state] (s16) [right = of s14] {$s^1_{\large\frownie}$};
            \node (dots) [below= of s11] {$\vdots$};
            \node[state] (happy) [left = 2cm of s11] {\large \smiley};
            \node (sn) [below= of dots] {$s^n$};

            \path[->] (start) edge node[right] {$\epsilon$} (s11);
            \path[->] (s11) edge node[right] {$\epsilon$} (dots);
            \path[->] (s11) edge node[above] {$\frac{1-\epsilon}{2}$} (s12);
            \path[->] (s12) edge node[above left] {$\frac{1}{2}$} (s13);
            \path[->] (s12) edge node[below left] {$\frac{1}{2}$} (s14);
            \path[->] (s13) edge[bend left] node [left] {$\frac{4}{5}$} (s14);
            \path[->] (s14) edge[bend left] node [left] {$\frac{1}{5}$} (s13);
            \path[->] (s13) edge node[above] {$\frac{1}{5}$} (s15);
            \path[->] (s14) edge node[above] {$\frac{4}{5}$} (s16);
            \path[->] (start) edge node[above left] {$1-\epsilon$} (happy);
            \path[->] (s11) edge node[below] {$\frac{1-\epsilon}{2}$} (happy);
            \path[->] (dots) edge node[right] {$\epsilon$} (sn);
            \path[->] (sn) edge node[below left] {$\frac{1}{2}$} (happy);
        \end{tikzpicture}
    }
    \caption{Diagram of the \textsc{unlikely} benchmark. Good states are denoted by $\smiley$ and bad states by $\frownie$. The observations for states $s_i^j$ are, the value of $j$, and whether $i\geq 2$ or $i\in\{\smiley,\frownie\}$. Other states are fully observable. Our specification is to never be in a bad state, and the horizon is 2. We copy the $s^1$ part of the diagram $n$ times, where, for the experiments we have chosen $n=15$. }
    \label{fig:unlikely}
\end{figure}

\subsection{Refinement vs No Refinement}
\label{apx:ref_no_ref}

For each of the benchmarks not included in \Cref{sec:No_ref_vs_ref}, we add corresponding analysis in terms of comparison in terms of distance to Ideal Monitor and comparison in terms of FNR and FPR across different thresholds between learning with and without refinement. The experimental evaluation of monitors for benchmarks airportA-7-40-20 and coarse airportB-7-40-20 is included in \Cref{apx:big_benchmarks}.

\begin{figure}[H]
    \centering
    \includegraphics[width=0.8\linewidth]{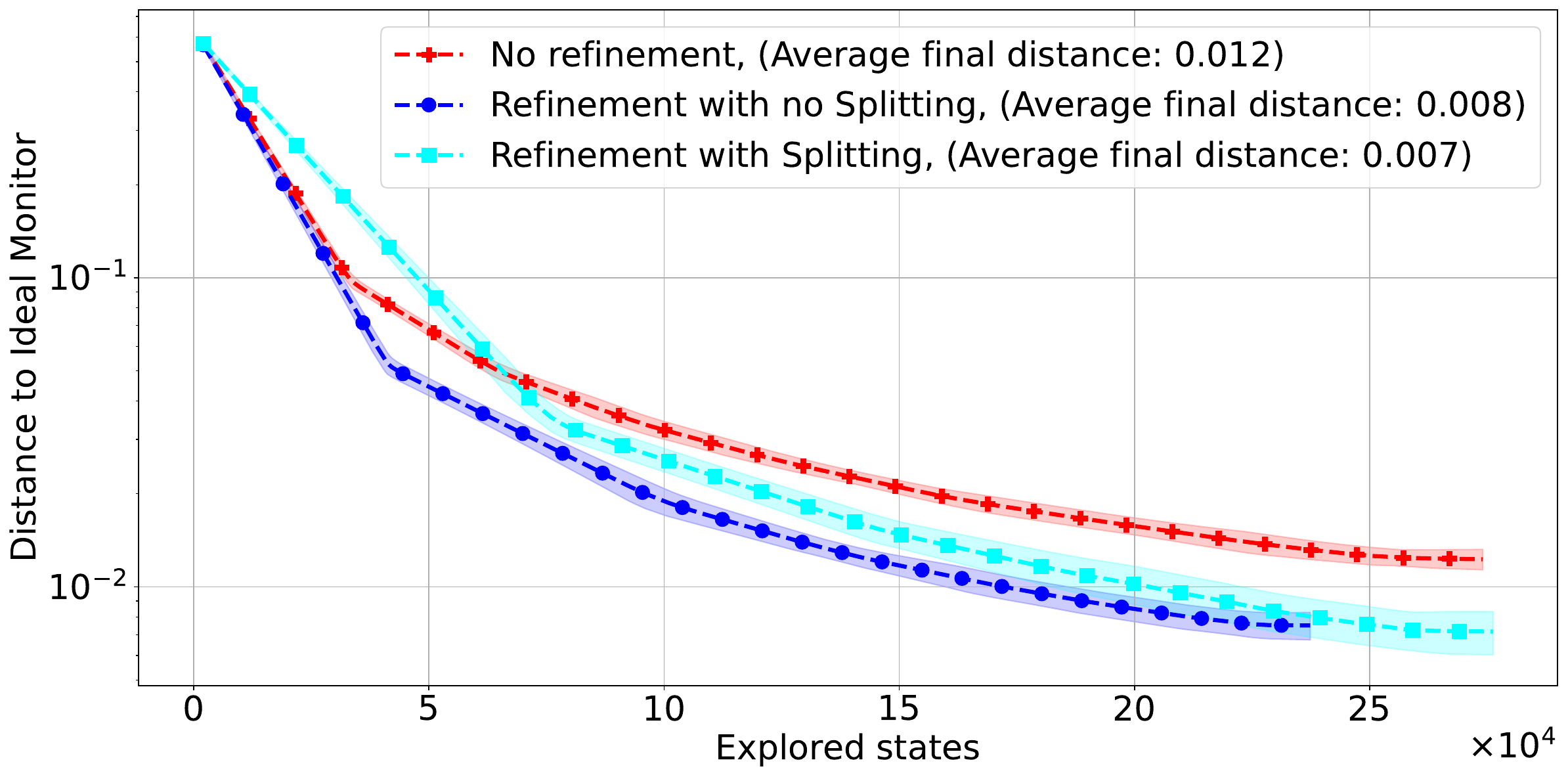}
    \caption{SnL-10x10, SC = 0.001 - Comparison between learning with and without refinement, in terms of distance to Ideal Monitor}
    \label{}
\end{figure}

\begin{figure}[H]
    \centering
    \includegraphics[width=0.8\linewidth]{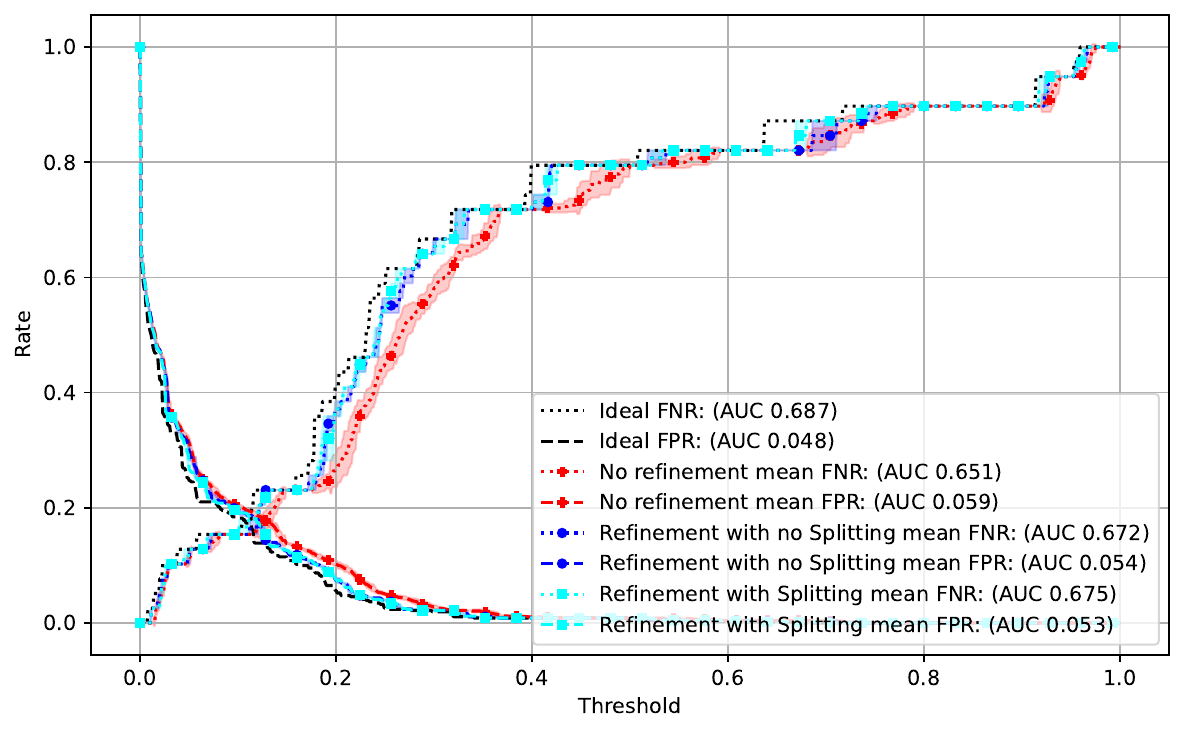}
    \caption{SnL-10x10, SC = 0.001 - Comparison between learning with and without refinement, in terms of FNR and FPR}
    \label{}
\end{figure}

\begin{figure}[H]
    \centering
    \includegraphics[width=0.8\linewidth]{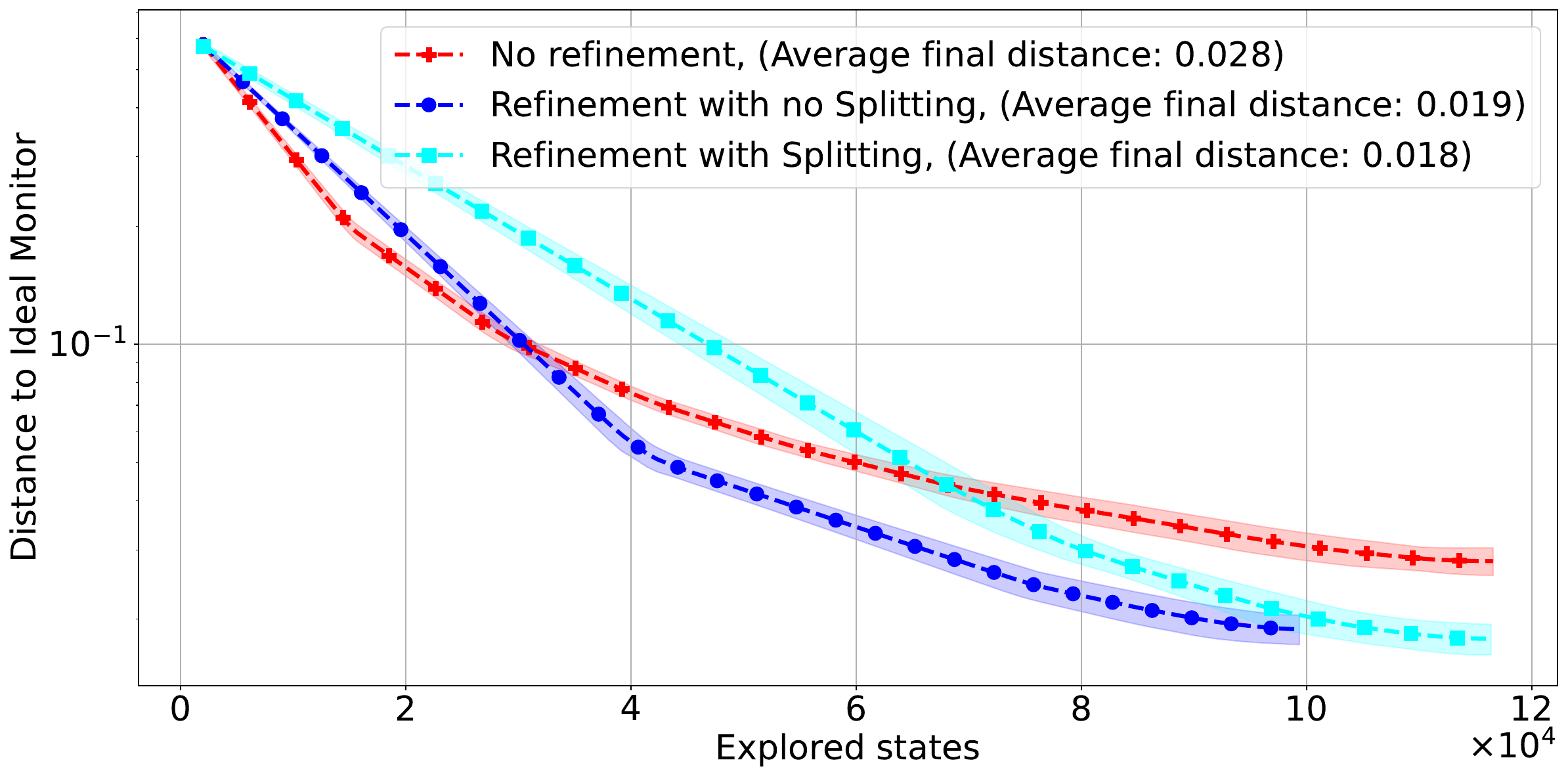}
    \caption{SnL-10x10, SC = 0.01 - Comparison between learning with and without refinement, in terms of distance to Ideal Monitor}
    \label{}
\end{figure}

\begin{figure}[H]
    \centering
    \includegraphics[width=0.8\linewidth]{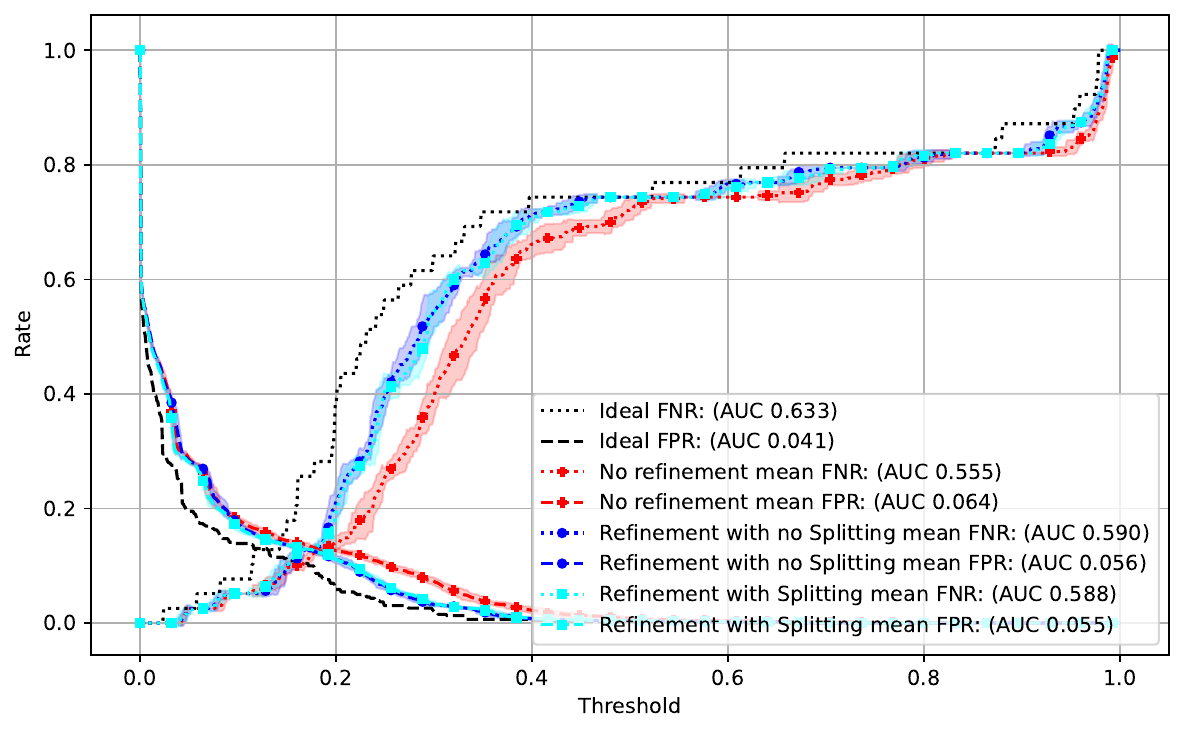}
    \caption{SnL-10x10, SC = 0.01 - Comparison between learning with and without refinement, in terms of FNR and FPR}
    \label{}
\end{figure}

\begin{figure}[H]
    \centering
    \includegraphics[width=0.8\linewidth]{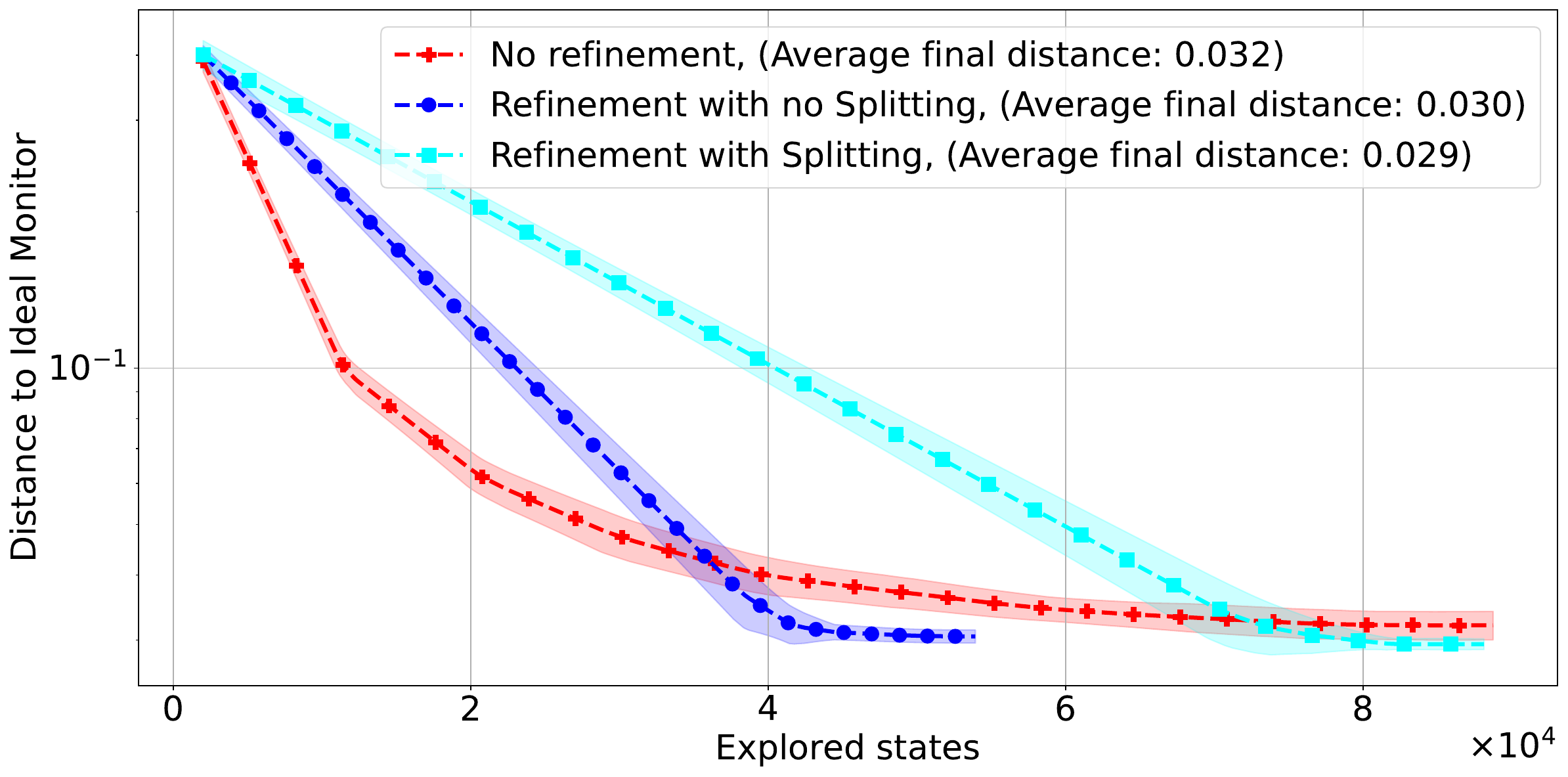}
    \caption{Coarse SnL-10x10, SC = 0.01, Comparison between learning with and without refinement, in terms of distance to Ideal Monitor}
    \label{}
\end{figure}

\begin{figure}[H]
    \centering
    \includegraphics[width=0.8\linewidth]{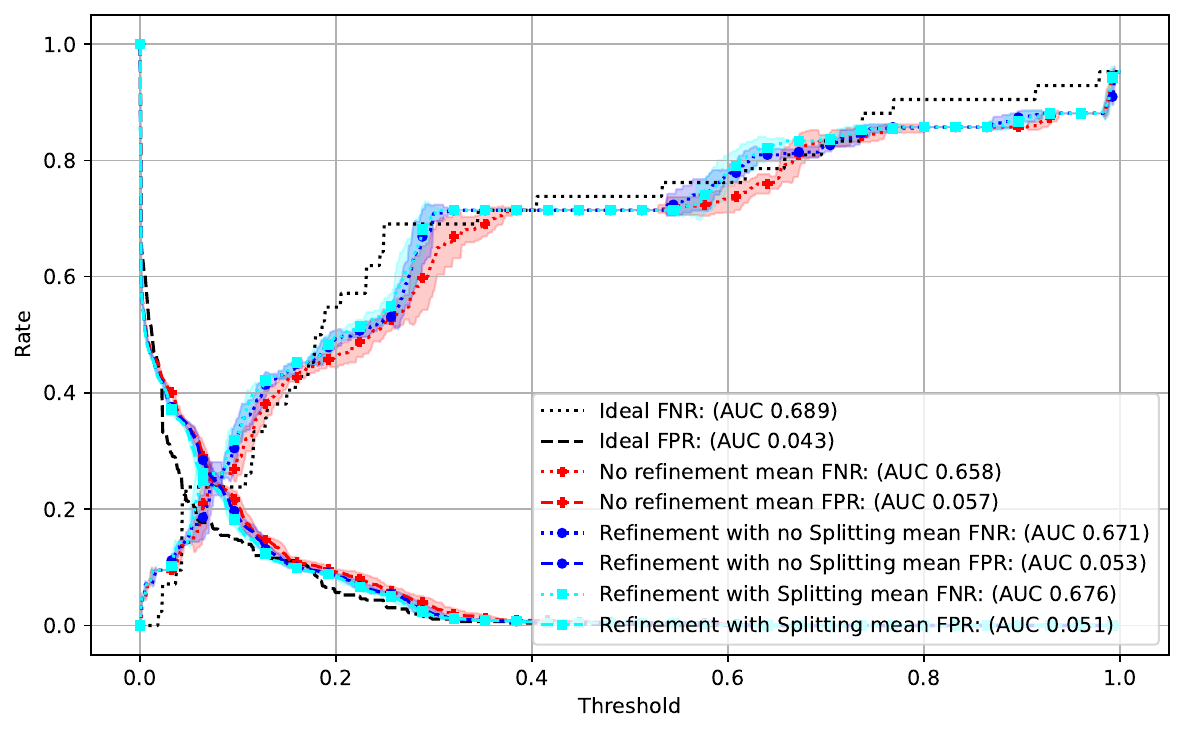}
    \caption{Coarse SnL-10x10, SC = 0.01, Comparison between learning with and without refinement, in terms of FNR and FPR}
    \label{}
\end{figure}

\begin{figure}[H]
    \centering
    \includegraphics[width=0.8\linewidth]{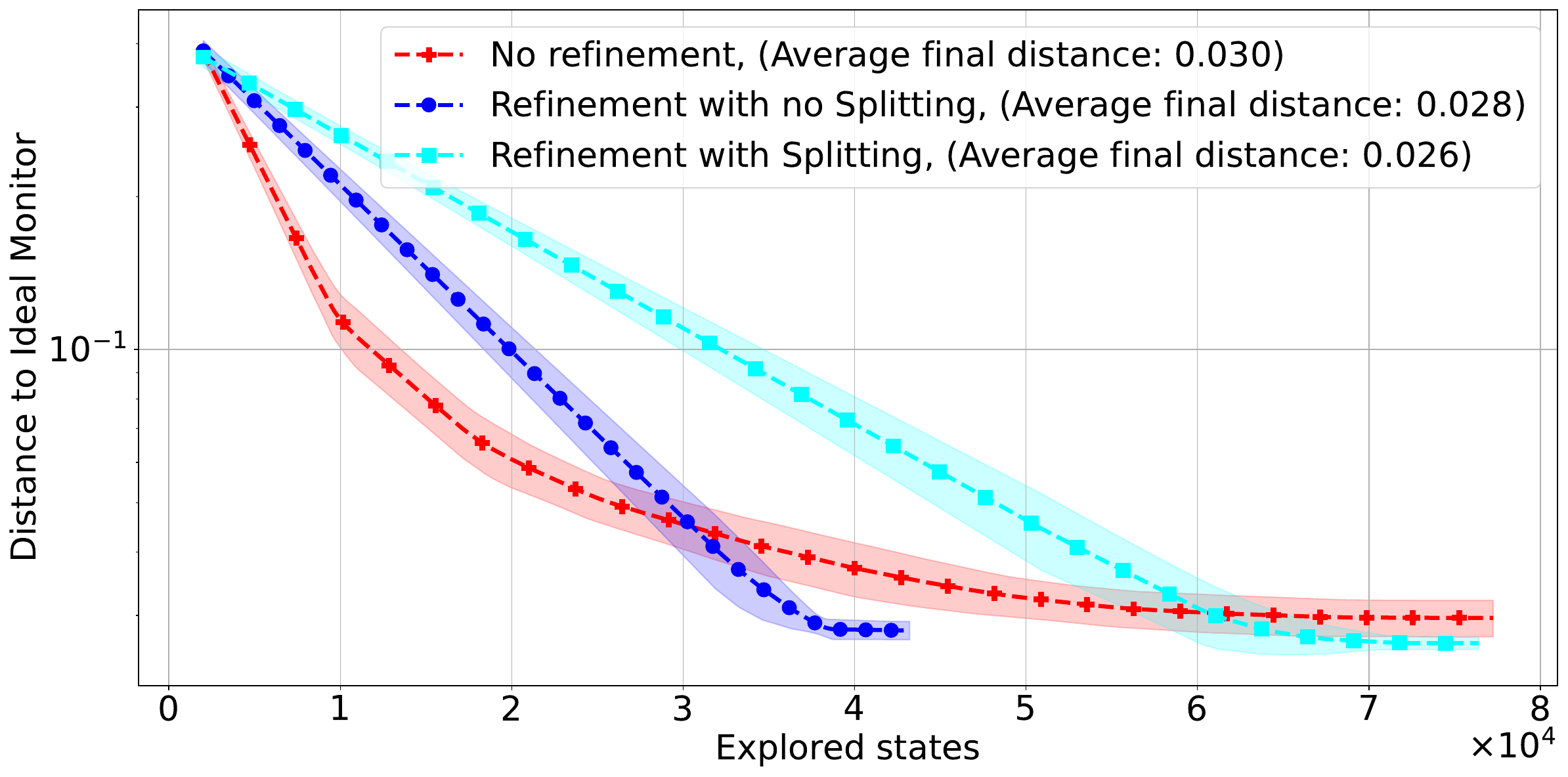}
    \caption{Coarse SnL-10x10, SC = 0.1 - Comparison between learning with and without refinement, in terms of distance to Ideal Monitor}
    \label{}
\end{figure}

\begin{figure}[H]
    \centering
    \includegraphics[width=0.8\linewidth]{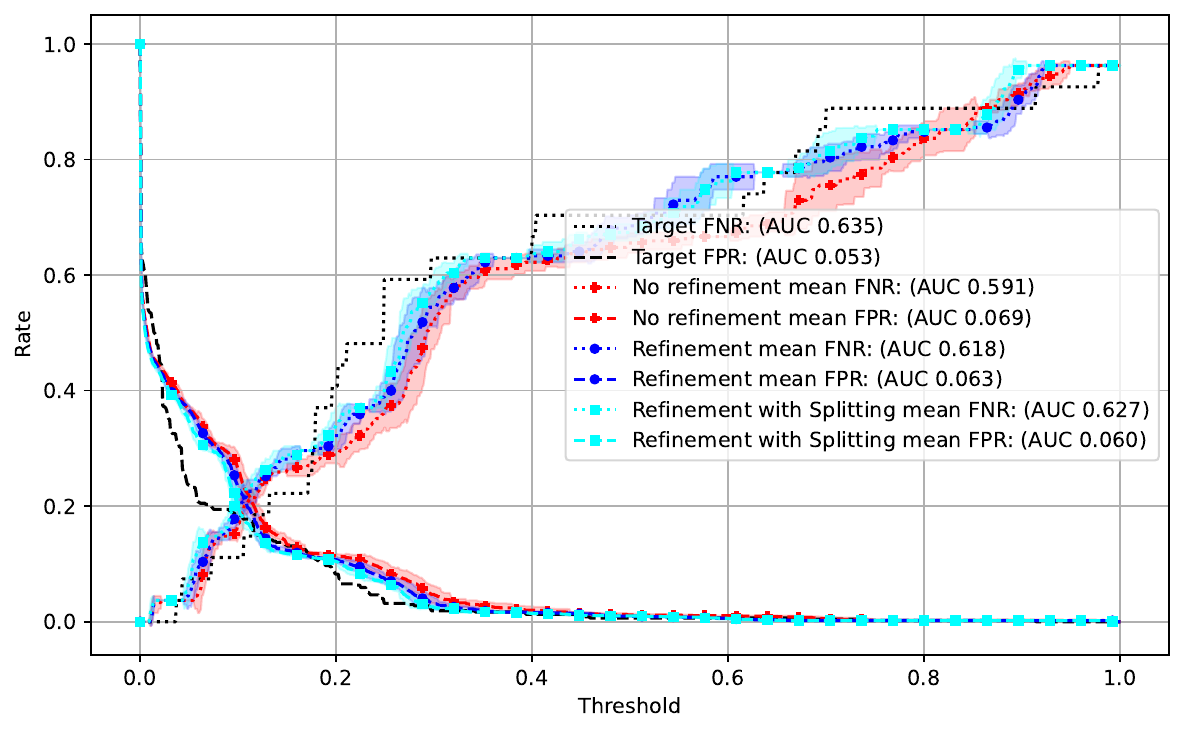}
    \caption{Coarse SnL-10x10, SC = 0.1 - Comparison between learning with and without refinement, in terms of FNR and FPR}
    \label{}
\end{figure}

\begin{figure}[H]
    \centering
    \includegraphics[width=0.8\linewidth]{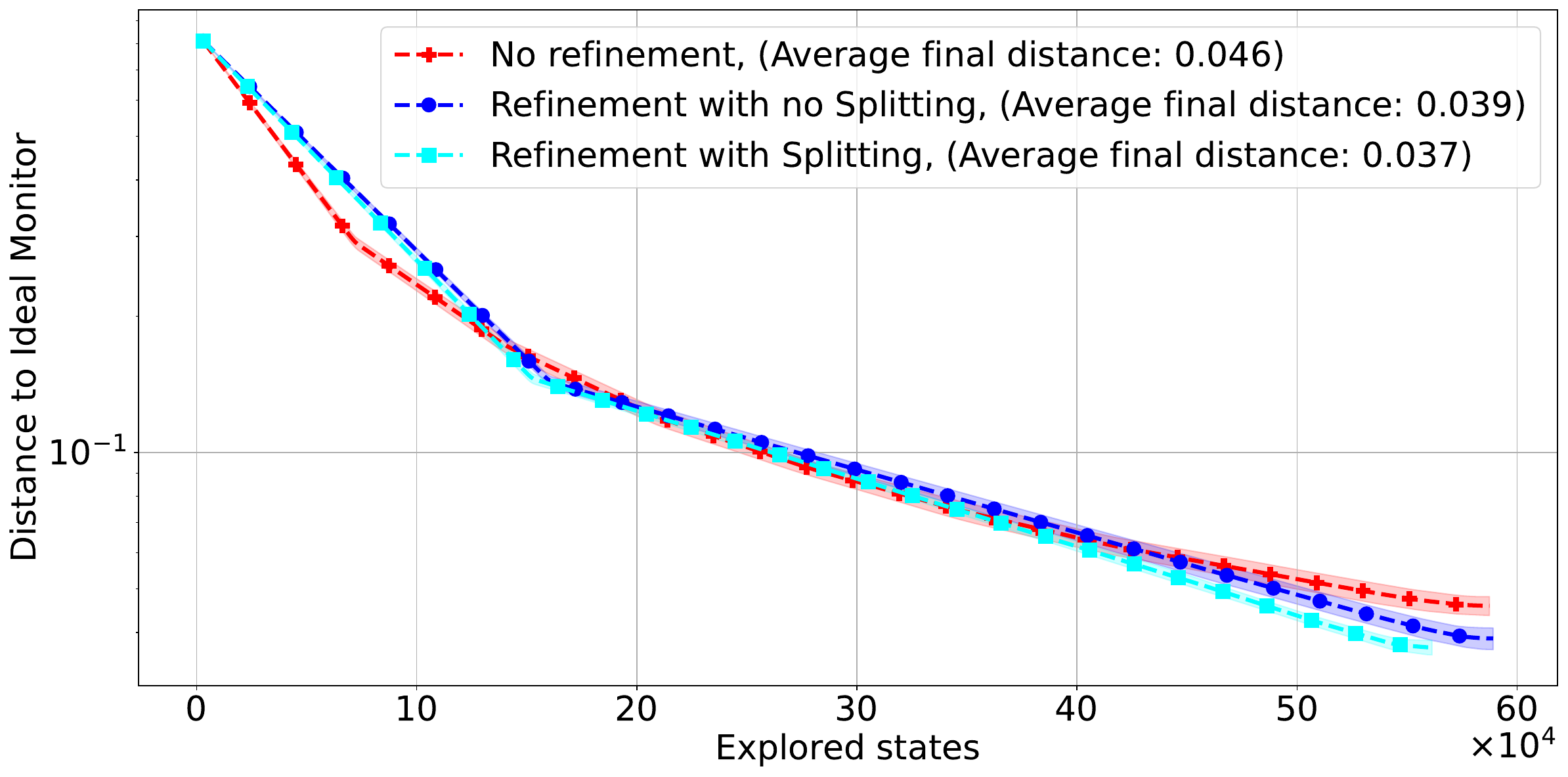}
    \caption{evadeV-5-3, SC = 0.01 - Comparison between learning with and without refinement, in terms of distance to Ideal Monitor}
    \label{}
\end{figure}

\begin{figure}[H]
    \centering
    \includegraphics[width=0.8\linewidth]{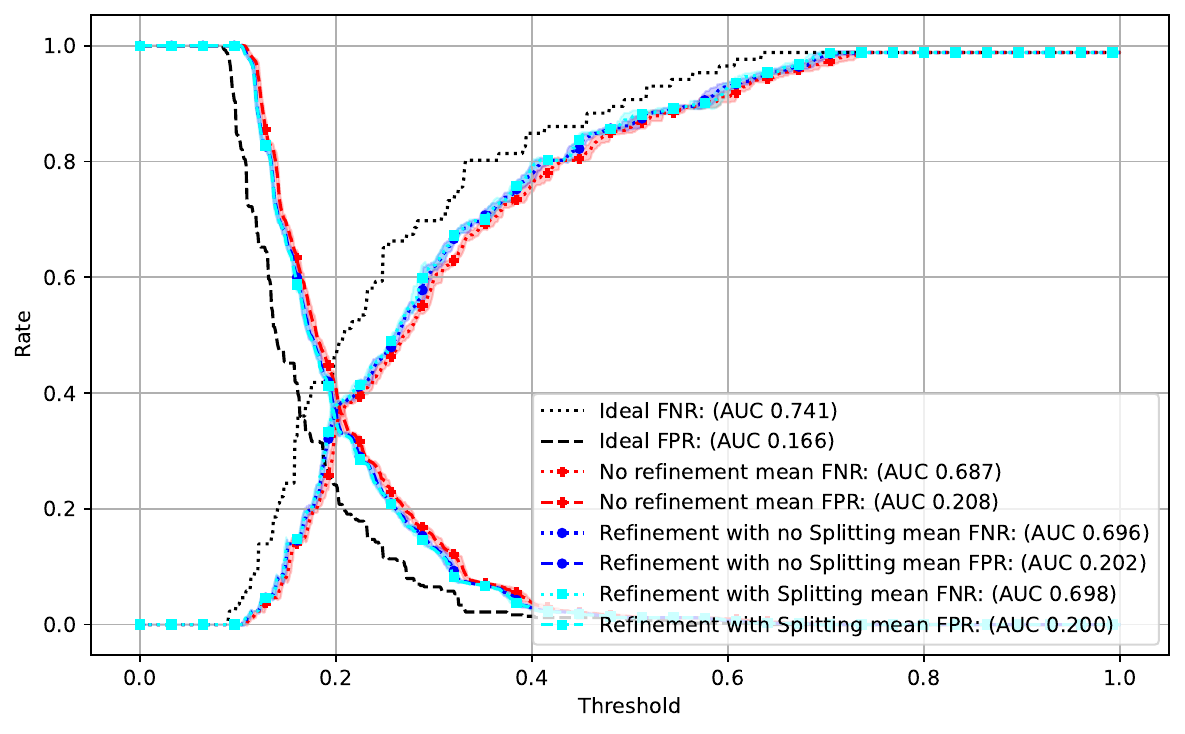}
    \caption{evadeV-5-3, SC = 0.01 - Comparison between learning with and without refinement, in terms of distance to FNR and FPR}
    \label{}
\end{figure}

\begin{figure}[H]
    \centering
    \includegraphics[width=0.8\linewidth]{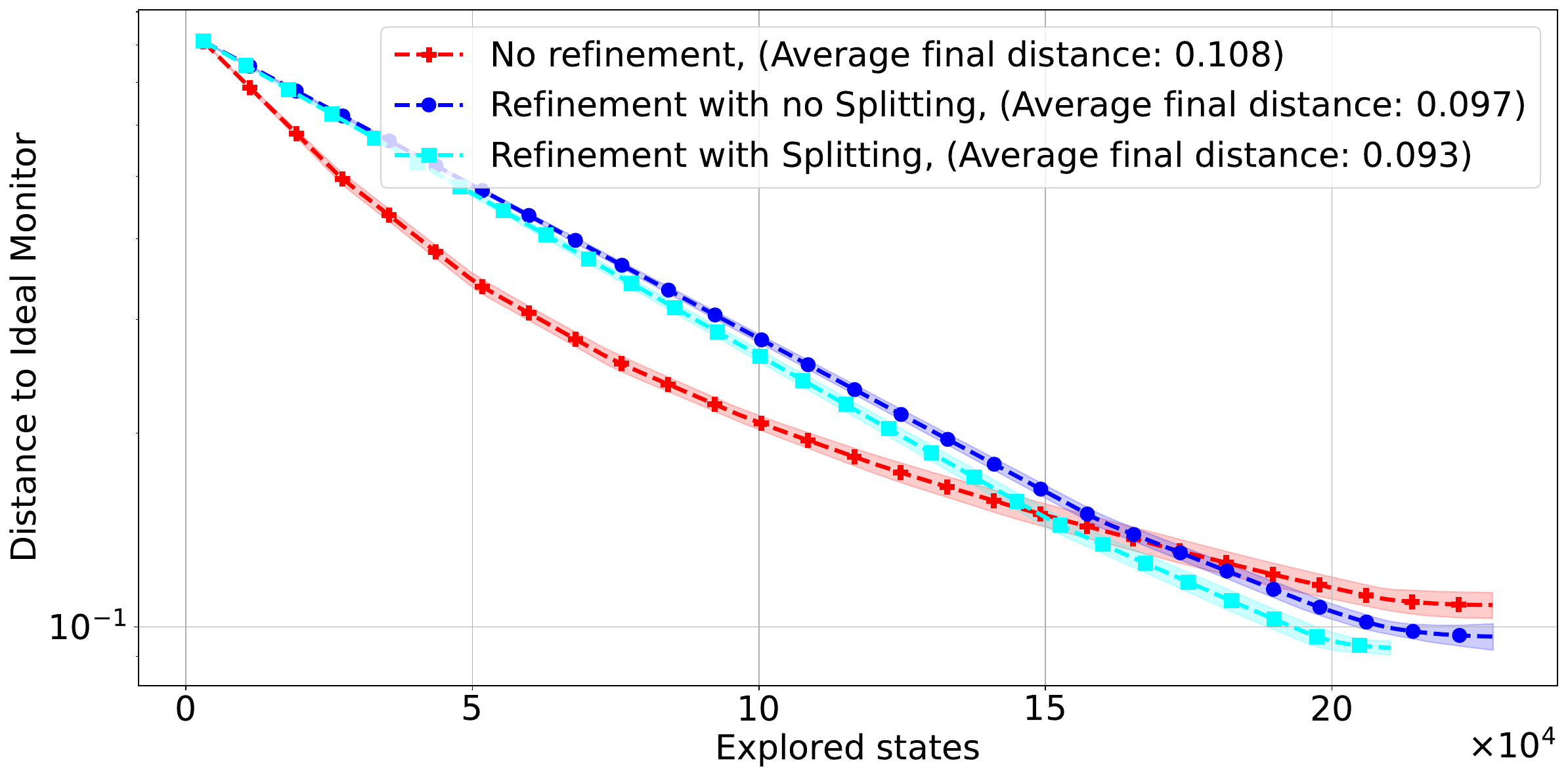}
    \caption{evadeV-5-3, SC = 0.1 - Comparison between learning with and without refinement, in terms of distance to Ideal Monitor}
    \label{}
\end{figure}

\begin{figure}[H]
    \centering
    \includegraphics[width=0.8\linewidth]{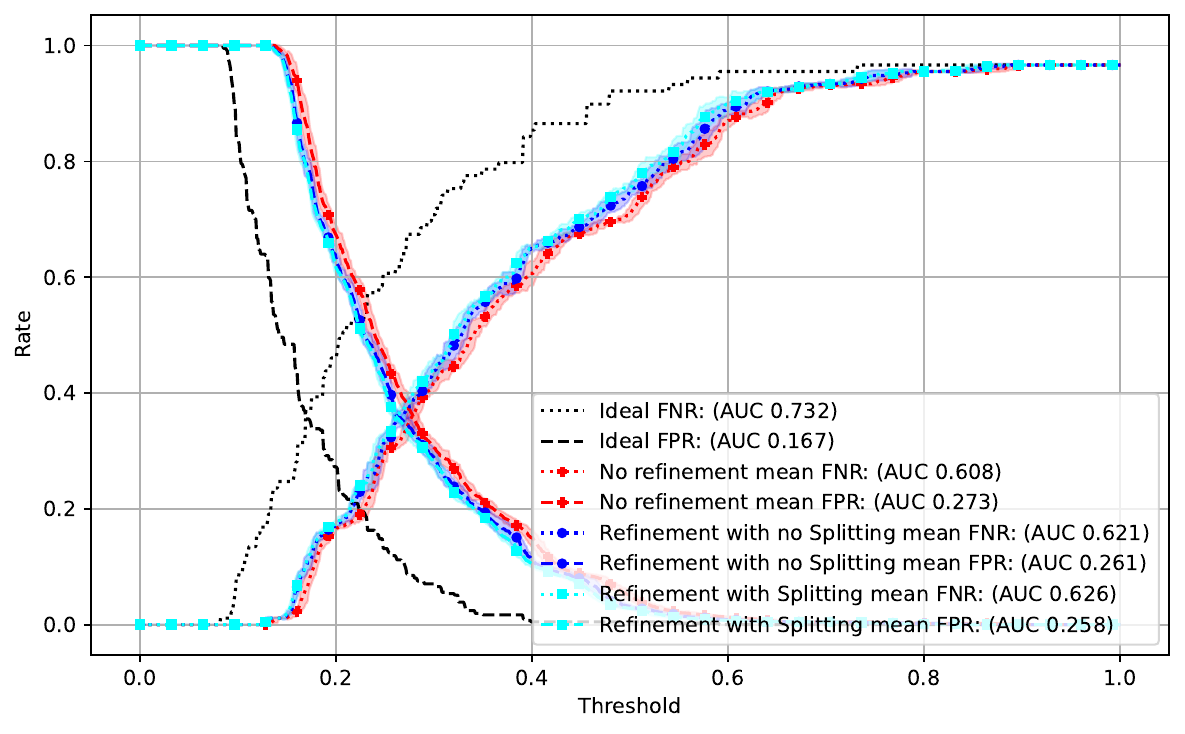}
    \caption{evadeV-5-3, SC = 0.1 - Comparison between learning with and without refinement, in terms of distance to FNR and FPR}
    \label{}
\end{figure}

\begin{figure}[H]
    \centering
    \includegraphics[width=0.8\linewidth]{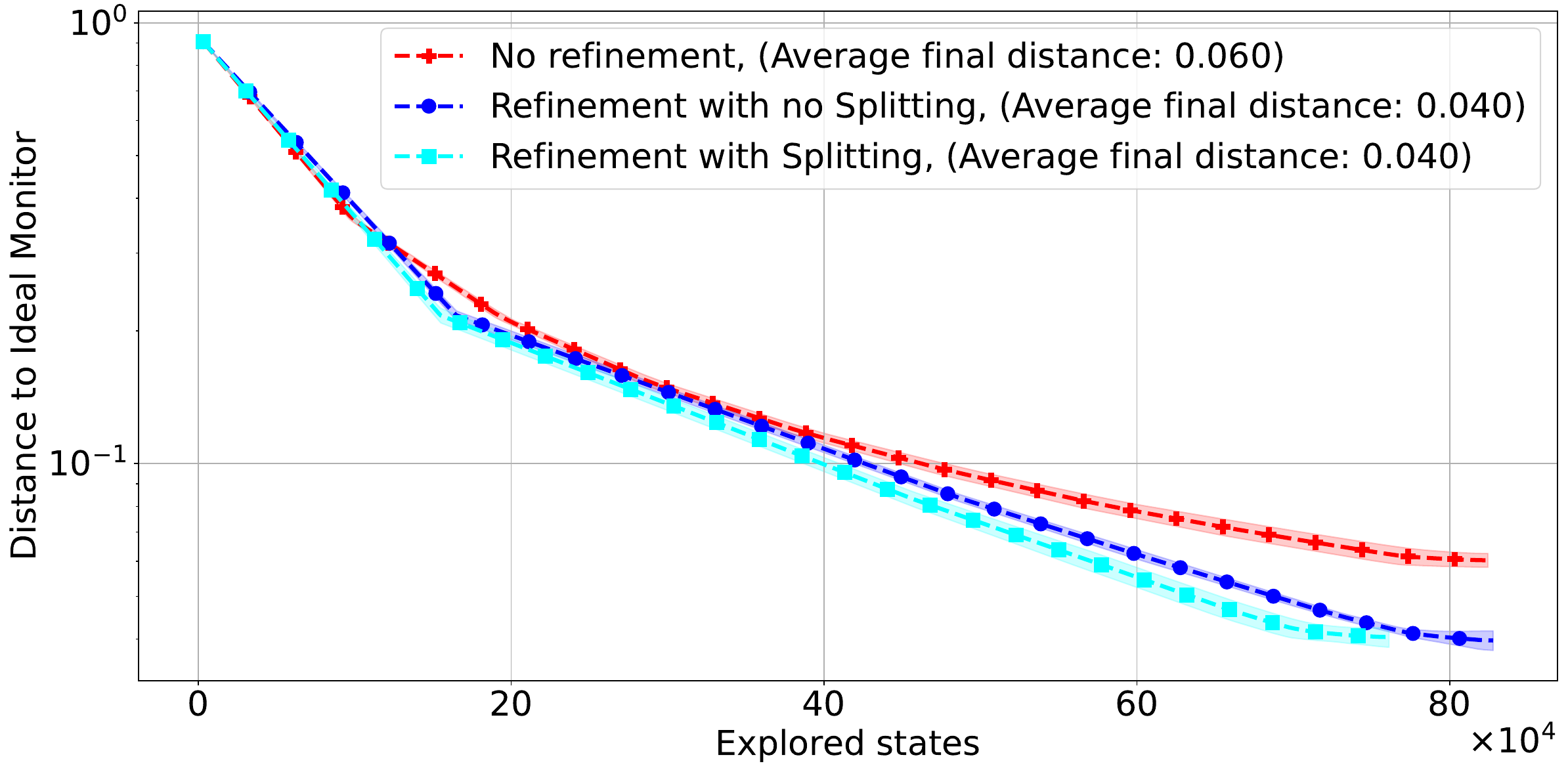}
    \caption{evadeV-6-3, SC = 0.01 - Comparison between learning with and without refinement, in terms of distance to Ideal Monitor}
    \label{}
\end{figure}

\begin{figure}[H]
    \centering
    \includegraphics[width=0.8\linewidth]{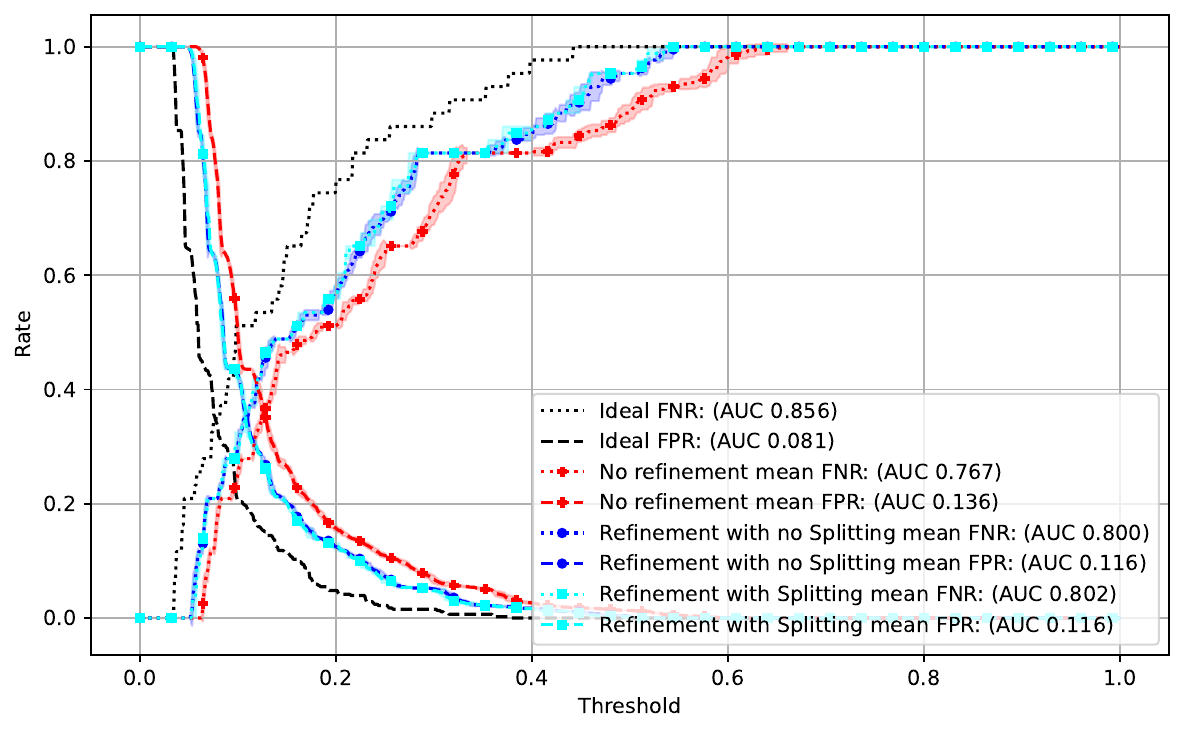}
    \caption{evadeV-6-3, SC = 0.01 - Comparison between learning with and without refinement, in terms of distance to FNR and FPR}
    \label{}
\end{figure}

\begin{figure}[H]
    \centering
    \includegraphics[width=0.8\linewidth]{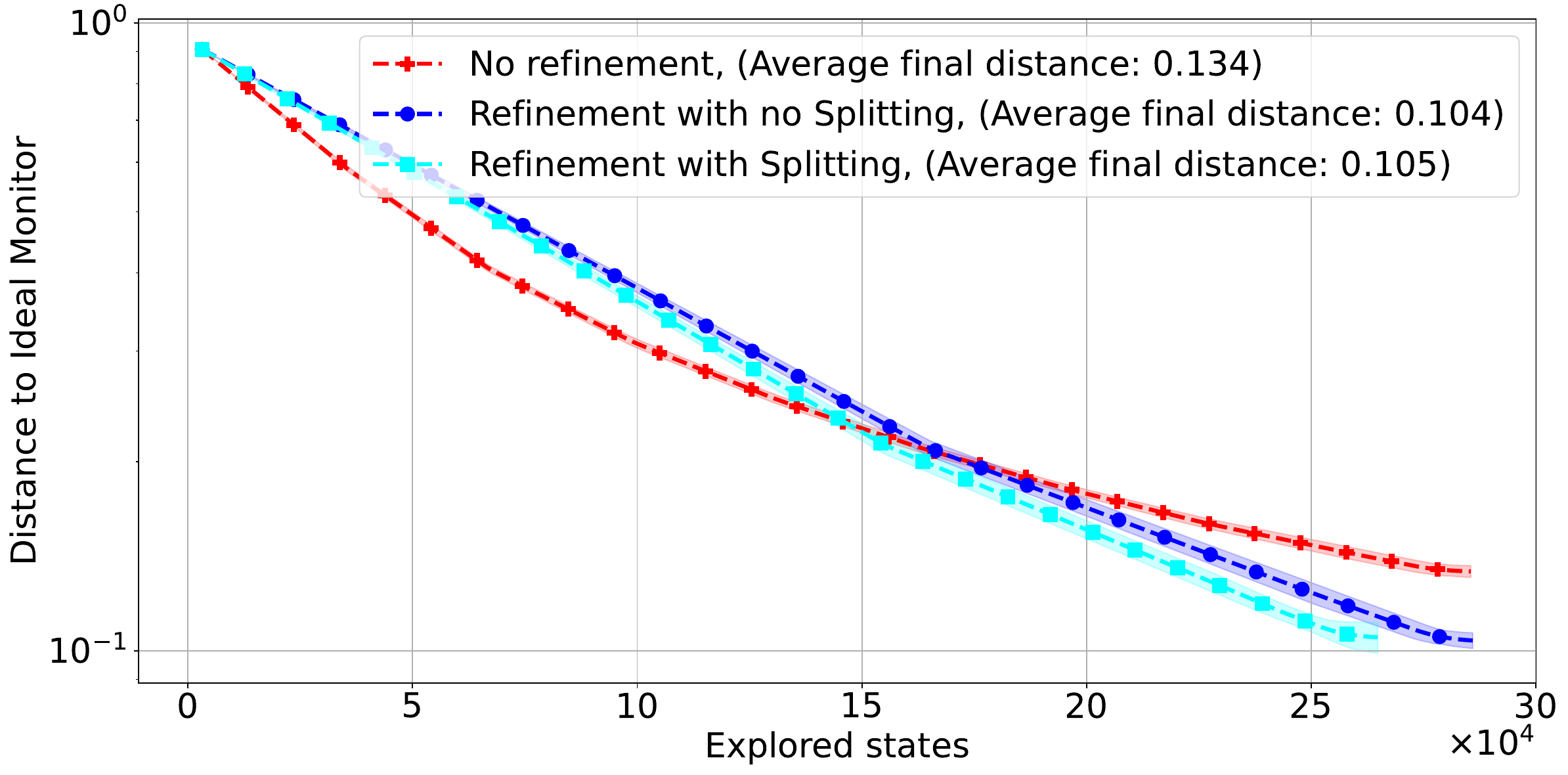}
    \caption{evadeV-6-3, SC = 0.1 - Comparison between learning with and without refinement, in terms of distance to Ideal Monitor}
    \label{}
\end{figure}

\begin{figure}[H]
    \centering
    \includegraphics[width=0.8\linewidth]{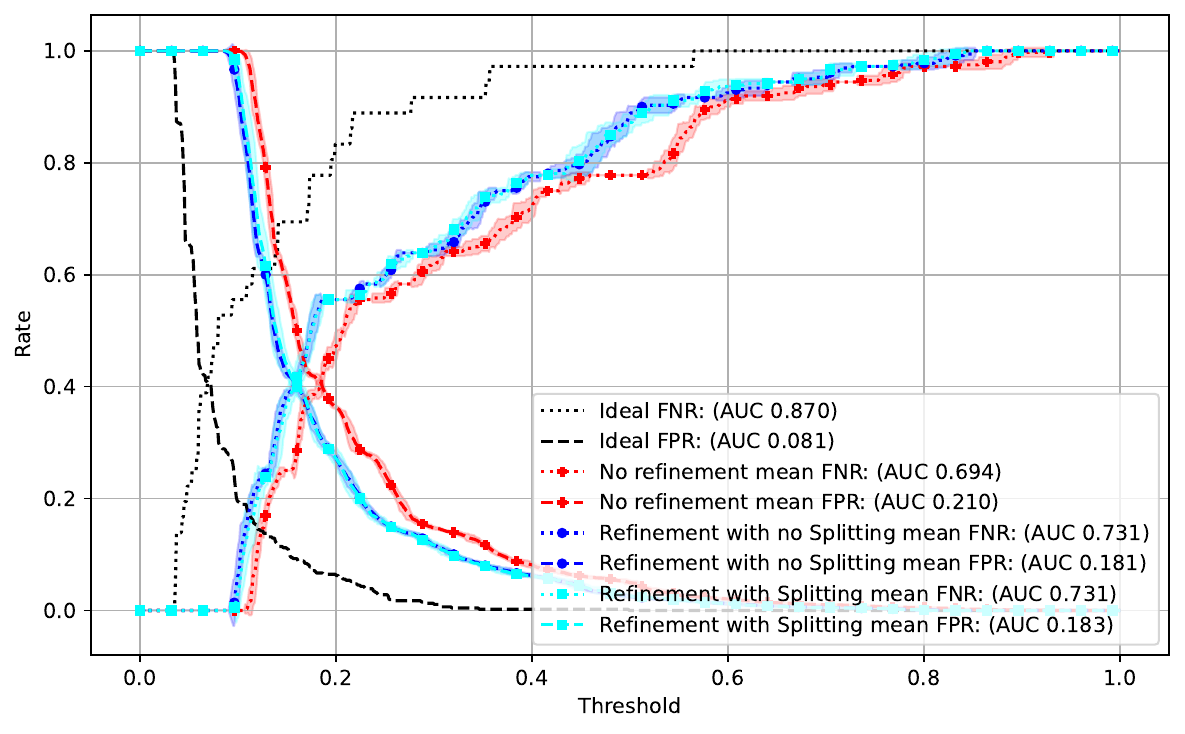}
    \caption{evadeV-6-3, SC = 0.1 - Comparison between learning with and without refinement, in terms of FNR and FPR}
    \label{}
\end{figure}

\begin{figure}[H]
    \centering
    \includegraphics[width=0.8\linewidth]{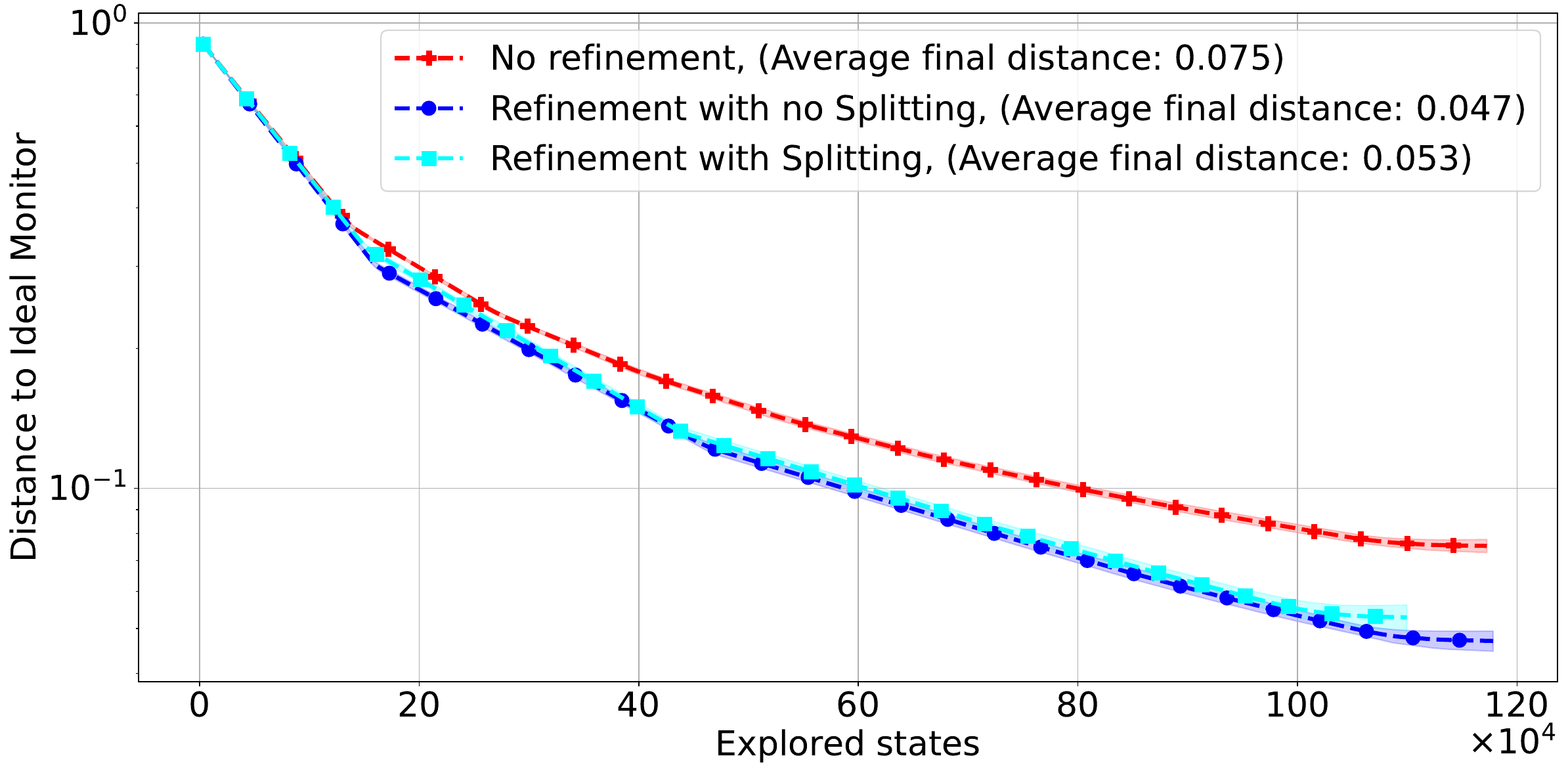}
    \caption{Coarse evadeV-6-3, SC = 0.01 - Comparison between learning with and without refinement, in terms of distance to Ideal Monitor}
    \label{}
\end{figure}

\begin{figure}[H]
    \centering
    \includegraphics[width=0.8\linewidth]{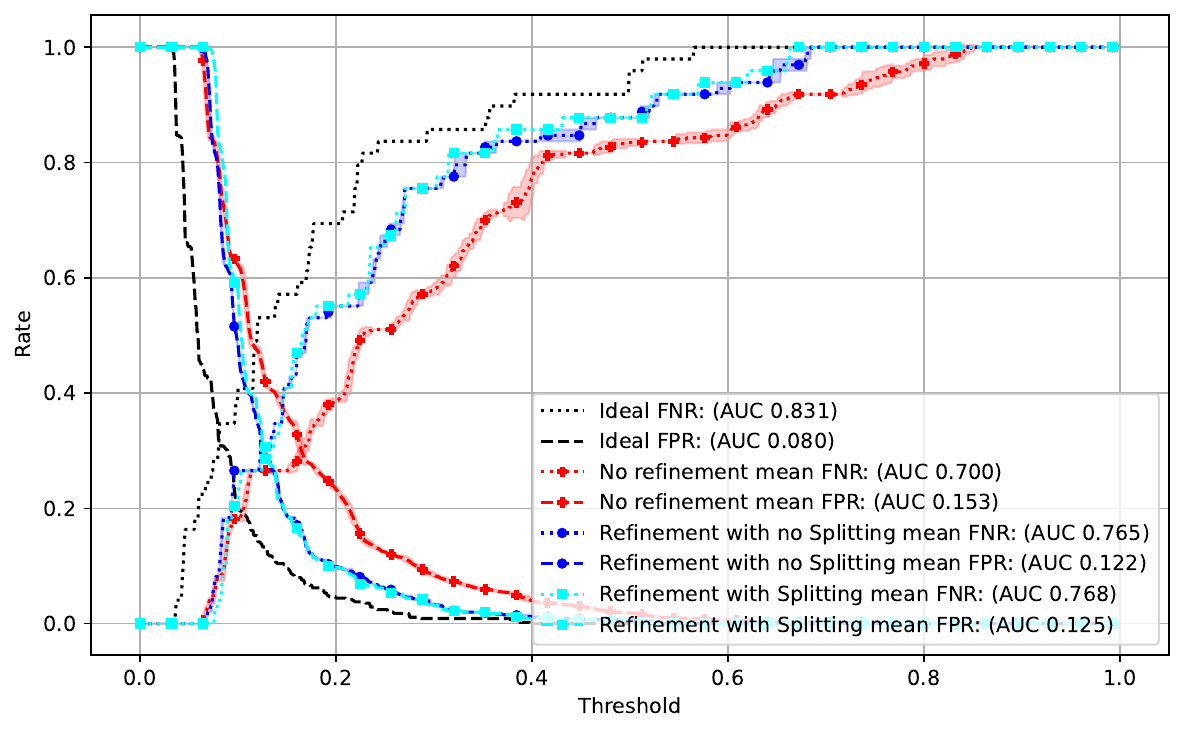}
    \caption{Coarse evadeV-6-3, SC = 0.01 - Comparison between learning with and without refinement, in terms of FNR and FPR}
    \label{}
\end{figure}

\begin{figure}[H]
    \centering
    \includegraphics[width=0.8\linewidth]{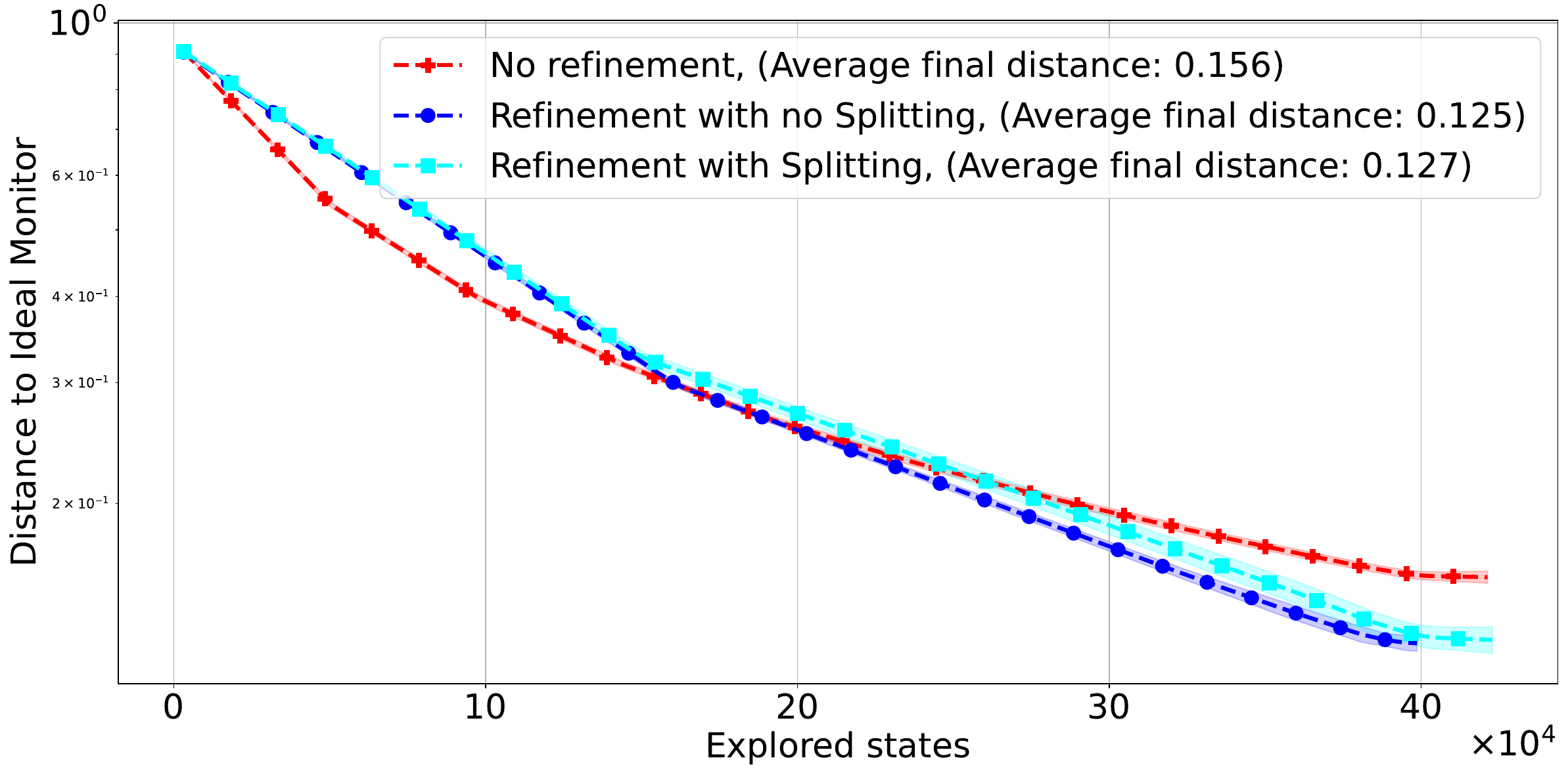}
    \caption{Coarse evadeV-6-3, SC = 0.1 - Comparison between learning with and without refinement, in terms of distance to Ideal Monitor}
    \label{}
\end{figure}

\begin{figure}[H]
    \centering
    \includegraphics[width=0.8\linewidth]{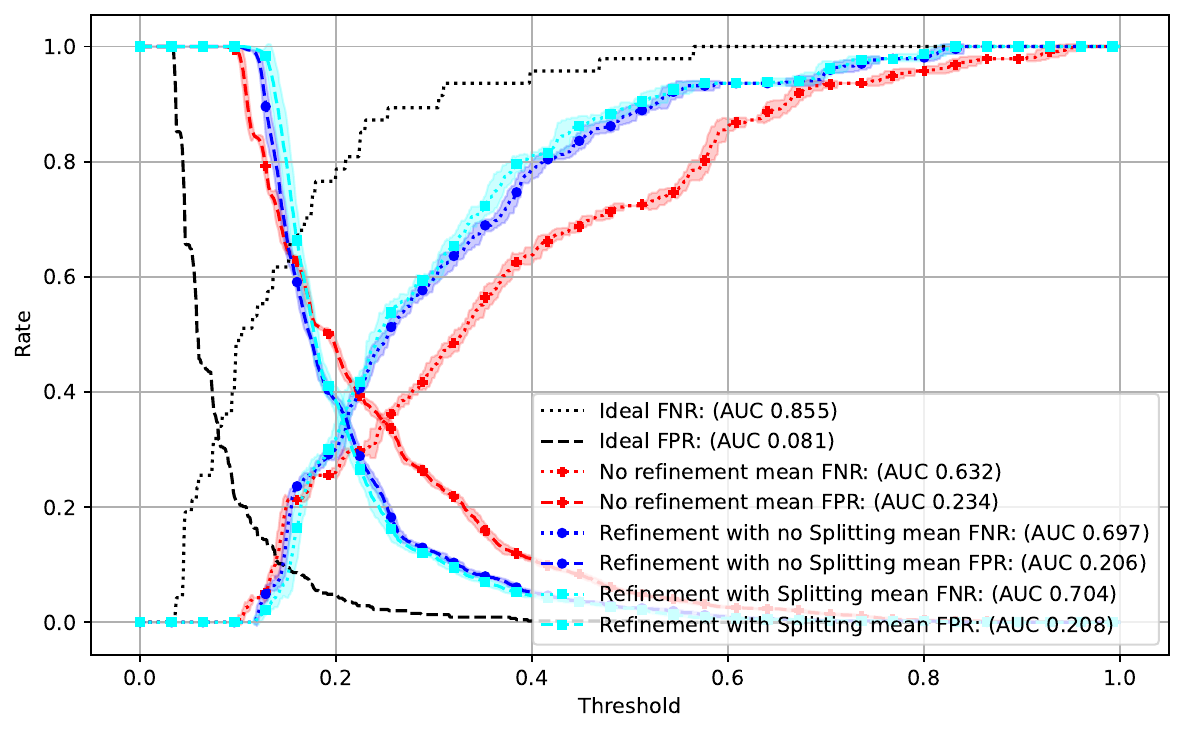}
    \caption{Coarse evadeV-6-3, SC = 0.1 - Comparison between learning with and without refinement, in terms of FNR and FPR}
    \label{}
\end{figure}

\begin{figure}[H]
    \centering
    \includegraphics[width=0.8\linewidth]{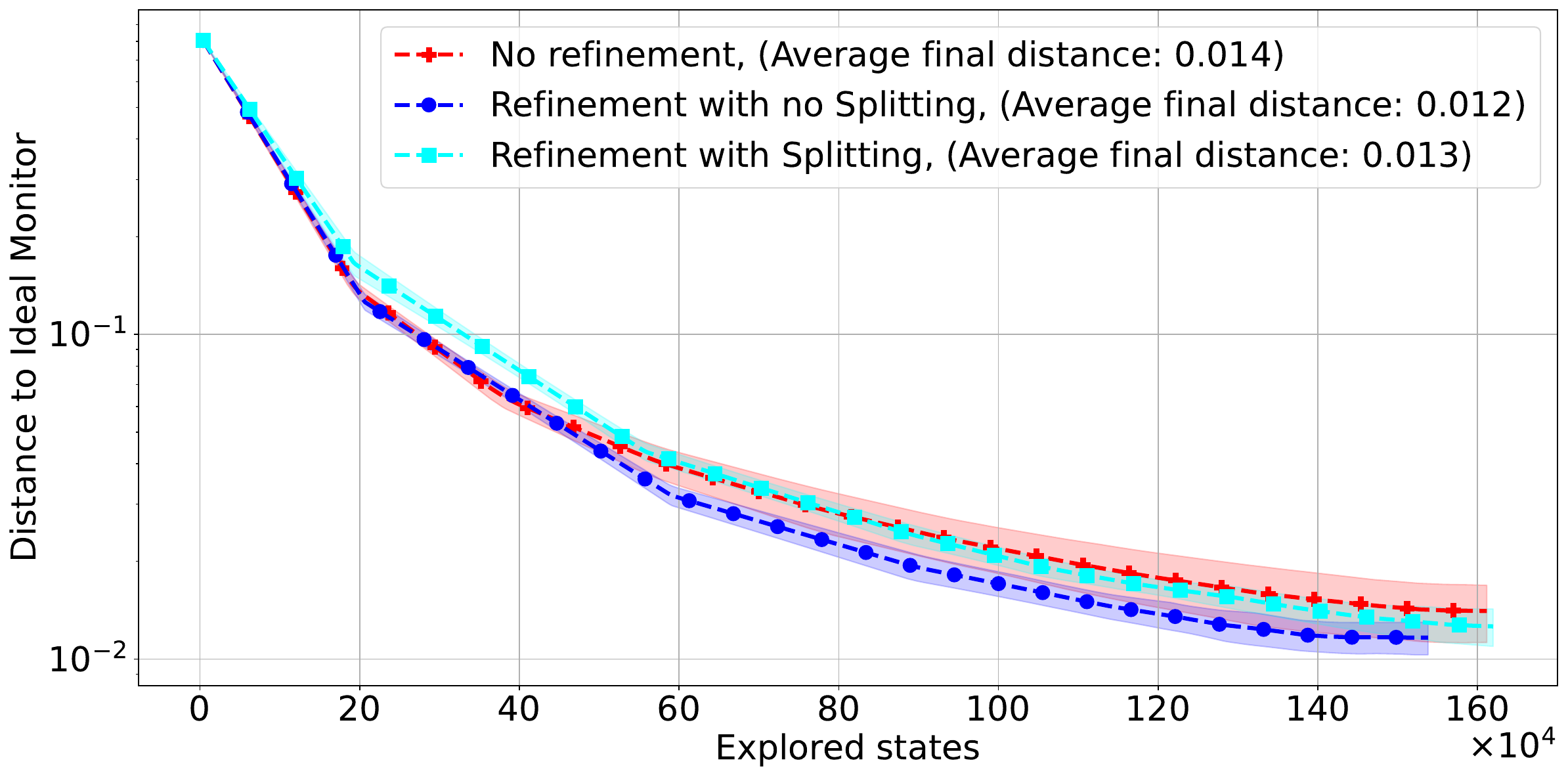}
    \caption{airportA-7-10-10, SC = 0.001 - Comparison between learning with and without refinement, in terms of distance to Ideal Monitor}
    \label{}
\end{figure}

\begin{figure}[H]
    \centering
    \includegraphics[width=0.8\linewidth]{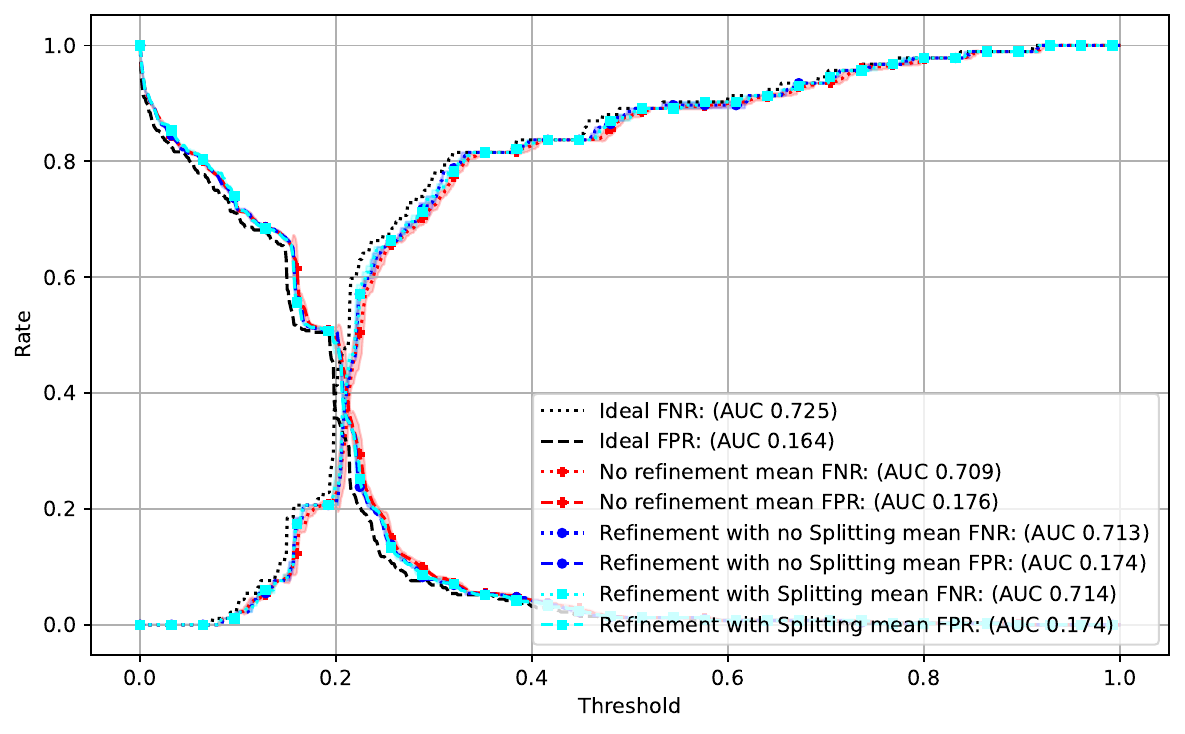}
    \caption{airportA-7-10-10, SC = 0.001 - Comparison between learning with and without refinement, in terms of FNR and FPR}
    \label{}
\end{figure}

\begin{figure}[H]
    \centering
    \includegraphics[width=0.8\linewidth]{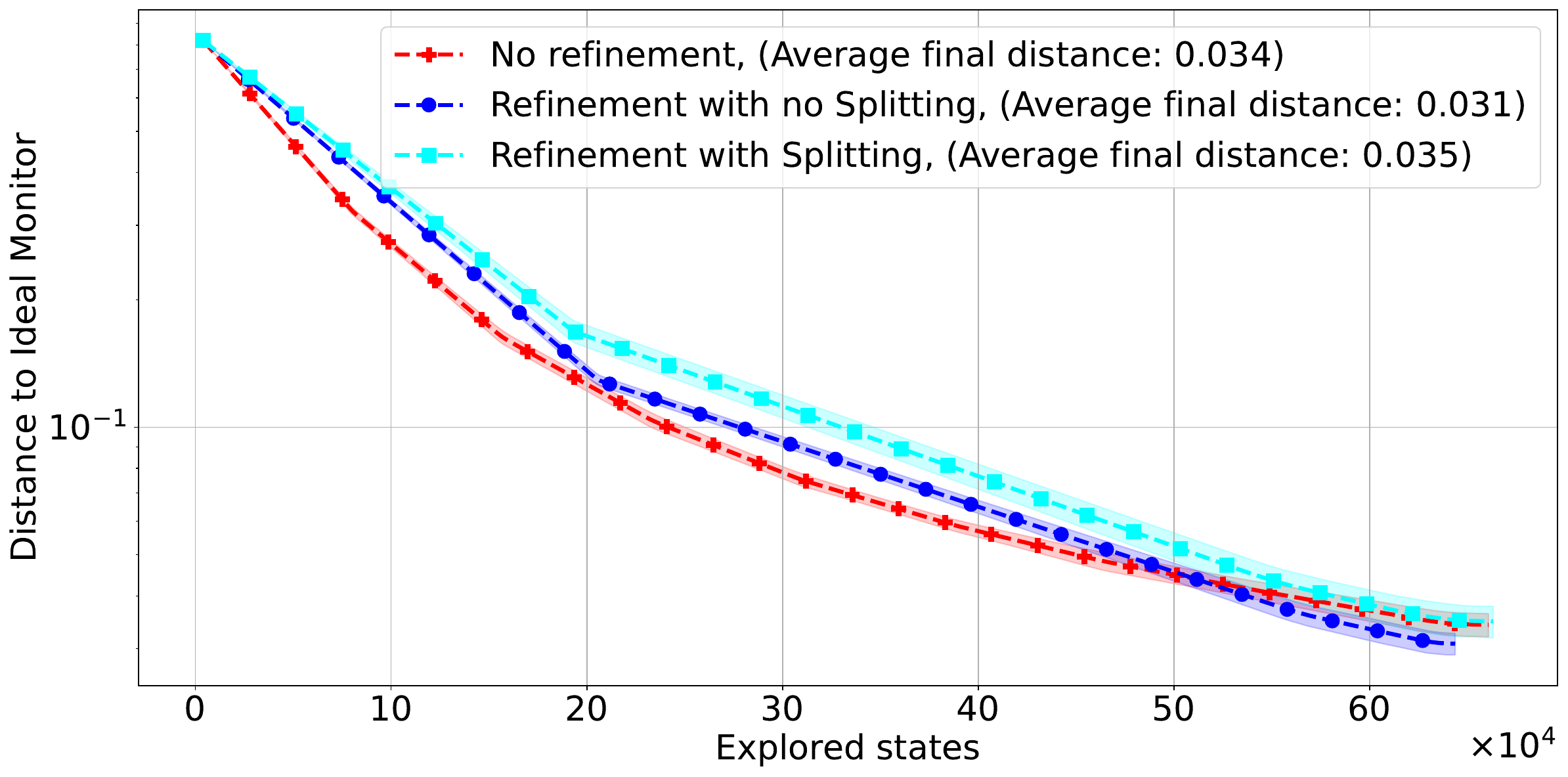}
    \caption{airportA-7-10-10, SC = 0.01 - Comparison between learning with and without refinement, in terms of distance to Ideal Monitor}
    \label{}
\end{figure}

\begin{figure}[H]
    \centering
    \includegraphics[width=0.8\linewidth]{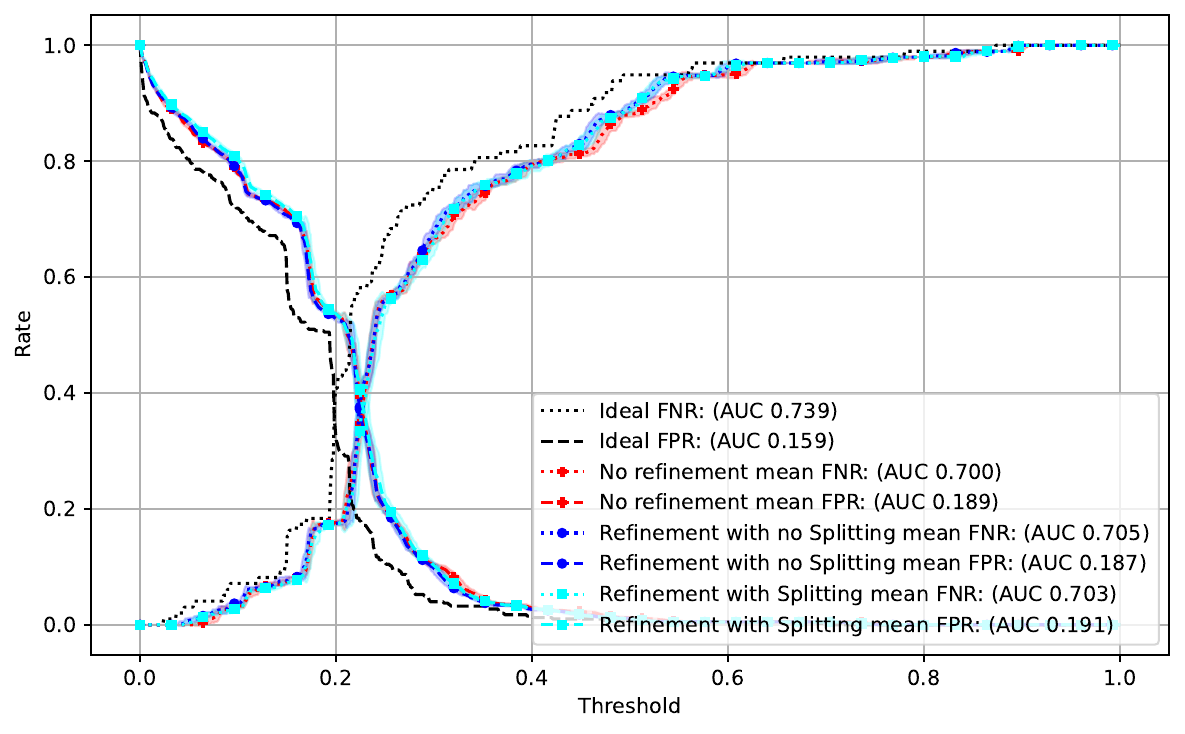}
    \caption{airportA-7-10-10, SC = 0.01 - Comparison between learning with and without refinement, in terms of FNR and FPR}
    \label{}
\end{figure}

\begin{figure}[H]
    \centering
    \includegraphics[width=0.8\linewidth]{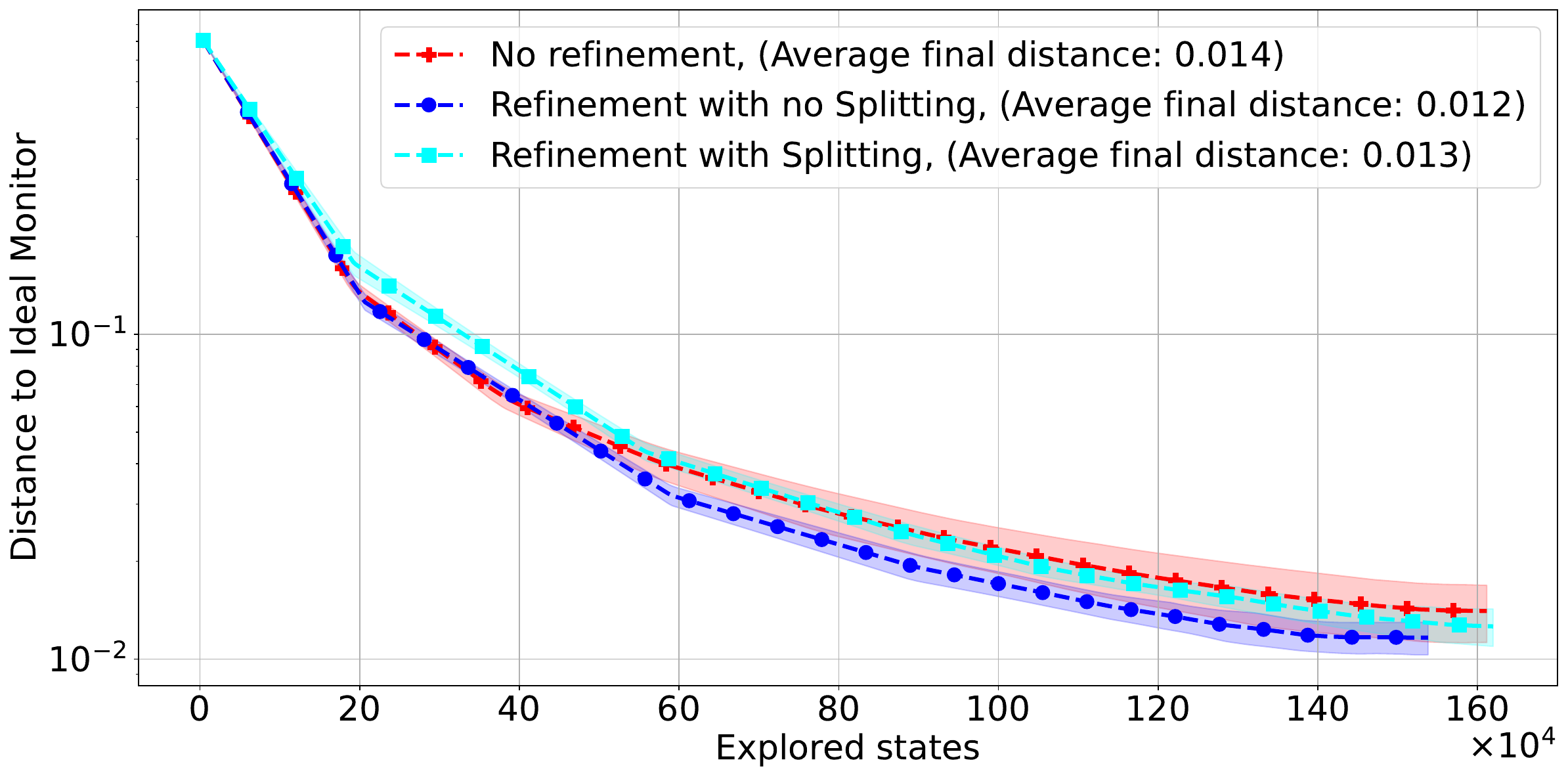}
    \caption{Coarse airportA-7-10-10, SC = 0.01 - Comparison between learning with and without refinement, in terms of distance to Ideal Monitor}
    \label{}
\end{figure}

\begin{figure}[H]
    \centering
    \includegraphics[width=0.8\linewidth]{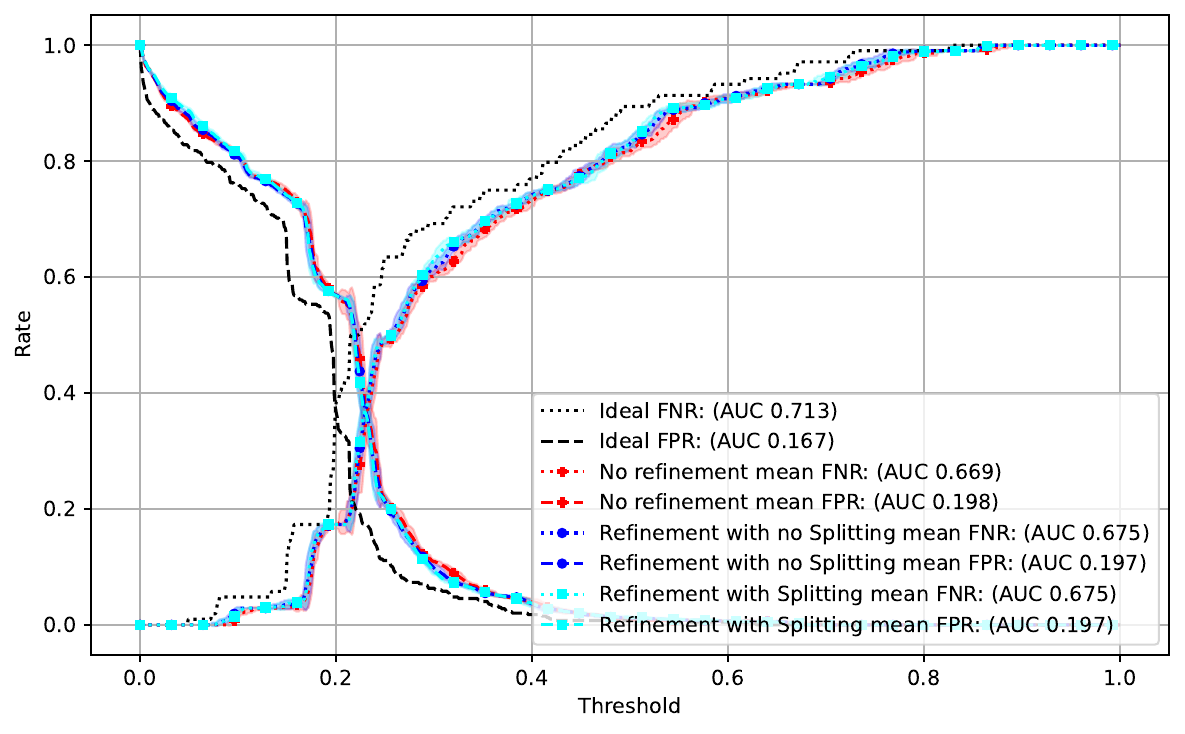}
    \caption{Coarse airportA-7-10-10, SC = 0.01 - Comparison between learning with and without refinement, in terms of FNR and FPR}
    \label{}
\end{figure}

\begin{figure}[H]
    \centering
    \includegraphics[width=0.8\linewidth]{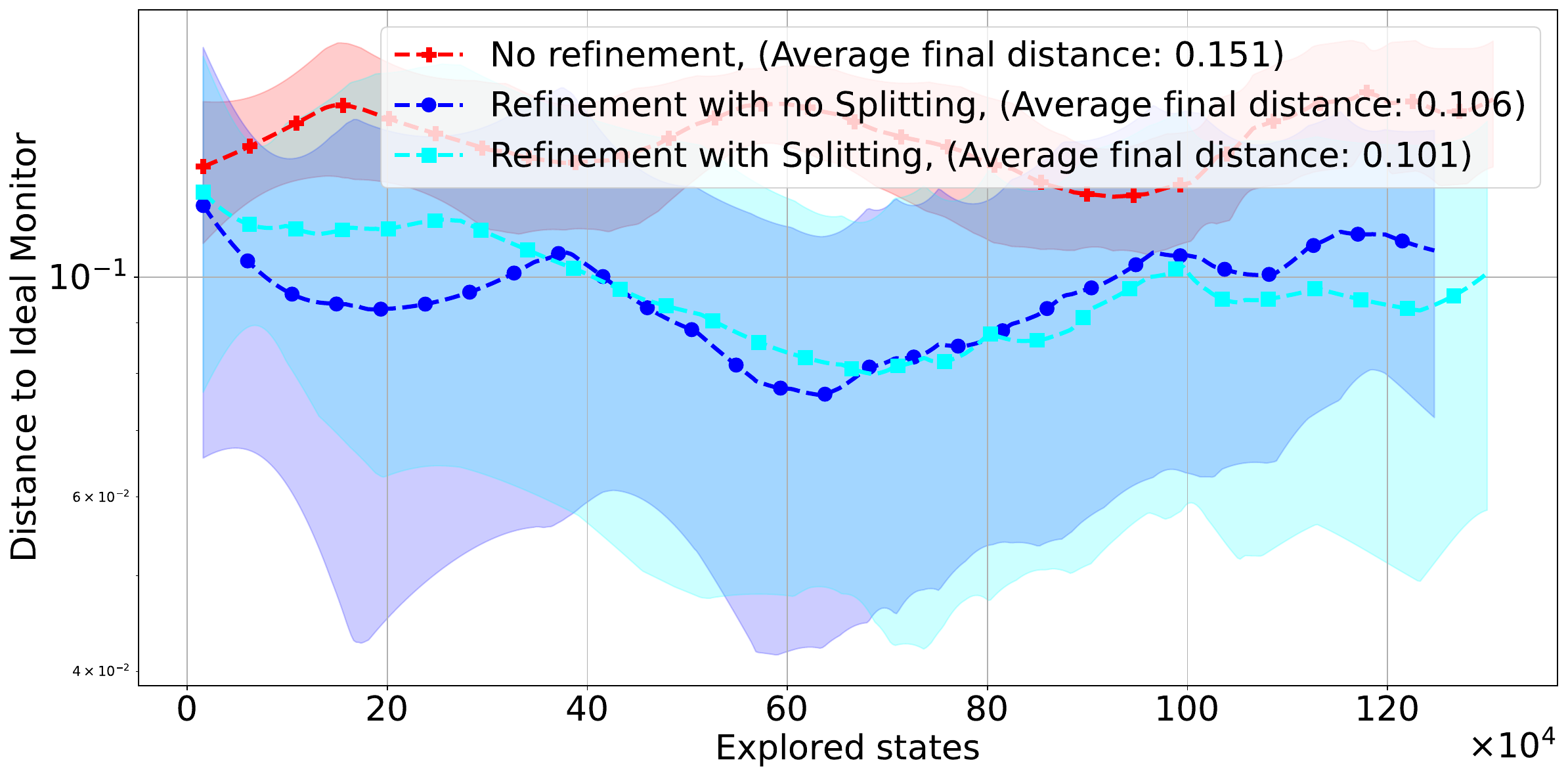}
    \caption{Coarse airportA-7-10-10, SC = 0.1 - Comparison between learning with and without refinement, in terms of distance to Ideal Monitor}
    \label{}
\end{figure}

\begin{figure}[H]
    \centering
    \includegraphics[width=0.8\linewidth]{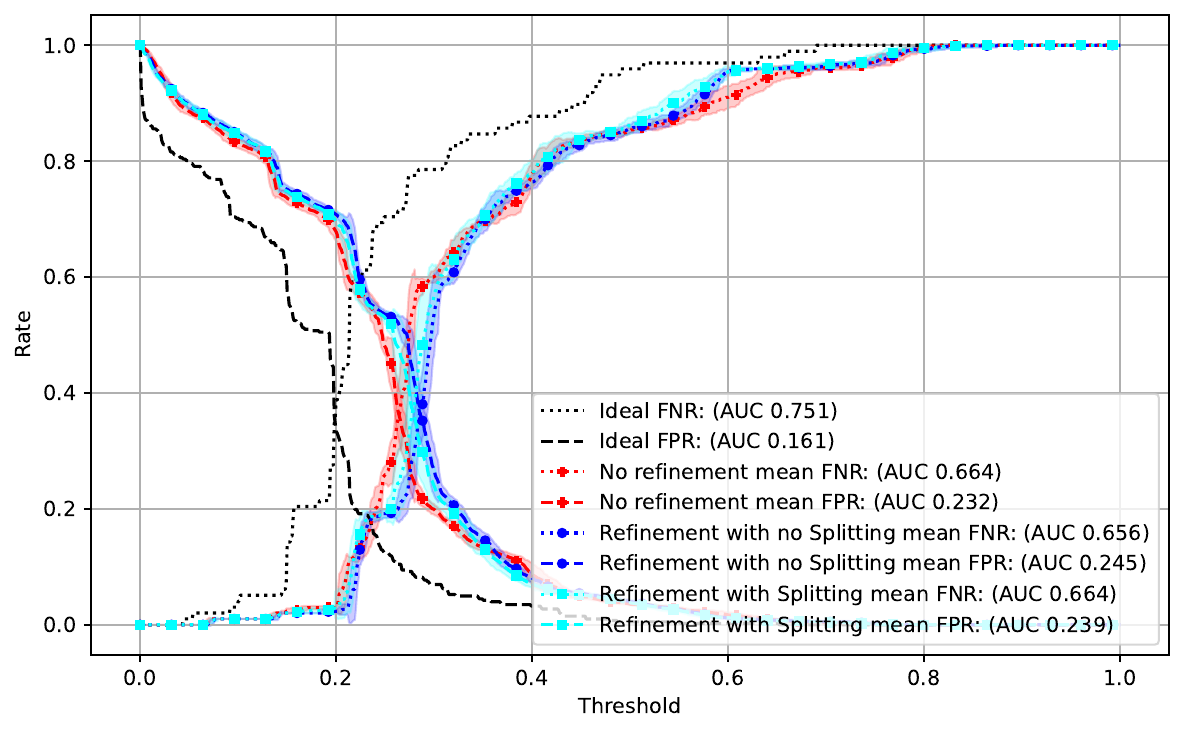}
    \caption{Coarse airportA-7-10-10, SC = 0.1 - Comparison between learning with and without refinement, in terms of FNR and FPR}
    \label{}
\end{figure}

\section{Experimental evaluation of monitors for aiportA-7-40-20 and coarse aiportB-7-40-20 benchmarks}

The experimental evaluation of the benchmarks aiportA-7-40-20 and coarse aiportB-7-40-20 is discussed separately. This is motivated by the fact that none of the considered methods learned reasonable monitors for these benchmarks. 

\label{apx:big_benchmarks}
\subsection{Model-based vs Model-free monitoring}
We include the results of the experimental comparison of model-based and model-free monitors for for aiportA-7-40-20 and coarse aiportB-7-40-20 benchmarks. 

\Cref{ariB-corase_model_based_free} due to memory issues reason is based od 100 test samples.

\begin{figure}[H]
    \centering
    \includegraphics[width=0.85\linewidth]{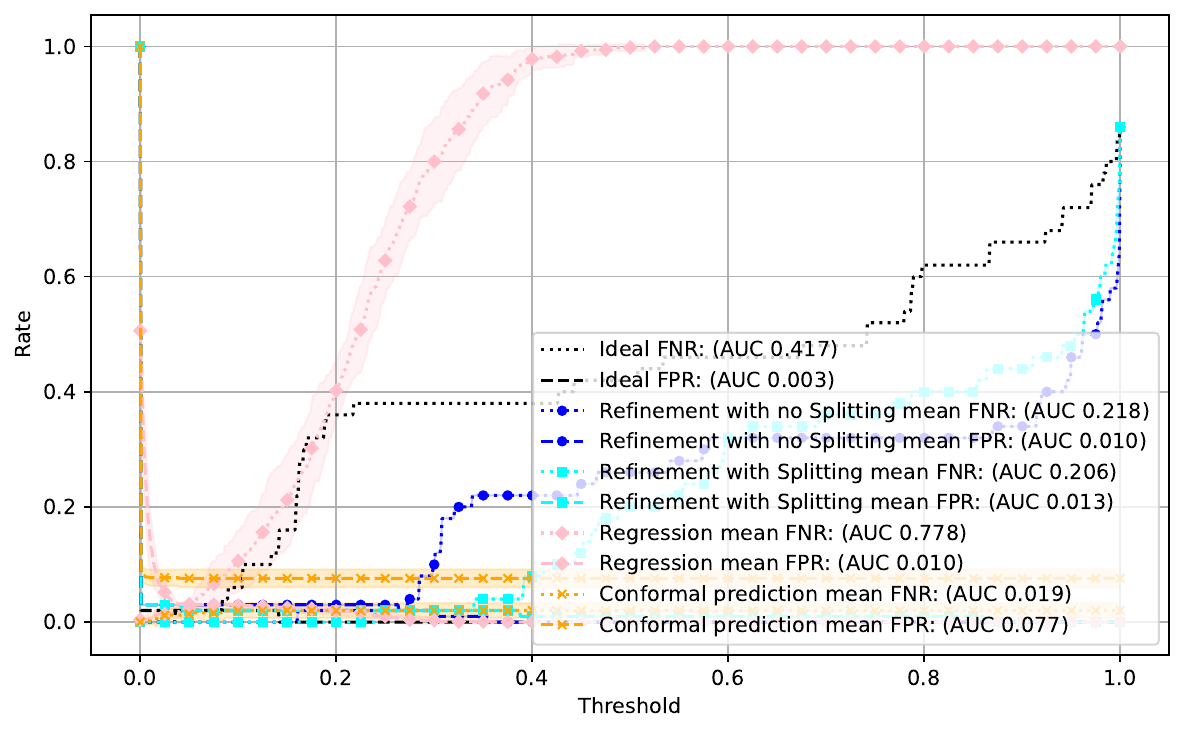}
    \caption{airportA-7-40-20, SC =0.01 - FNR and FPR com-
parison between model-based and model-free methods}
    \label{}
\end{figure}

\begin{figure}[H]
    \centering
    \includegraphics[width=0.85\linewidth]{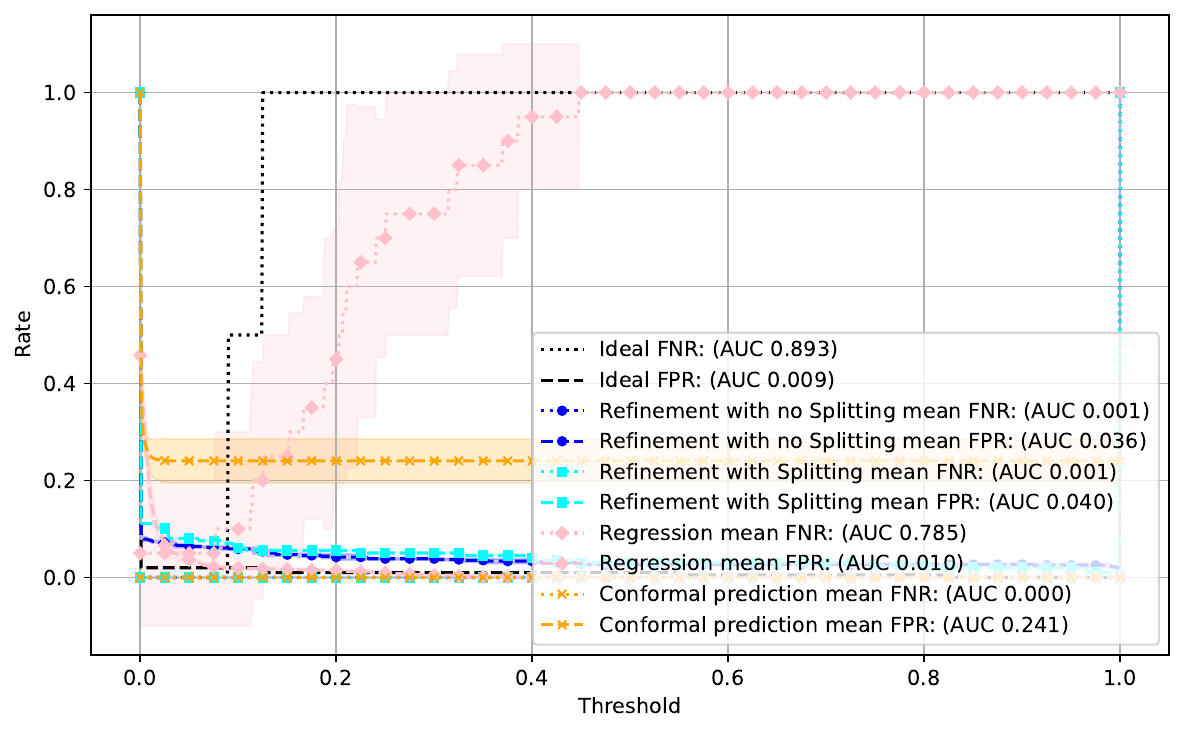}
    \caption{airportA-7-40-20, SC = 0.1 - FNR and FPR com-
parison between model-based and model-free methods}
    \label{}
\end{figure}

\begin{figure}[H]
    \centering
    \includegraphics[width=0.85\linewidth]{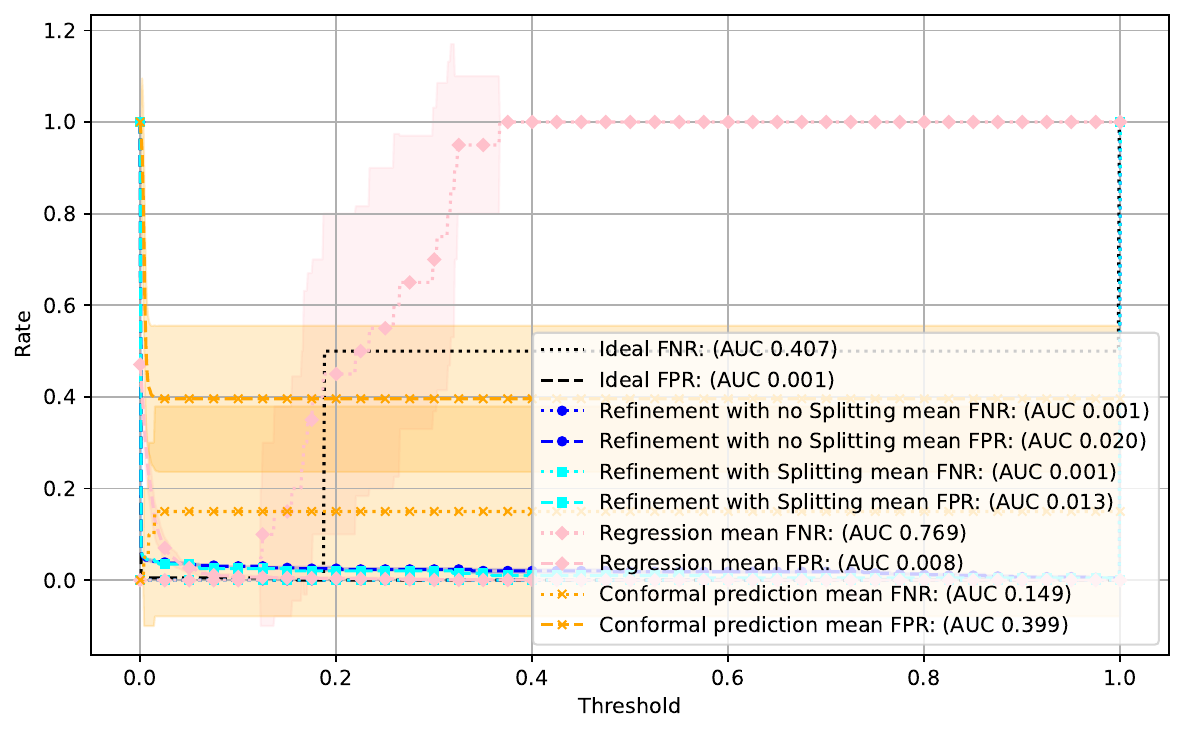}
    \caption{Coarse airportB-7-40-20, SC = 0.01 - FNR and FPR com-
parison between model-based and model-free methods}
    \label{ariB-corase_model_based_free}
\end{figure}

\begin{figure}[H]
    \centering
    \includegraphics[width=0.85\linewidth]{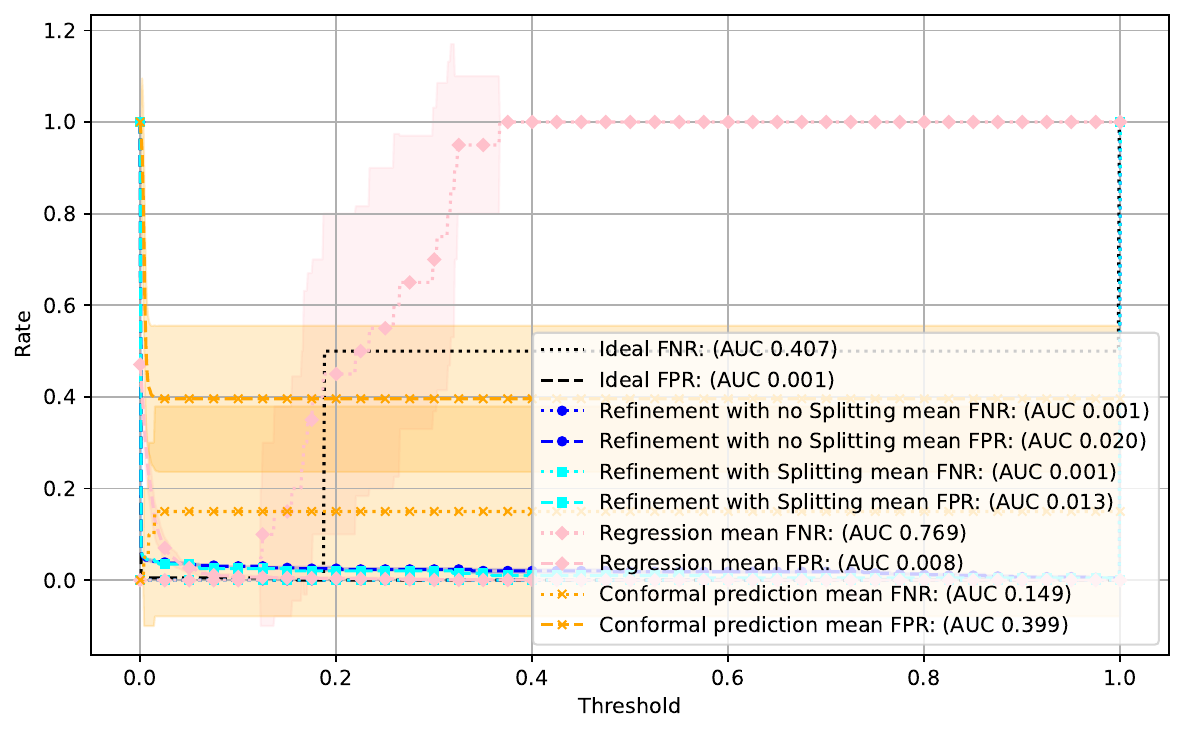}
    \caption{Coarse airportB-7-40-20, SC = 0.1 - FNR and FPR com-
parison between model-based and model-free methods}
    \label{}
\end{figure}

\subsection{iHMM vs HMM monitoring}

We include the results of the comparison of iHMM and HMM monitoring for airportA-7-40-20 and coarse airportB-7-40-20. 

\begin{figure}[H]
    \centering
    \includegraphics[width=0.85\linewidth]{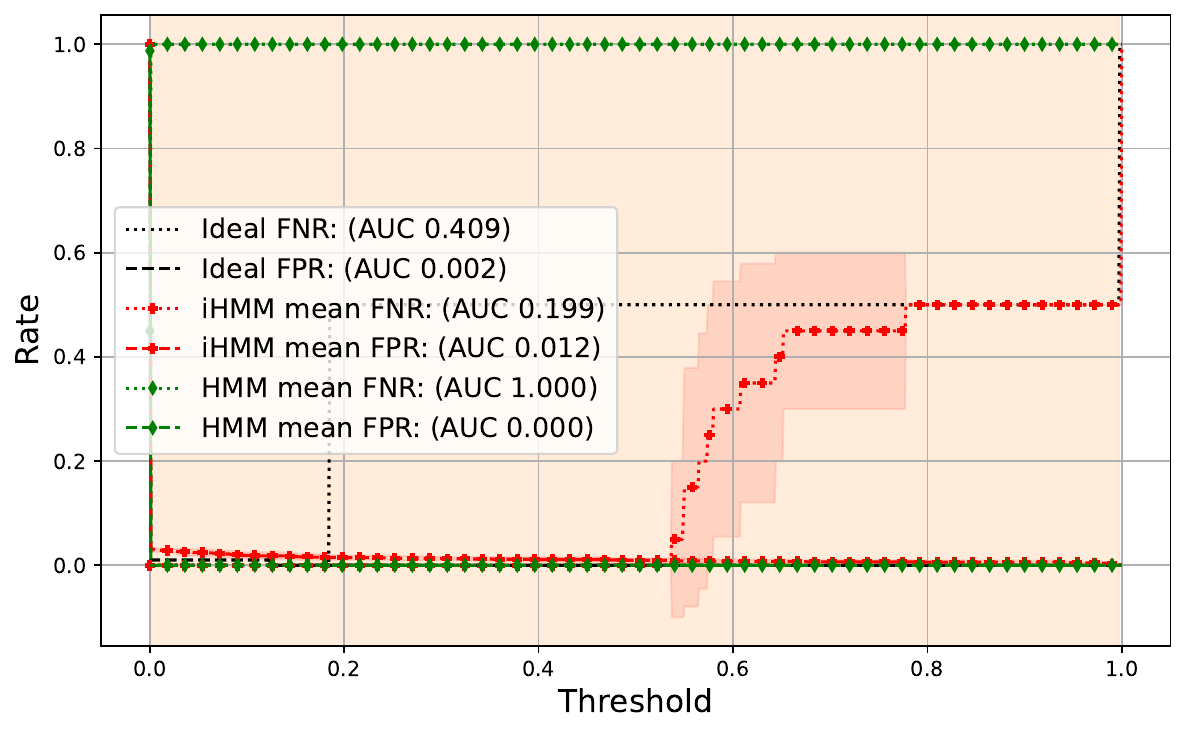}
    \caption{airportA-7-40-20, SC =0.01 - FNR and FPR comparison between iHMM and HMM}
    \label{}
\end{figure}

\begin{figure}[H]
    \centering
    \includegraphics[width=0.85\linewidth]{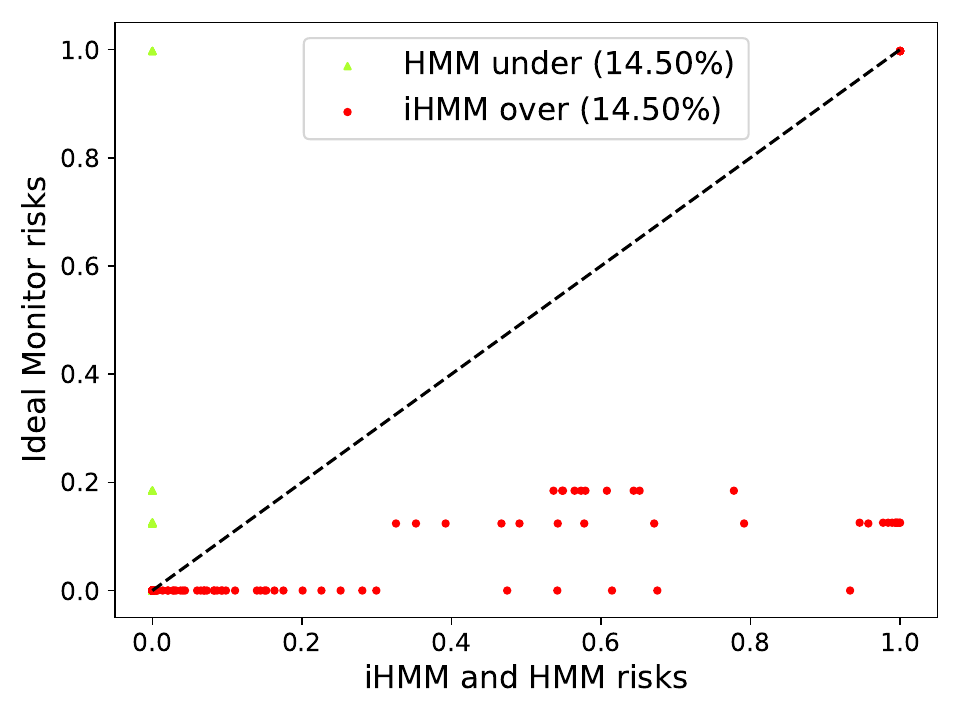}
    \caption{airportA-7-40-20, SC = 0.01 - risk estimation subject to target, comparison between iHMM and HMM0.0 0.2 0.4}
    \label{}
\end{figure}

\begin{figure}[H]
    \centering
    \includegraphics[width=0.85\linewidth]{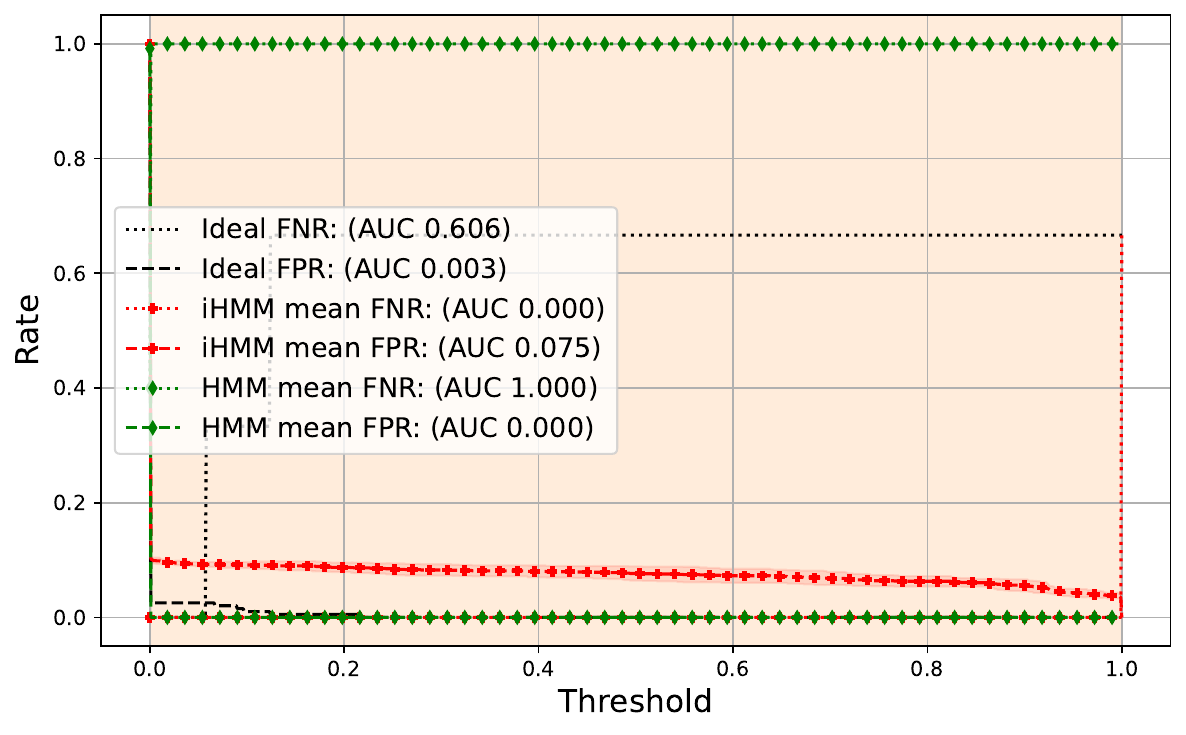}
    \caption{airportA-7-40-20, SC = 0.1 - FNR and FPR comparison between iHMM and HMM}
    \label{}
\end{figure}

\begin{figure}[H]
    \centering
    \includegraphics[width=0.85\linewidth]{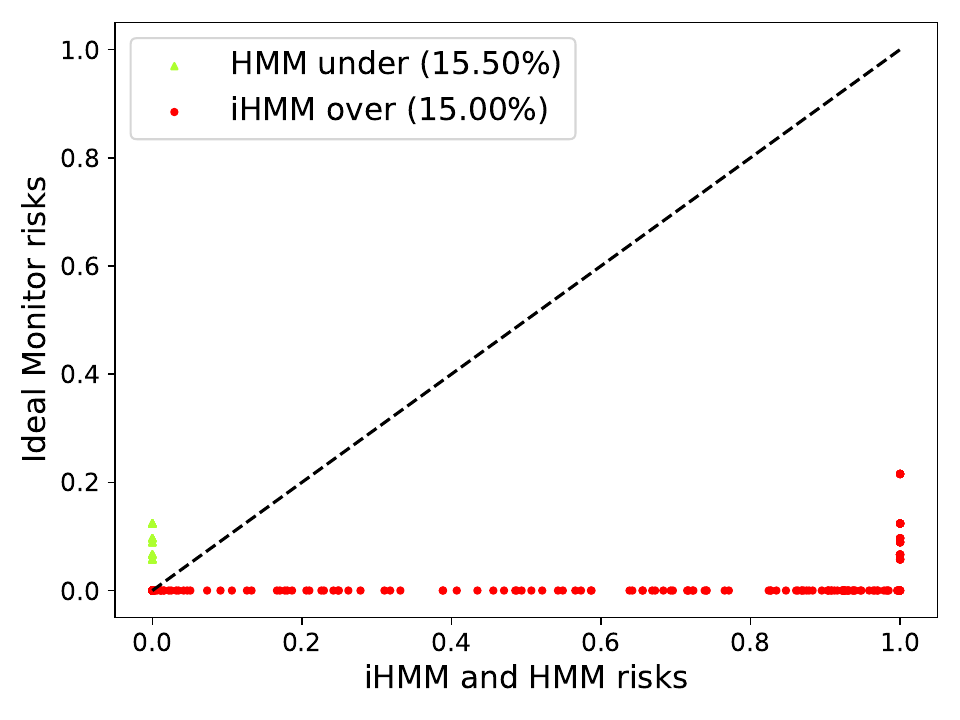}
    \caption{airportA-7-40-20, SC = 0.1 - risk estimation subject to target, comparison between iHMM and HMM0.0 0.2 0.4}
    \label{}
\end{figure}

\begin{figure}[H]
    \centering
    \includegraphics[width=0.85\linewidth]{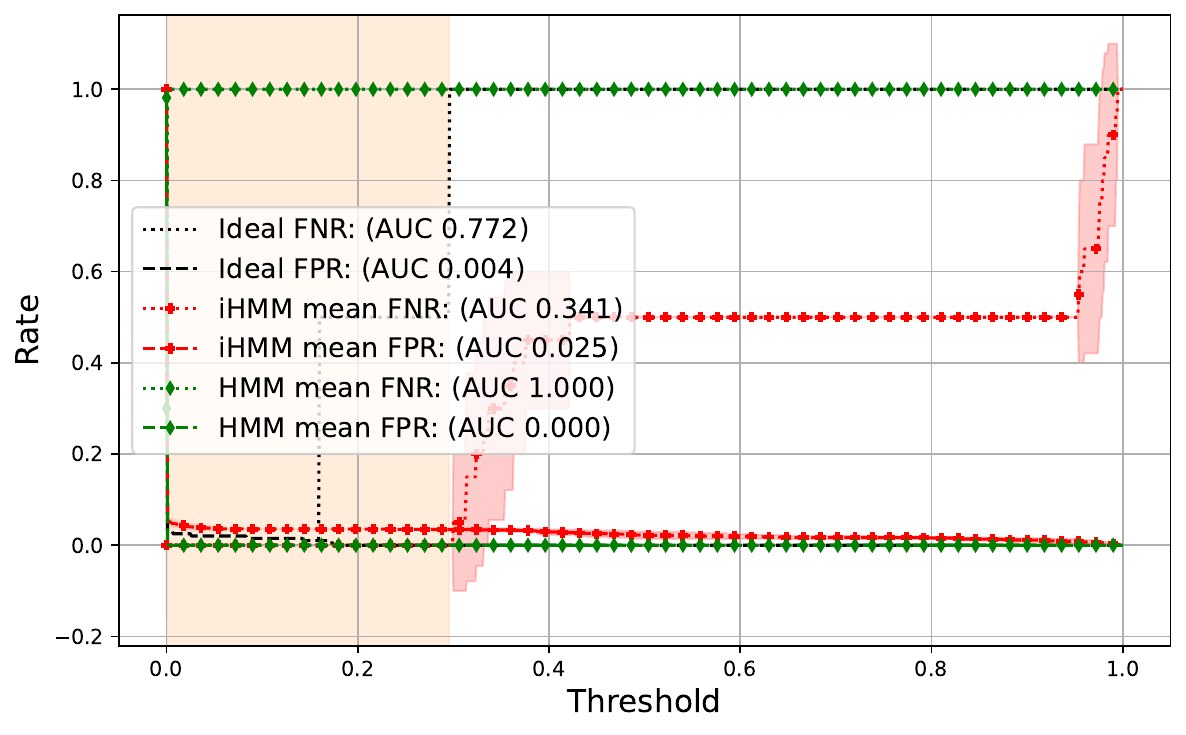}
    \caption{Coarse airportB-7-40-20, SC = 0.01 - FNR and FPR comparison between iHMM and HMM}
    \label{}
\end{figure}

\begin{figure}[H]
    \centering
    \includegraphics[width=0.85\linewidth]{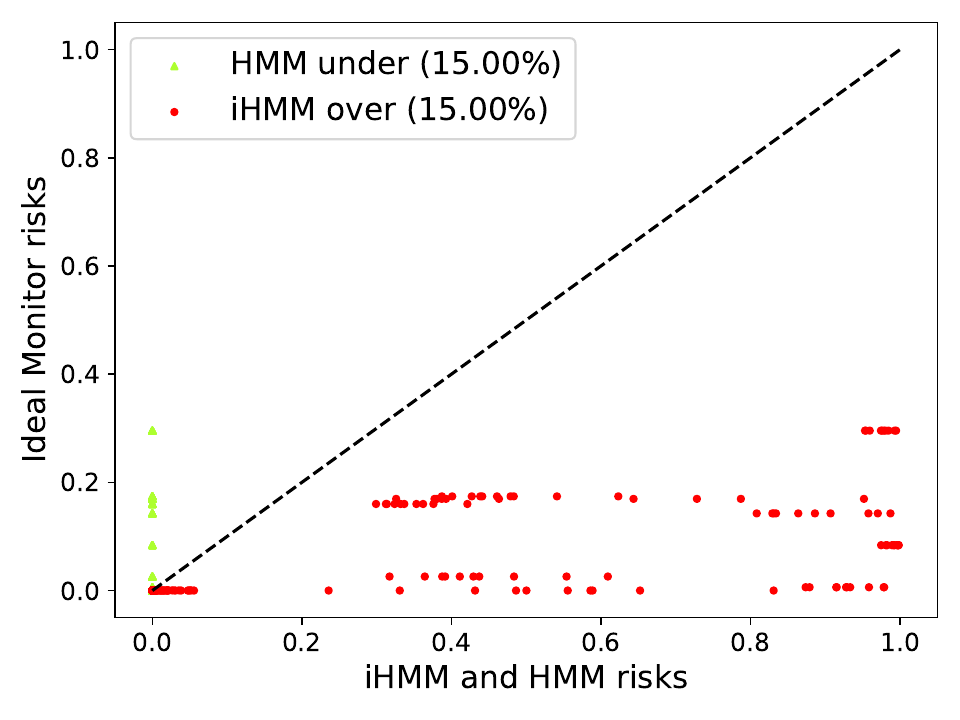}
    \caption{Coarse airportB-7-40-20, SC = 0.01 - risk estimation subject to target, comparison between iHMM and HMM0.0 0.2 0.4}
    \label{}
\end{figure}

\begin{figure}[H]
    \centering
    \includegraphics[width=0.85\linewidth]{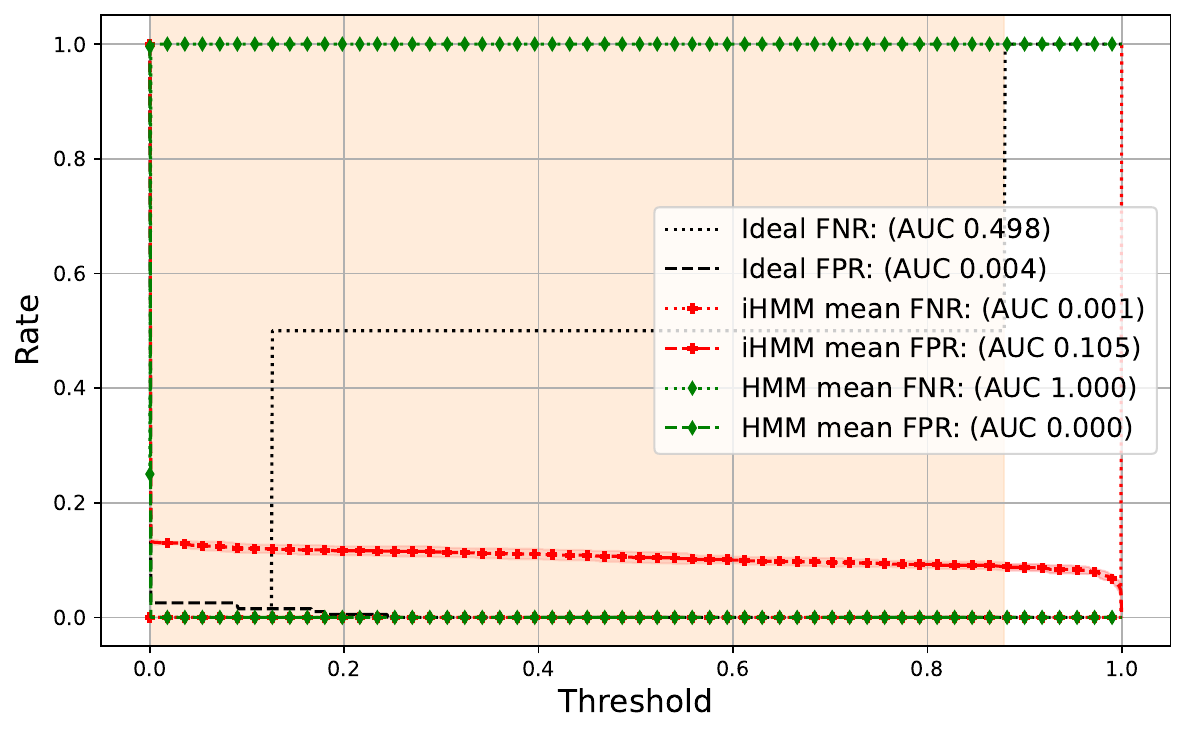}
    \caption{Coarse airportB-7-40-20, SC = 0.1 - FNR and FPR comparison between iHMM and HMM}
    \label{}
\end{figure}

\begin{figure}[H]
    \centering
    \includegraphics[width=0.85\linewidth]{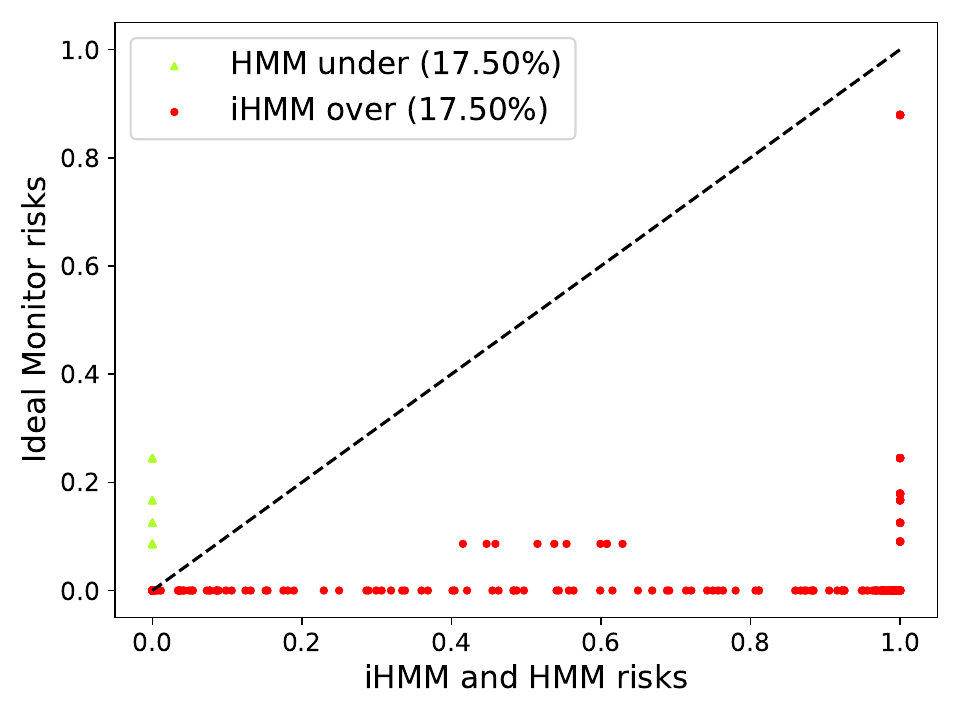}
    \caption{Coarse airportB-7-40-20, SC = 0,1 - risk estimation subject to target, comparison between iHMM and HMM0.0 0.2 0.4}
    \label{}
\end{figure}

\subsection{No refinement vs refinement learning}
We include the results of the comparison of refinement and no refinement learning for airportA-7-40-20 and coarse airportB-7-40-20. 

\begin{figure}[H]
    \centering
    \includegraphics[width=0.85\linewidth]{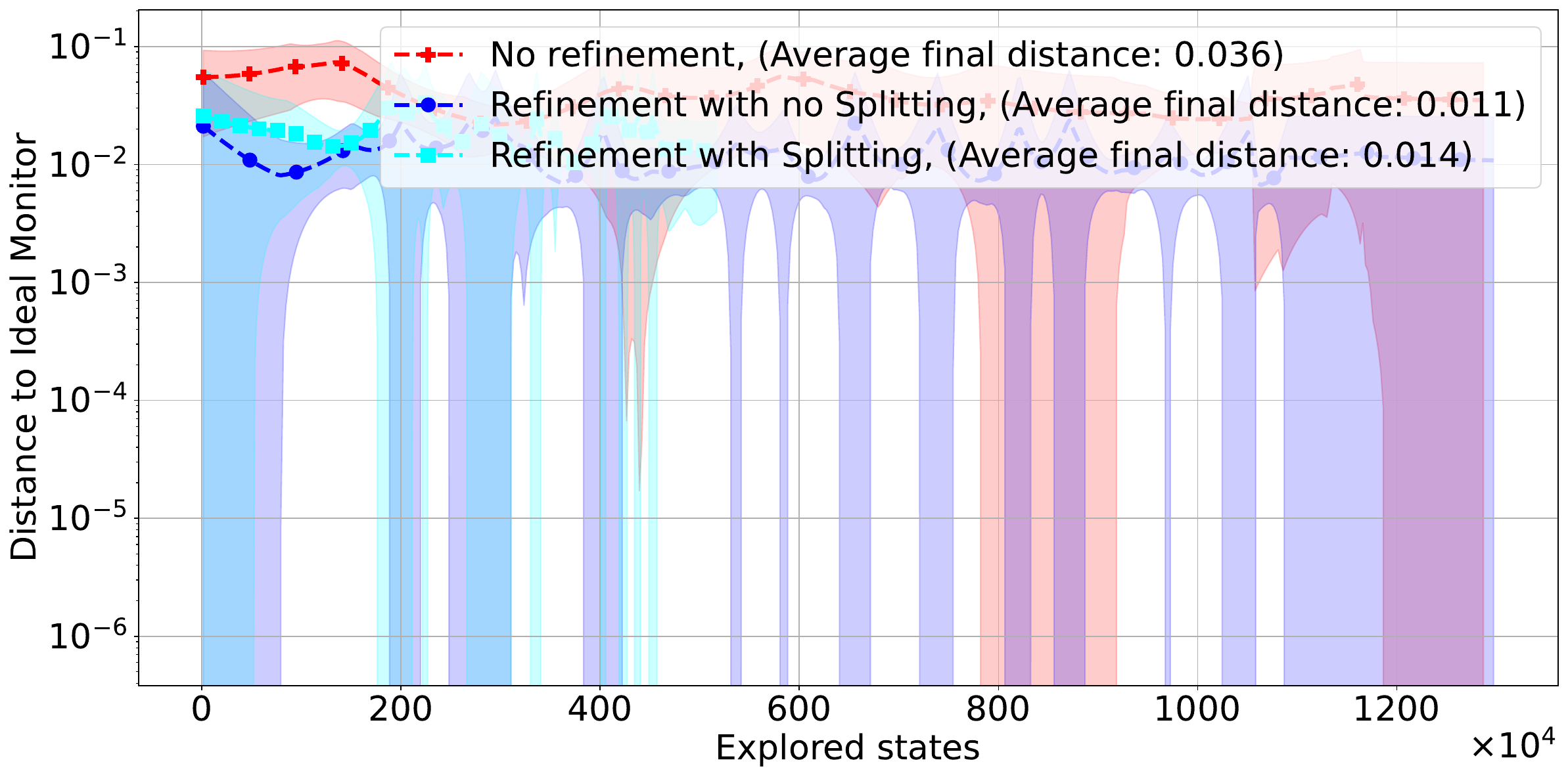}
    \caption{airportA-7-40-20, SC =0.01 - Comparison between learning with and without refinement, in terms of distance to Ideal Monitor}
    \label{}
\end{figure}

\begin{figure}[H]
    \centering
    \includegraphics[width=0.85\linewidth]{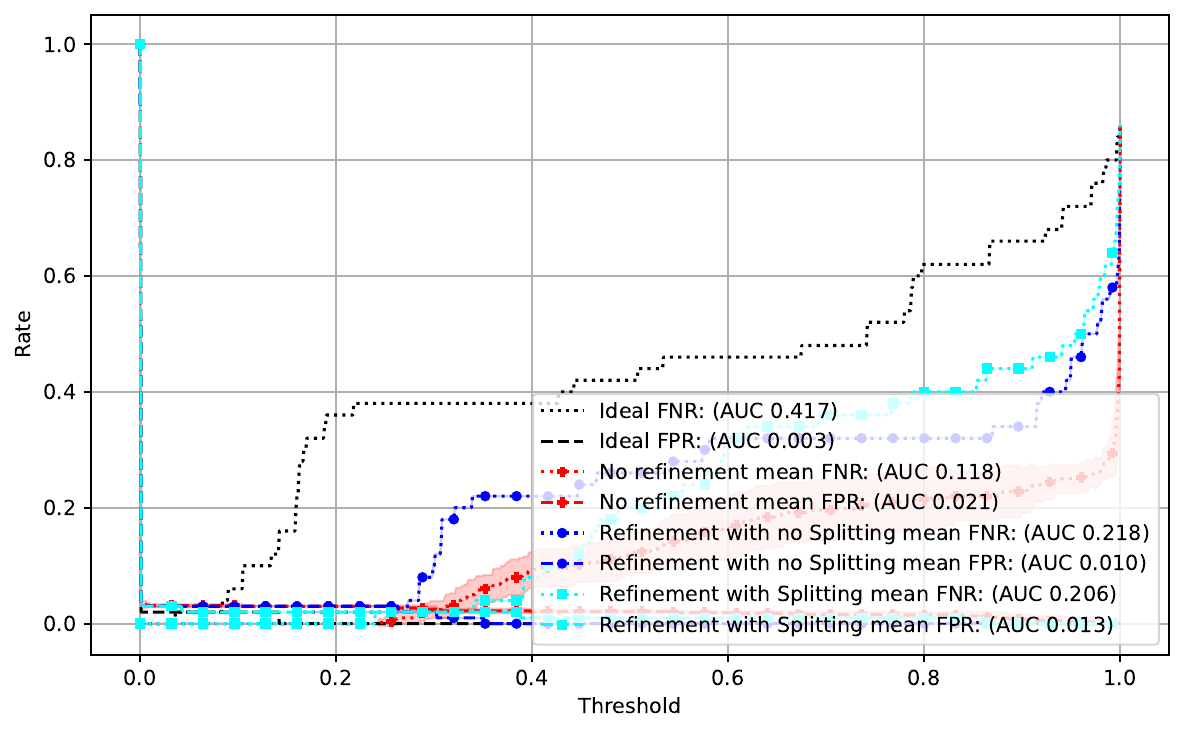}
    \caption{airportA-7-40-20, SC = 0.01 - Comparison between learning with and without refinement, in terms of FNR and FPR}
    \label{}
\end{figure}

\begin{figure}[H]
    \centering
    \includegraphics[width=0.85\linewidth]{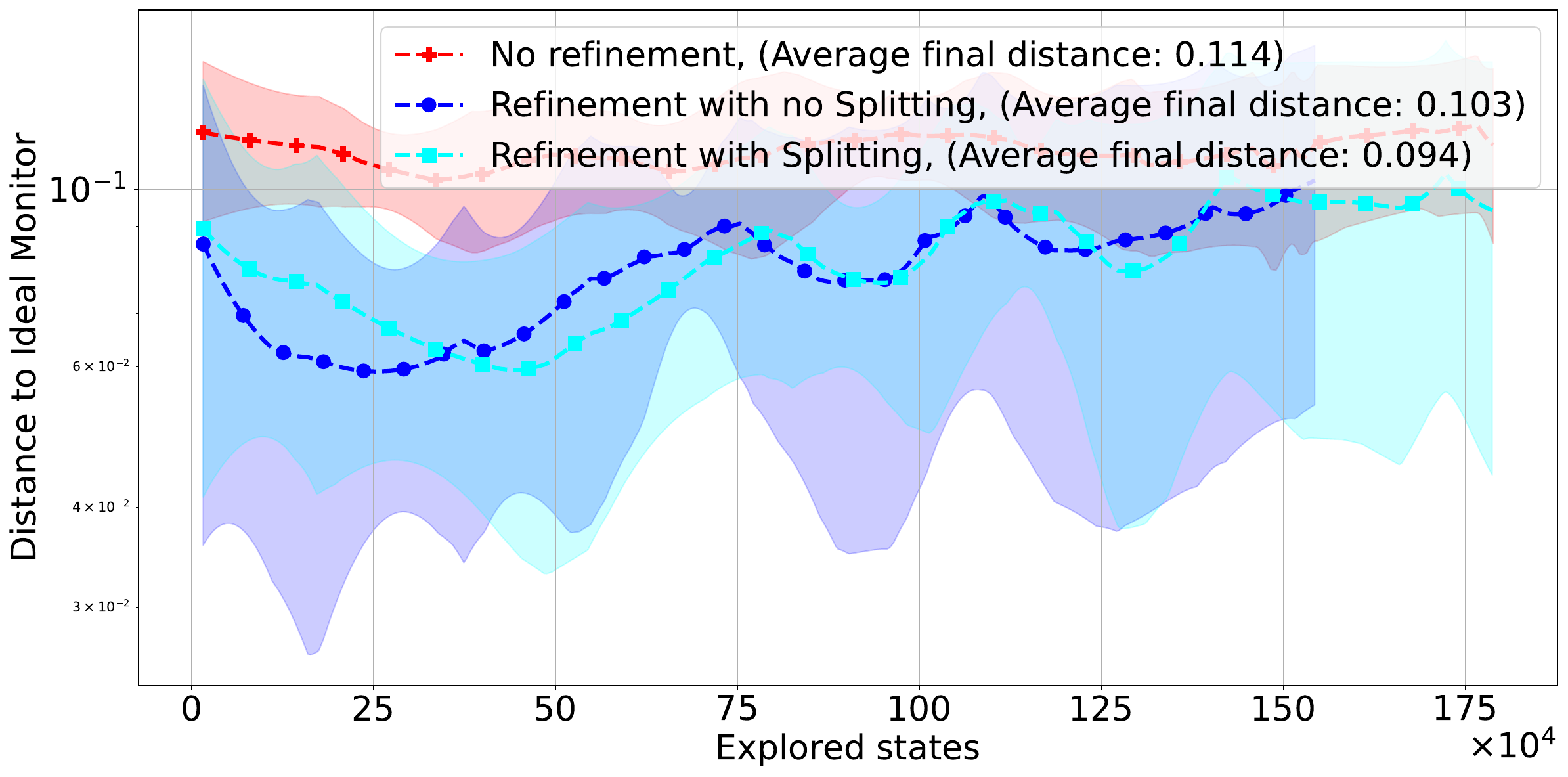}
    \caption{airportA-7-40-20, SC = 0.1 - Comparison between learning with and without refinement, in terms of distance to Ideal Monitor}
    \label{}
\end{figure}

\begin{figure}[H]
    \centering
    \includegraphics[width=0.85\linewidth]{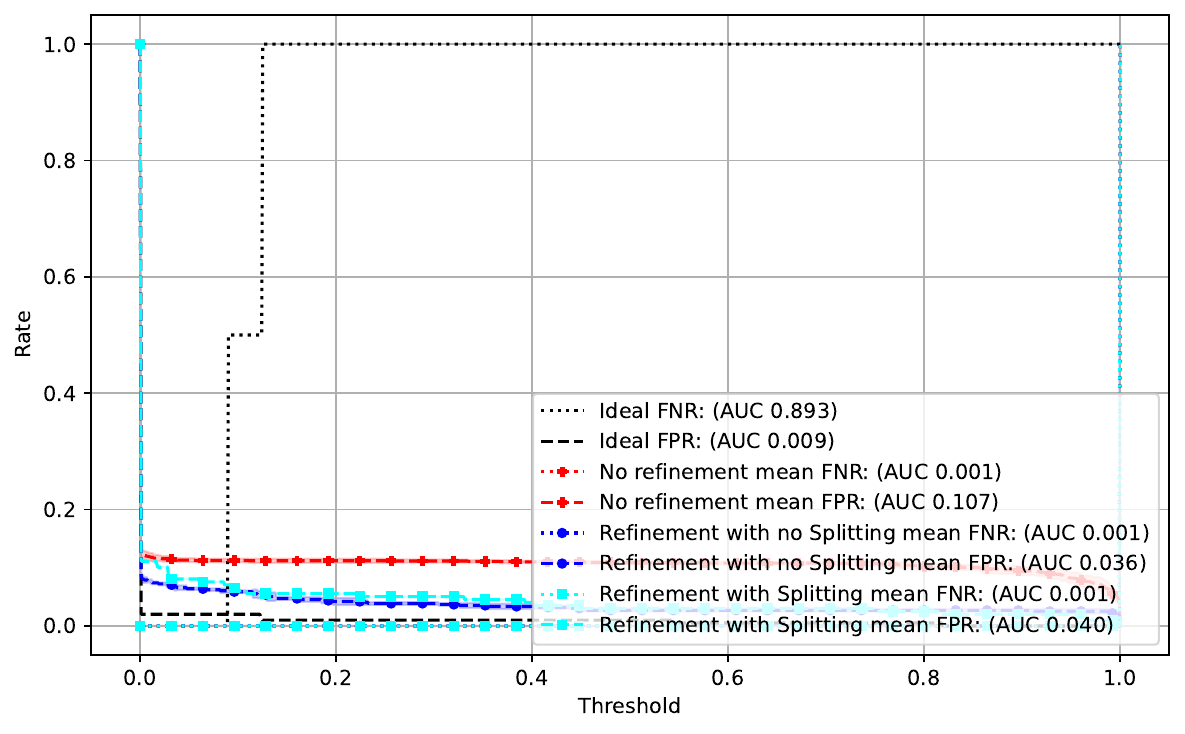}
    \caption{airportA-7-40-20, SC = 0.1 - Comparison between learning with and without refinement, in terms of FNR and FPR}
    \label{}
\end{figure}

\begin{figure}[H]
    \centering
    \includegraphics[width=0.85\linewidth]{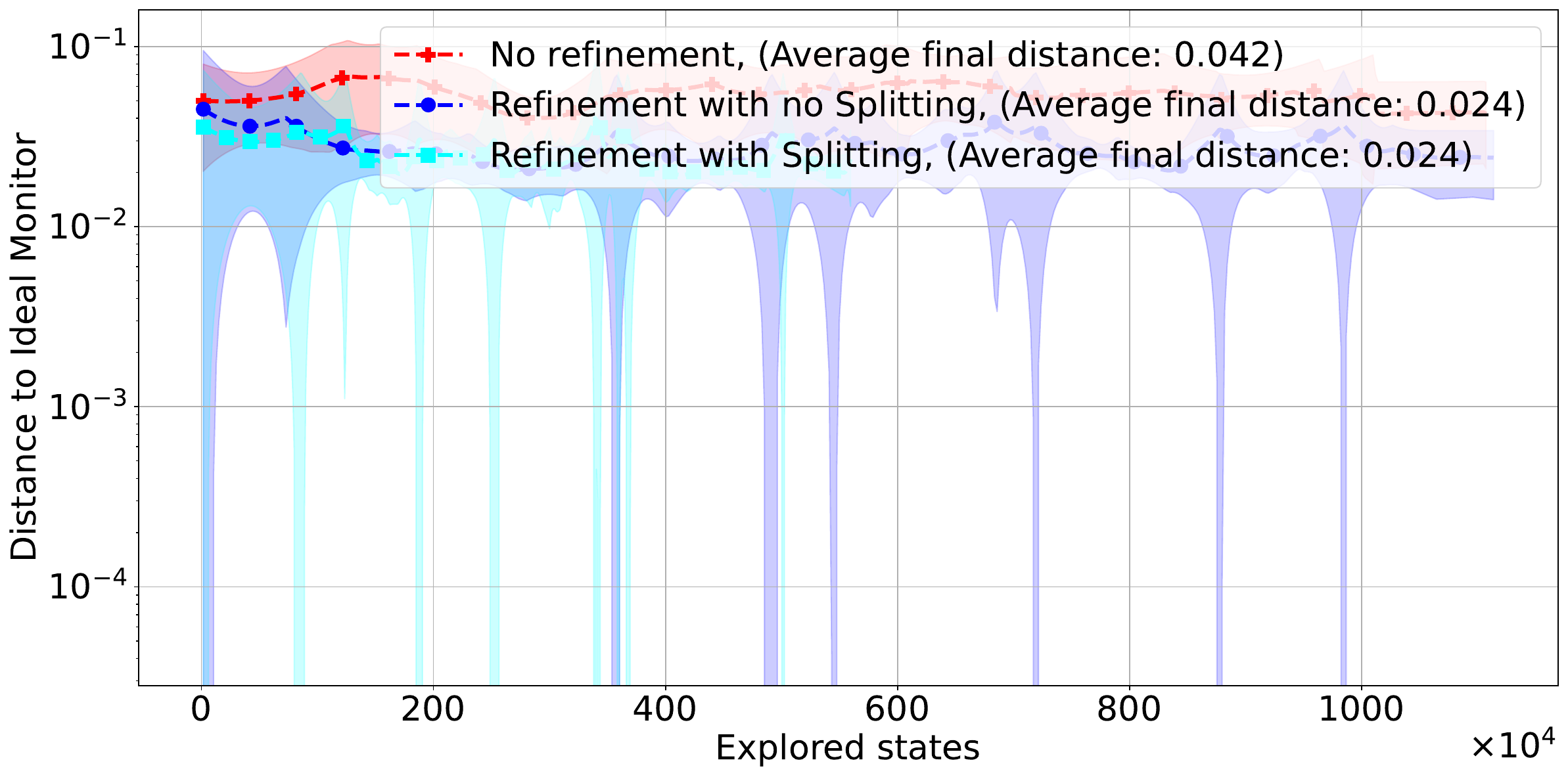}
    \caption{Coarse airportB-7-40-20, SC = 0.01 - Comparison between learning with and without refinement, in terms of distance to Ideal Monitor}
    \label{}
\end{figure}

\begin{figure}[H]
    \centering
    \includegraphics[width=0.85\linewidth]{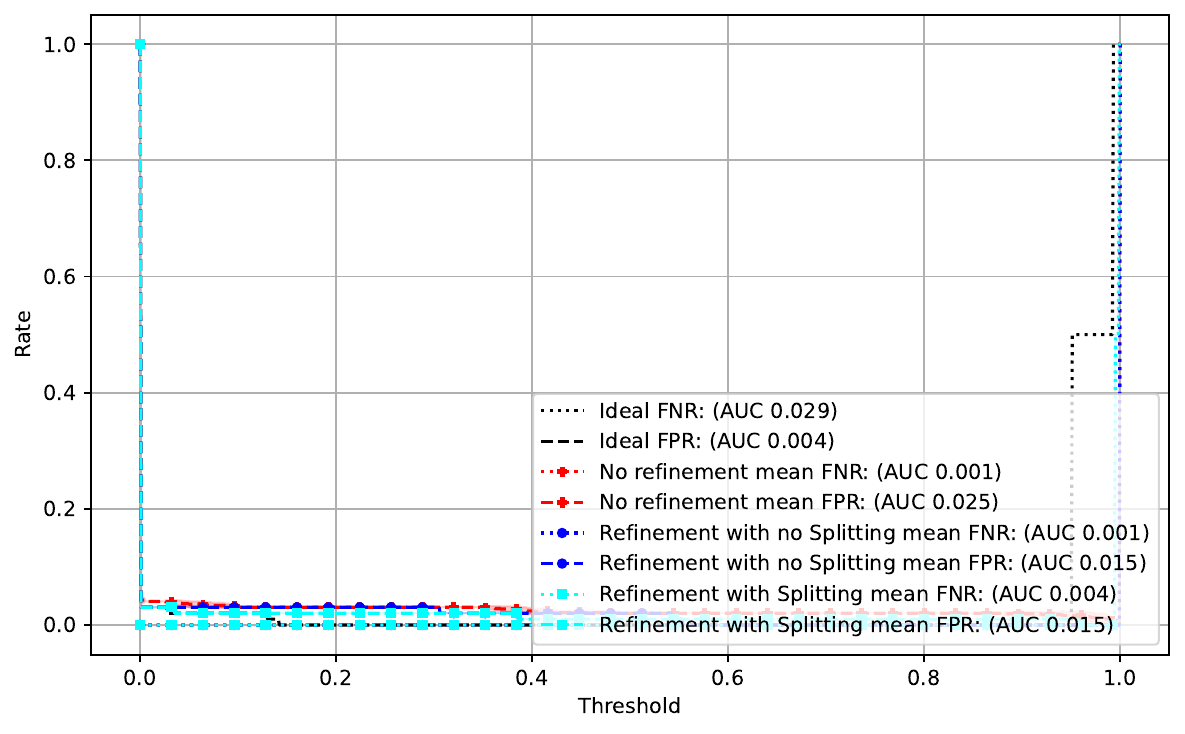}
    \caption{Coarse airportB-7-40-20, SC = 0.01 - Comparison between learning with and without refinement, in terms of FNR and FPR}
    \label{}
\end{figure}

\begin{figure}[H]
    \centering
    \includegraphics[width=0.85\linewidth]{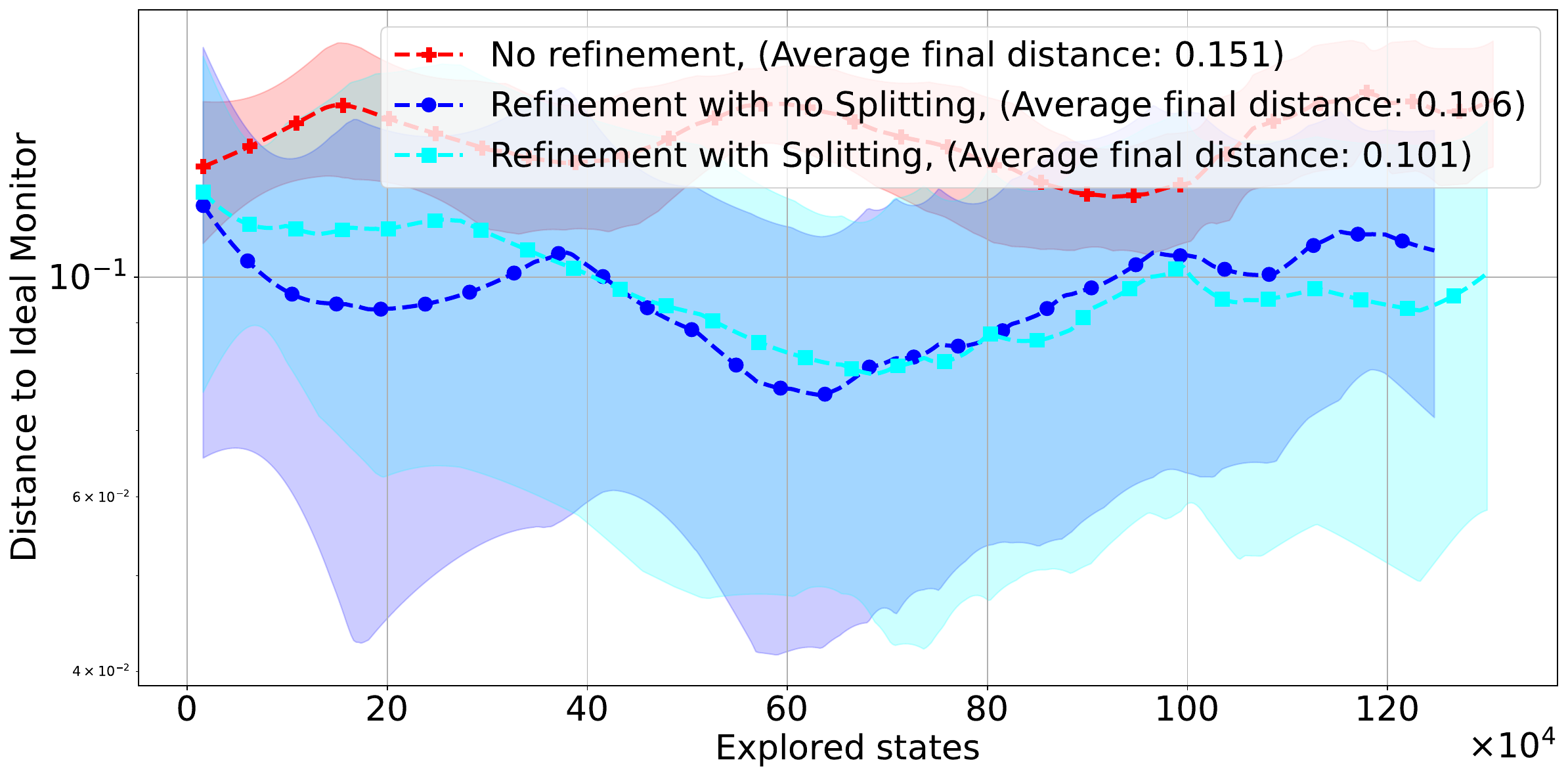}
    \caption{Coarse airportB-7-40-20, SC = 0.1 - Comparison between learning with and without refinement, in terms of distance to Ideal Monitor}
    \label{}
\end{figure}

\begin{figure}[H]
    \centering
    \includegraphics[width=0.85\linewidth]{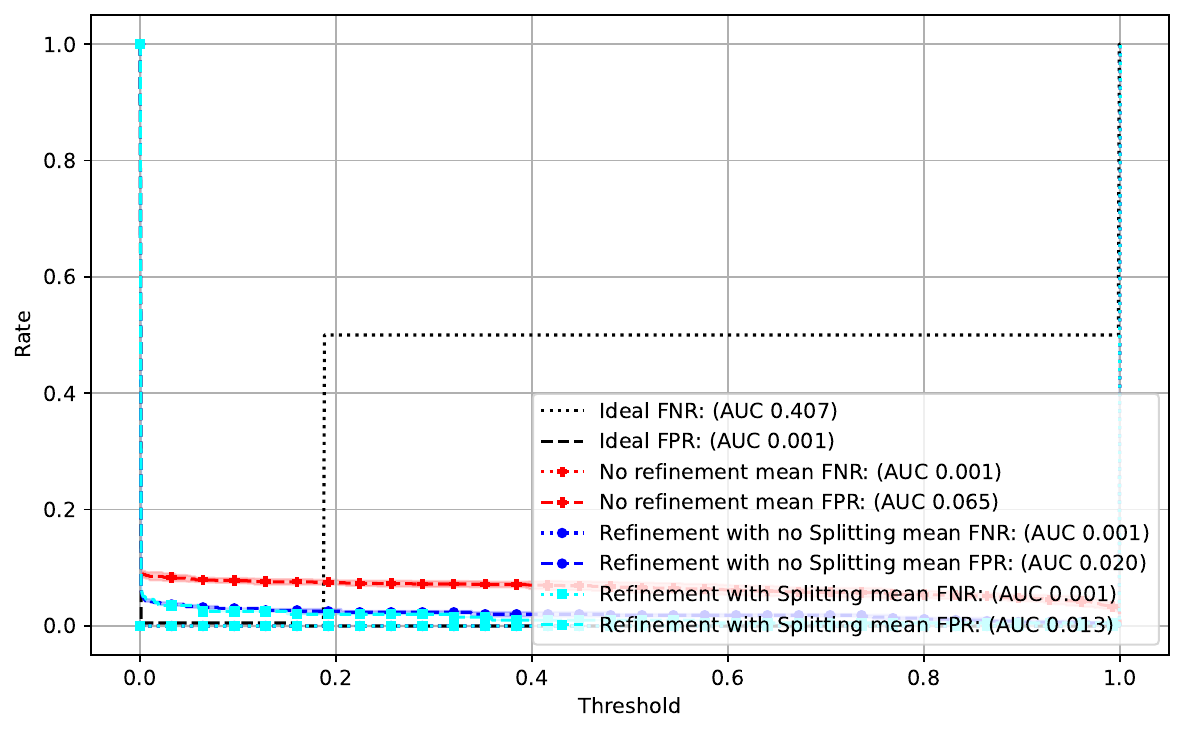}
    \caption{Coarse airportB-7-40-20, SC = 0.1 - Comparison between learning with and without refinement, in terms of FNR and FPR}
    \label{}
\end{figure}